\documentclass[11pt,reqno]{amsart}
\usepackage[margin=1in]{geometry}
\usepackage{amssymb}
\usepackage{amsfonts}
\usepackage{amsmath}
\usepackage{stmaryrd}
\usepackage{physics}
\usepackage{braket}
\usepackage{dsfont}
\usepackage{bbold}
\usepackage{verbatim}
\usepackage{fontawesome}
\usepackage{graphicx}
\usepackage{relsize}
\usepackage{makecell}
\usepackage[table,xcdraw]{xcolor}
\usepackage{enumerate}
\usepackage[pagebackref, colorlinks = true, linkcolor = blue, urlcolor  = blue, citecolor = red]{hyperref}
\usepackage{enumitem}
\usepackage{mathtools}
\usepackage{graphbox}
\usepackage{comment}
\usepackage{float}
\usepackage[percent]{overpic}
\usepackage{caption}
\usepackage{subcaption}
\usepackage{float}
\usepackage{lipsum}
\usepackage{tikz}
\usepackage{mathtools}
\usepackage[capitalize]{cleveref}
\usepackage[rightcaption]{sidecap}
\usepackage[most]{tcolorbox}

\newtheorem{theorem}{Theorem}[section]
\newtheorem{definition}[theorem]{Definition}
\newtheorem*{definition*}{Definition}
\newtheorem{proposition}[theorem]{Proposition}

\newtheorem{lemma}[theorem]{Lemma}

\newtheorem{conjecture}[theorem]{Conjecture}
\newtheorem*{conjecture*}{Conjecture}

\DeclareMathOperator{\vol}{vol}

\DeclareMathOperator{\perm}{perm}

\DeclareMathOperator{\tw}{tw}

\author{Ankit Kumar Jha}

\author{Ion Nechita}
\email{ion.nechita@univ-tlse3.fr}
\address{Laboratoire de Physique Th\'eorique, Universit\'e de Toulouse, CNRS, UPS, France}

\title[Random classical marginal problem \& quantum information theory]{On random classical marginal problems\\with applications to quantum information theory}

\begin{document}

\begin{abstract}
    In this paper, we study random instances of the classical marginal problem. We encode the problem in a graph, where the vertices have assigned fixed binary probability distributions, and edges have assigned random bivariate distributions having the incident vertex distributions as marginals. We provide estimates on the probability that a joint distribution on the graph exists, having the bivariate edge distributions as marginals. Our study is motivated by Fine's theorem in quantum mechanics. We study in great detail the graphs corresponding to CHSH and Bell-Wigner scenarios providing rations of volumes between the local and non-signaling polytopes. 
\end{abstract}

\maketitle

\tableofcontents

\section{Introduction}
Alongside quantum entanglement, quantum nonlocality is one of the main features of quantum theory that sets it apart from classical mechanics. It is the principle that some statistics observed in quantum experiments do not allow for a local realistic explanation. Arguably one of the most impactful developments in the foundations of quantum theory, nonlocality is often modeled in terms of \emph{Bell inequalities} \cite{bell1964einstein,clauser1969proposed}, mathematical relations that impose bounds on the correlations that can be explained by local theories; these inequalities have been experimentally violated, excluding local hidden variable models as possible theories explaining Nature \cite{aspect1982experimental,hensen2015loophole}. Violations of Bell inequalities show that quantum correlations that can be obtained in some experimental setting contain as a \emph{strict} subset classical correlations. Importantly, Popescu and Rohrlich introduce a superset of correlations, called \emph{non-signaling}, that obey the principle from special relativity that no faster-than-light communication is allowed. These non-signaling correlations \cite{popescu1994quantum} form a strict superset of quantum correlations. Understanding how the three sets of classical, quantum, and non-signaling correlations that can be obtained in a given setting is a central problem in the foundations of quantum theory \cite{brunner2014bell,scarani2019bell}. 

This work continues this line of investigation by analizing the containment of the set of local correlations inside the set of non-signaling correlations, in various scenarios encoded by graphs. We connect the problem of computing the ratio of the volumes of these two convex sets (which are polytopes) to two other mathematical problems and, using this connection, we provide exact computations in various specific and general scenarios. The volume of the classical ($\mathcal L$), quantum ($\mathcal Q$), and non-signaling ($\mathcal N$) sets of correlations in the setting of the CHSH game have been computed in \cite{cabello2005much}: 
$$\vol(\mathcal L) = \frac{2^5}{3} < \vol(\mathcal Q) = \frac{3 \pi^2}{2} < \vol(\mathcal N) = 2^4,$$
which leads to a ratio of 
$$\frac{\vol(\mathcal L)}{\vol(\mathcal N)} = \frac 2 3.$$
Further research on the relative volume of classical and quantum correlations have been performed with the help of the tensor norm formalism in \cite{gonzalez2017random,duarte2018concentration}. 

In this work, we shall focus not on the set of correlations itself, but on the set of conditional probabilities (sometimes called behaviours) that yield the correlations in a Bell scenario: $\mathbb P(a,b | x,y)$. The non-signaling condition, which is satisfied by both classical and quantum strategies, states that the marginal distribution with respect to one party should be independent of that party's question: 
\begin{align*}
    \sum_b \mathbb P(a,b | x,y) \quad &\text{does not depend on $y$}\\
    \sum_a \mathbb P(a,b | x,y) \quad &\text{does not depend on $x$}.
\end{align*}
In other words, one can define marginal distributions $\mathbb P(a|x)$ and $\mathbb P(b|y)$. We propose in this work an analysis of the set of local and non-signaling conditional probabilities having a \emph{fixed set of marginals}. We study this problem in a much more general setting than that of non-local games, by considering a bijection with combinatorial objects known as the \emph{correlation polytope} and its relaxation. The correlation polytope is motivated by 0-1 programming and combinatorial optimization, and it is defined via its extremal points. To an arbitrary graph $G=(V,E)$ we associate a polytope $\mathsf{COR}(G)$ in $\mathbb R^{|V|+|E|}$ defined by its extremal points $u = v \sqcup w$, where $v$ is an arbitrary bit string of length $|V|$, while $w$ is a bit string indexed by the edges $e=(x,y) \in E$ with $w_e = v_x \cdot v_y$. Pitowsky \cite{pitowsky1986range} realized that there is an intimate connection between correlation polytopes associated to some graphs and the conditional probabilities that corresponds to classical strategies various non-local games. Moreover, it turns out that these questions are also instances of \emph{classical marginal problems}, where one asks whether a given family of probability distributions is compatible, that is whether there exist a joint probability having the elements of the family as marginals. The connection between the local vs.~non-signaling polytope in quantum information theory and the classical marginal problem is one of the main \emph{conceptual contributions} of this work. These three equivalent formulations admit corresponding \emph{relaxations}, i.e. larger polytopes with a simpler structure. We summarize this situation in the first two columns of the following table. 

\medskip
\begin{center}
\bgroup
\def\arraystretch{1.5}
\begin{tabular}{|r|l|c|}
\hline
\rowcolor[HTML]{C0C0C0} 
\textbf{Polytope}       & \textbf{Relaxation}       & \textit{\textbf{Slice}}             \\ \hline
\rowcolor[HTML]{FFFFFF} 
Classical strategies    & Non-signaling strategies & \textit{fixed Alice/Bob marginals}  \\ \hline
\rowcolor[HTML]{FFFFFF} 
Correlation polytope    & LP-relaxation             & \textit{fixed vertex probabilities} \\ \hline
\rowcolor[HTML]{FFFFFF} 
Compatible probabilites & All probabilities         & \textit{fixed 1-site marginals}     \\ \hline
\end{tabular}
\egroup
\end{center}
\medskip

In this paper, we shall study the \emph{volume ratio} of \emph{slices} of these polytopes. In the non-local game strategy point of view, the slices that we consider correspond to having fixed marginals $\mathbb P(a|x)$ and $\mathbb P(b|y)$. The corresponding notion of slice for the other equivalent formulations are summarized in the third column of the table above. It is worth mentioning the formulation in terms of compatible probability distributions. We are given a graph $G=(V,E)$ and, and, for each vertex $v \in V$, some probability $p_v \in [0,1]$. The volume ratio discussed above corresponds to the following probability: 

\medskip
\noindent\emph{For each edge $e=(v,w) \in E$, sample uniformly a joint probability distribution $(X^{(e)}_v,X^{(e)}_w)$ such that $\mathbb P(X^{(e)}_v = 1) = p_v$ and $\mathbb P(X^{(e)}_w = 1) = p_w$. What is the probability that these pairwise distributions are compatible, i.e. that there exists a family $(Y_v)_{v \in V}$ such that 
$$\forall (v,w) \in E, \qquad (Y_v, Y_w) \overset{\mathrm{dist}}{=} (X^{(e)}_v,X^{(e)}_w) \, ?$$}
\medskip

Similar volume ratio computations have been performed in the optimization literature, where the local and the non-signaling polytopes are known (at least in the case of complete graphs) as the \emph{boolean quadric polytope} and its relaxation \cite{padberg1989boolean}. In particular, the ratio of volumes has been proposed as a measure of the quality of approximation of the local polytope by the (simpler) non-signaling polytope \cite{ko1997volume,lee2020volume}. In the quantum information theory literature, the relative volumes of the classical and the non-signaling polytopes have been computed, in the CHSH game scenario, in \cite{cabello2005much}. 

The \emph{technical contribution} of this work is twofold. On the one hand, we compute the volume ratio of slices the local and the non-signaling polytopes associated of very simple graphs: the triangle, the square, and cyclic graphs in general, for different value of the parameters defining the slices. Note that the triangle graph corresponds to the \emph{Bell-Wigner scenario} \cite{pitowsky1989quantum}, the square graph corresponds to the \emph{CHSH game} \cite{clauser1969proposed}, while the cyclic graph on 5 vertices corresponds to the \emph{KCBS scenario} \cite{klyachko2008simple}. Below are the plots corresponding to the triangle graph (with symmetric and non-symmetric slice parameters) and the square graph. The slice parameter $t$ corresponds to Alice's and Bob's marginals being fixed Bernoulli distributions with paramter $t$. We refer to \cref{prop:K3-vol,prop:K3-skew-vol,prop:K22-vol} for the exact formulas. 
\begin{center}
    \includegraphics[width=.31\textwidth]{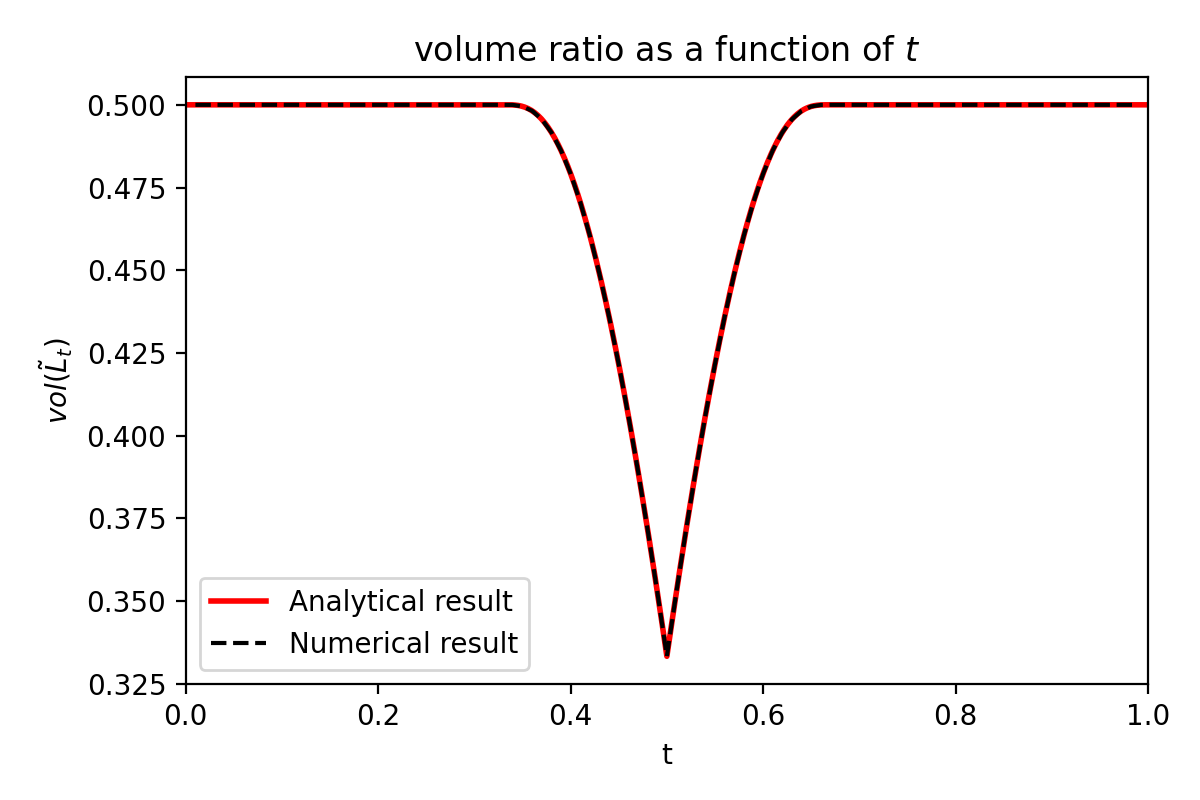}\quad
    \includegraphics[width=.31\textwidth]{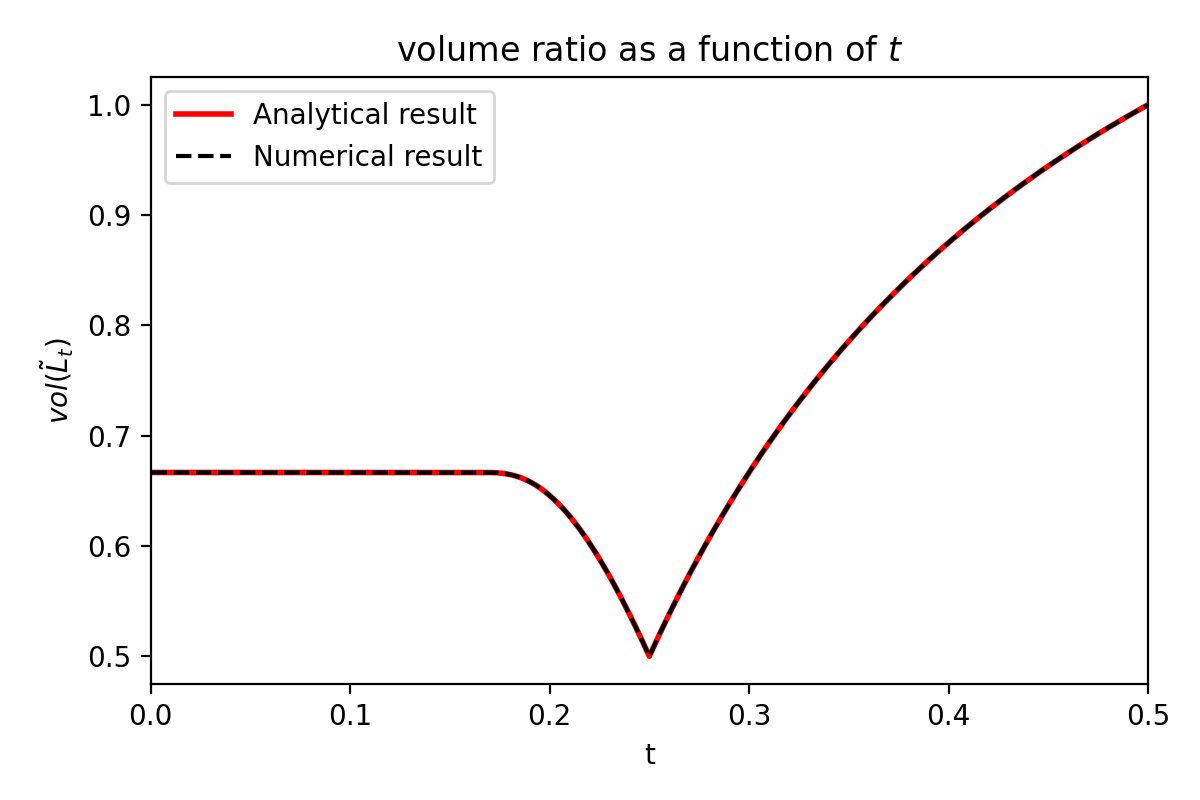}\quad
    \includegraphics[width=.31\textwidth]{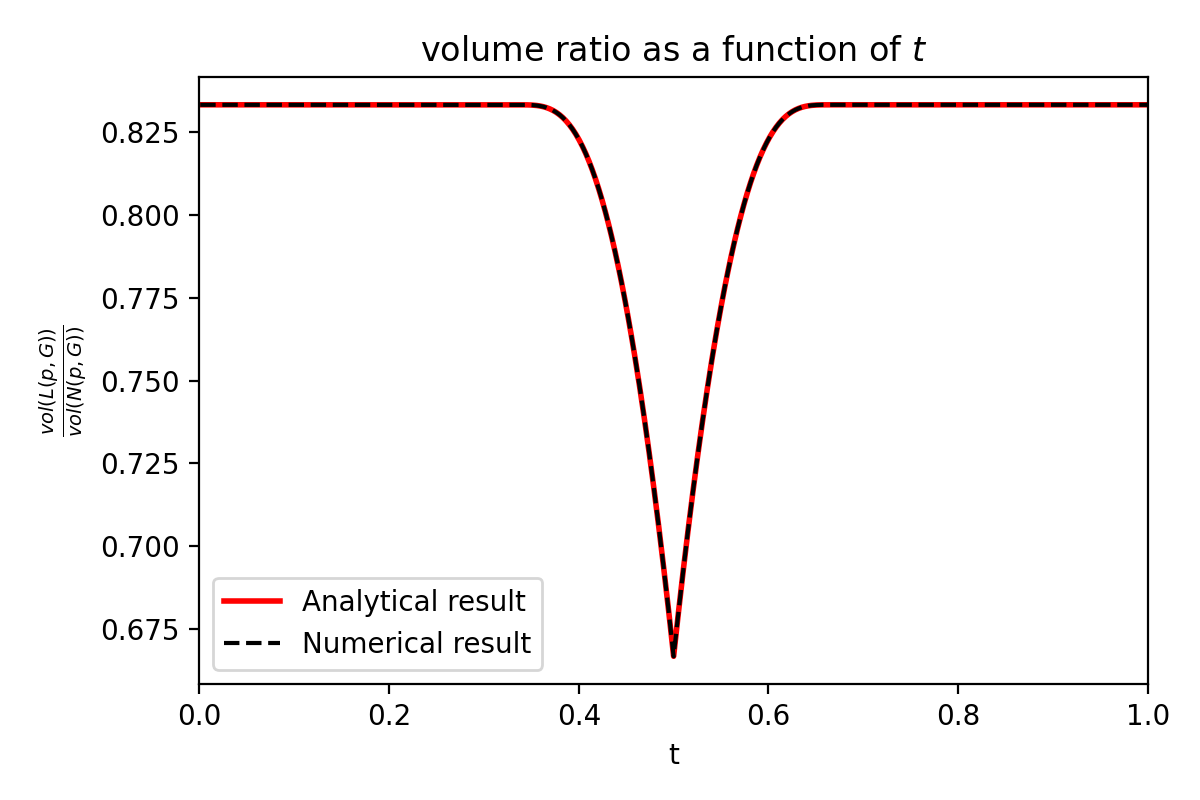}
\end{center}

The shape of the first and of the third graphs above (corresponding to slices having the same parameter) lead us to the second main technical contribution of this paper. One notices on this example that the volume ratio stays constant for values of the parameter $t$ that are close to 0 or close to 1. We define the \emph{fall-off} value $\tau(G)$ of the graph $G$ to be the largest value of $t$ for which the volume ratio is constant on the interval $[0,t]$. In the picture above, we have that the fall-off value of the triangle graph and that of the square graph in the symmetric cases (first and last graph) are equal to $\tau = 1/3$. This leads us to the second contribution of our work, which is more conceptual. 

We study the value of the fall-off parameter for general graphs, and we conjecture that its inverse is one plus the \emph{treewidth} of the graph. The conjecture is trivial for trees, and we prove it for graphs of treewidth two, i.e.~\emph{series-parallel} graphs. The conjecture is also supported by the values computed for simple graphs, such $K_4$ or $K_5$ (see \cref{appendix} for tables containing these values for graphs with 4 and 5 vertices). Moreover, we study how the volume ratio and the fall-off value behave under simple graph operations, using Fourier-Motzkin elimination. 

\bigskip

The paper is structured as follows. We start in \cref{sec:marginal-problem} by introducing the topic from the marginal problem perspective. The next two sections contain a presentation of the equivalent formulations: \cref{sec:slices-correlation} from the perspective of correlation polytopes, while \cref{sec:quantum-info} from the perspective of quantum information theory. In \cref{sec:random-marginal} we define precisely define the volume ratio as the probability that random bivariate distributions are compatible. Sections \cref{sec:triangle,sec:square,sec:cycle,sec:K4} contain the results about, respectively, the triangle graph (or the Bell-Wigner scenario), the square graph (or $K_{2,2}$, or the CHSH scenario), arbitrary cycle graphs, and the complete graph on four vertices $K_4$. In \cref{sec:FM} we gather several important results about general graphs, in particular our results about the relation between the fall-off value of the volume ratio and the treewidth of the graph. In the final \cref{sec:conclusion} we summarize our work and present some open problems and future research directions.

\section{The classical marginal problem}\label{sec:marginal-problem}

The classical marginal problem can be informally stated as follows: 

\medskip

\noindent\emph{When can a set of probability distributions 
$$\{p_J(x_{j_1}, \ldots, x_{j_{|J|}})\}_{J \in \mathcal J}$$
be extended to a joint probability distribution of all the variables $(x_i)$?
}

\medskip

This is well-studied question in probability theory and statistics which goes back at least to Hoeffding \cite{hoeffding1940masstabinvariante} and Fr\'echet \cite{frechet1951tableaux}. One can formalize it as follows. 

\begin{definition}\label{def:classical-marginal-general}
Let $G=(V,E)$ be a (finite) hypergraph, where each vertex $v \in V$ comes with a finite alphabet $\mathcal X_v$. Each hyperedge $E \ni e \subsetneq V$ comes with a probability distribution $p_e$ over the finite set
$$X_e:= \bigtimes_{v \in e} \mathcal X_v.$$

The \emph{classical marginal problem} associated with the hypergraph $G$ and the probability distributions $\{p_e\}_{e \in E}$ asks whether there exists a \emph{joint} probability distribution $p$ over the alphabet $\bigtimes_{v \in V} \mathcal X_v$ having all the $p_e$'s as marginals: 
$$\forall e \in E, \, \forall x \in \mathcal X_e \qquad p_e(x) = \sum_{y \in \mathcal X_{\bar e}} p(x,y),$$
where the variables in $y$ correspond to vertices that do not belong to the hyperedge $e$. If such a $p$ exist, we call the probabilities $(p_e)_{e \in E}$ $G$-compatible. 
\end{definition}

An obvious necessary condition for the marginal distributions to be $G$-compatible is that they should agree on their intersection: 
\begin{equation}\label{eq:marginal-compatibility}
    \forall e, f \in E, \, \forall x \in \mathcal X_{e \cap f}, \qquad \sum_{y \in \mathcal X_{e \setminus f}} p_e(x,y) = \sum_{z \in \mathcal X_{f \setminus e}} p_f(x,z).
\end{equation}

Importantly, these conditions are not sufficient in the general case, making the classical marginal problem an interesting one. In order to give an example to this point, let us introduce the specialization of the classical marginal problem that we will deal with in the current work. We shall consider the following special case, consisting of two main points: 
\begin{itemize}
    \item The individual random variables will be \emph{binary}, i.e.~$\mathcal X_v = \{0,1\}$, for all $v \in V$
    \item All the given marginals will be \emph{bivariate}, i.e.~$|e|=2$, for all $e \in E$. In other words, $G=(V,E)$ will be a (simple, undirected) graph. 
\end{itemize}

To illustrate the type of questions that we shall discuss in this work, consider the following triangle graph
\begin{center}
\begin{tikzpicture}[auto]

    \tikzstyle{vertex}=[circle, draw=blue, fill=blue!10!, ultra thick]
    \tikzstyle{edge}=[draw=black, thick]
    
    \node[vertex] (v1) at (-2,2) {1};
    \node[vertex] (v2) at (2,2) {2};
    \node[vertex] (v3) at (0, 0) {3};
    
    \draw[edge]  (v1) edge node{$p_{12}$} (v2);
    \draw[edge]  (v2) edge node{$p_{23}$} (v3);
    \draw[edge]  (v3) edge node{$p_{13}$} (v1);
    
\end{tikzpicture}
\end{center}
where the three edge bivariate probabilities are identical:
$$p_{12}(\cdot, \cdot) = p_{13}(\cdot, \cdot) = p_{23}(\cdot, \cdot) = 
\bgroup
\def\arraystretch{1.5}
\begin{tabular}{|c|c|}
    \hline 
    0 & 1/2 \\
    \hline 
    1/2 & 0 \\
    \hline
\end{tabular}
\egroup
$$
meaning that all three probabilities are given by
$$p_{ij}(a,b) = \begin{cases} 
1/2 \qquad &\text{ if } a \neq b\\
0 \qquad &\text{ if } a = b,
\end{cases}$$
for $a,b \in \{0,1\}$.
These edge marginals clearly satisfy the condition from Eq.~\eqref{eq:marginal-compatibility}, since their marginals are all symmetric Bernoulli distributions $\mathsf b(1/2)$:
$$p_i(a) = \frac 1 2 \quad a=1,2.$$
Note however that this particular classical marginal problem does not have a solution, that is the three marginals $p_{12}, p_{13}, p_{23}$ are not $K_3$-compatible. Indeed, if they where, there would exist binary random variables $X_{1,2,3}$ on a common probability space having the property
$$\mathbb P(X_1 \neq X_2) = \mathbb P(X_1 \neq X_3) = \mathbb P(X_2 \neq X_3) = 1,$$
which clearly is impossible since $X_i \in \{0,1\}$. This is a toy example of the phenomenon of \emph{frustration}, which has received a lot of attention in statistical physics. 

\section{Slices of the correlation polytope} \label{sec:slices-correlation}

As stated in the previous section, we shall focus in this work on the special case of the classical marginal problem where the given marginals are 2-partite. In this section we further specialize the problem, by fixing the single-variable marginals. This procedure can be very naturally formulated in terms of \emph{correlation polytopes}, and slices thereof. They were introduced by Pitowsky in \cite{pitowsky1986range} and have received a lot of attention in the recent years, see \cite{deza1994applicationsI, deza1994applicationsII} and the references therein. In this section, we follow the presentation from \cite{pitowsky1991correlation}, adapted to the setting of this work. We first introduce the correlation polytope associated to a graph. 

\begin{definition}\label{def:correlation-polytope}
    Given a graph $G=(V,E)$, we define its \emph{correlation polytope}
    $$\mathsf{COR}(G):=\operatorname{conv}\Big\{u_f \, ; \, f:V \to \{0,1\}\Big\} \subset \mathbb R^{|V|+|E|},$$
    where $u_f$ is the 0/1 vector having coordinates
    \begin{align*}
        \forall i \in V \qquad &u_f(i) = f(i) \\
        \forall e=(i,j) \in E \qquad &u_f(e) = f(i)f(j).
    \end{align*}
\end{definition}

As an example, for the triangle graph $K_3$ discussed previously, the $2^3=8$ vectors $u_f$ defined above are the rows of Table \ref{tbl:VE-K3}.
\begin{SCtable}[2][htb]
\bgroup
\def\arraystretch{1.5}
\begin{tabular}{|ccc|ccc|}
\hline
\rowcolor[HTML]{C0C0C0} 
\multicolumn{3}{|c|}{$V$-part}                         & \multicolumn{3}{c|}{$E$-part}                            \\ \hline
\rowcolor[HTML]{EFEFEF} 
1 & 2 & 3 & 12 & 13 & 23 \\ \hline
\hline
0 & 0 & 0 & 0 & 0 & 0 \\ \hline
0 & 0 & 1 & 0 & 0 & 0 \\ \hline
0 & 1 & 0 & 0 & 0 & 0 \\ \hline
0 & 1 & 1 & 0 & 0 & 1 \\ \hline
1 & 0 & 0 & 0 & 0 & 0 \\ \hline
1 & 0 & 1 & 0 & 1 & 0 \\ \hline
1 & 1 & 0 & 1 & 0 & 0 \\ \hline
1 & 1 & 1 & 1 & 1 & 1 \\ \hline
\end{tabular}
\egroup
\caption{The extremal points of the correlation polytope corresponding to the triangle graph $K_3$. The 8 extremal points have coordinates corresponding to the vertices of the graph (``the $V$-part'') and to the edges of the graph (``the $E$-part''). An entry $e_{ij}$ of the $E$-part is computed using the \textsf{AND} operation of the vertices $v_i$ and $v_j$.}
\label{tbl:VE-K3}
\end{SCtable}

The correlation polytope is related to basic probability theory by the following crucial result.

\begin{proposition}[{{\cite[Theorem 1.1]{pitowsky1991correlation}}}]\label{prop:Boole}
    A vector $p$ belongs to the correlation polytope $\mathsf{COR}(G)$ if and only if there exist probability events $(A_i)_{i \in V}$ on some common probability space such that
    \begin{align*}
        \forall i \in V, \qquad &p(i) = \mathbb P(A_i) \\
        \forall e=(i,j) \in E, \qquad &p(e) = \mathbb P(A_i \cap A_j).
    \end{align*}
\end{proposition}

The connection to the classical marginal problem from Definition \ref{def:classical-marginal-general} is now clear, by considering the distribution of the indicator random variables $X_v = \mathds 1_{A_v}$. 

The correlation polytope is introduced via its $V$-representation, by a convex hull construction, giving an overcomplete list of vertices as the vectors $u_f$ in Definition \ref{def:correlation-polytope}. We refer the reader to \cite{ziegler2012lectures} for the background on the convex geometry of polytopes. We recall here that a polytoepe has two mathematically equivalent representations: 
\begin{itemize}
	\item the $V$-representation, as a convex hull of its (finitely many) extremal points
	\item the $H$-representation, as an intersection of (finitely many) half-spaces defined by its facets. 
\end{itemize}
Obtaining a concise list of facet-defining inequalities (i.e.~the $H$-representation) for correlation polytopes is an active field of current research; partial results are known for small graphs of interest, while the general question of testing membership in $\mathsf{COR}(G)$ is NP-complete in general \cite[Section 3]{pitowsky1991correlation}. The complete set of inequalities (the $H$-representation) for the triangle graph $K_3$ consists of 16 inequalities and is given explicitly in \cref{eq:K3-N,eq:K3-L}.

We now come to the main object of study of this paper, particular slices of the correlation polytope, obtained by fixing the uni-variate probabilities $\{p_i\}_{i \in V}$. 

\begin{definition}\label{def:correlation-slice}
    Given a graph $G=(V,E)$ and a vector of probabilities $p \in [0,1]^{|V|}$, we define the \emph{correlation slice}
    $$\mathsf{COR}(p,G):=\{u \in \mathsf{COR}(G) \, : \, \forall i \in V, \, u(i)=p(i)\} \subset \mathbb R^{|E|}.$$
\end{definition}

Clearly, $\mathsf{COR}(p,G)$ is a non-empty polytope, since it contains the point
$$\Big(p(i)p(j)\Big)_{(i,j) \in E}$$
corresponding to independent events $A_i$ in Proposition \ref{prop:Boole}.

For every edge $e=(i,j) \in E$, the value $p_e \equiv p_{ij}$ completely determines the bivariate probability in the setting of the classical marginal problem: 

\medskip
\begin{center}
\bgroup
\def\arraystretch{1.5}
\begin{tabular}{lllcc}
                                 &                           &                        & $\sum=1-p_j$               & $\sum=p_j$                        \\
                                 &                           &                        & $\uparrow$                     & $\uparrow$                            \\
                                 &                           &                        & 0                      & 1                             \\ \cline{4-5} 
\multicolumn{1}{r}{$\sum = 1-p_i$} & \multicolumn{1}{c}{$\leftarrow$} & \multicolumn{1}{l|}{0} & \multicolumn{1}{c|}{$1-p_i-p_j+p_{ij}$} & \multicolumn{1}{c|}{$p_j-p_{ij}$}        \\ \cline{4-5} 
\multicolumn{1}{r}{$\sum=p_i$}     & \multicolumn{1}{c}{$\leftarrow$} & \multicolumn{1}{l|}{1} & \multicolumn{1}{c|}{$p_i-p_{ij}$} & \multicolumn{1}{c|}{$p_{ij}$} \\ \cline{4-5} 
\end{tabular}
\egroup
\end{center}

\medskip

Hence, in order for the vector $(p_{ij})_{(i,j) \in E}$ to be an element of the correlation slice, the following four elements need to be positive:
\begin{equation} \label{eq:trans-ineq}
    \begin{aligned}
    1-p_i-p_j+p_{ij} &\geq 0\\
    p_i-p_{ij} &\geq 0\\
    p_j-p_{ij} &\geq 0\\
    p_{ij} &\geq 0.
\end{aligned}
\end{equation}

These conditions state precisely the fact that the bivariate probability distribution above belongs to the transportation polytope defined by the Bernoulli distributions $\mathsf b(p_i)$ and $\mathsf b(p_j)$. We are thus led to introduce the following polytope. 

\begin{definition}\label{def:transportation-slice}
    Given a graph $G=(V,E)$ and a vector of probabilities $p \in [0,1]^{|V|}$, we define the \emph{transportation slice}
    $$\mathsf{TRA}(p,G):=\bigtimes_{e=(i,j) \in E} \Big[\max(0,p_i+p_j-1), \min(p_i,p_j) \Big].$$
\end{definition} 

The transportation slice is a cartesian product of slices of transportation polytopes, and encodes the trivial, uncoupled, inequalities that the bivariate probabilities $p_{ij}$ need to satisfy in order to belong to the correlation slice. For example, in our running example of the triangle graph $K_3$, these are the inequalities~\eqref{eq:K3-N} (at fixed $p_i, p_j$). We have the following obvious result. 

\begin{proposition}\label{prop:COR-pG-subset-TRA-pG}
    For all graphs $G=(V,E)$, and for all vectors $p \in [0,1]^{|V|}$,
    \begin{equation}\label{eq:COR-pG-subset-TRA-pG}
        \mathsf{COR}(p,G) \subseteq \mathsf{TRA}(p,G).
    \end{equation}
\end{proposition}

Actually, in order to mimic the construction of the correlation slice from Definition \ref{def:correlation-slice}, one can construct a transportation body associated to a graph $G$, such that the transportation slice can be obtained by fixing the $V$-coordinates of the transportation body. Note that we choose to call this polytope the transportation body in order to avoid any confusion with the established terminology of transportation polytope. 

\begin{definition}\label{def:transportation-body}
    Given a graph $G=(V,E)$, we define the \emph{transportation body}
    $$\mathsf{TRA}(G):=\{ (p,q) \in [0,1]^{|V|} \times [0,1]^{|E|} \, : \, q \in \mathsf{TRA}(p,G)\}.$$
\end{definition}
Similarly to Proposition \ref{prop:COR-pG-subset-TRA-pG}, we have the following inclusion:
$$\mathsf{COR}(G) \subseteq \mathsf{TRA}(G).$$

\bigskip

The goal of the rest of the paper is to study the inclusion in Eq.~\eqref{eq:COR-pG-subset-TRA-pG} and to quantify how close it is to being an equality. It is a well-know fact that for trees, we have an equality. 

\begin{proposition}\cite[Theorem V.2]{budroni2010extension}\label{prop:tree}
    If $G=(V,E)$ is a tree, then 
    $$\mathsf{COR}(G) = \mathsf{TRA}(G).$$
    Equivalently, for all vectors $p \in [0,1]^{|V|}$ we have 
    $$\mathsf{COR}(p,G) = \mathsf{TRA}(p,G).$$
\end{proposition}
\begin{proof}
For a tree $G$ with vertex degrees $d(v)$, one can show that the following joint probability distribution has the correct 2-marginals
$$p(x_1, x_2, \ldots, x_{|V|}) := \frac{\prod_{e=(v,w) \in E} p_e(x_v, x_w)}{\prod_{v \in V} p_v(x_v)^{d(v)-1}},$$
by recursively eliminating the leafs of the tree; see \cref{def:classical-marginal-general}.
\end{proof}

The reciprocal implication is also true. 
\begin{proposition}[{{\cite[Proposition 8]{padberg1989boolean}}}]\label{prop:forest}
    The correlation and the transportation bodies are equal $\mathsf{COR}(G) = \mathsf{TRA}(G)$ iff $G$ is a forest, i.e.~a collection of trees. 
\end{proposition}

\section{Relation to quantum information theory and contextuality}\label{sec:quantum-info}

In this section we would like to shed light on the connection between the transportation body and correlation polytope and the no-signaling and local polytopes obtained in non-local games in quantum information theory. This connection was the main motivation for our work, and the main objects we shall investigate in the rest of the paper are inspired by quantum information theory.

Let us first describe the setting of \emph{non-local games} from quantum information theory \cite{palazuelos2016survey}. These mathematical scenarios are modern formulations of (thought \cite{bell1964einstein})  experiments \cite{aspect1982experimental,hensen2015loophole} in foundational quantum mechanics related to non-locality. 
We consider below the case of the complete bipartite graph $G = K_{2,2}$ which is connected to the CHSH game \cite{clauser1969proposed} in quantum information theory. Let us first discuss how the the polytopes $\mathsf{COR}(K_{2,2})$ and $\mathsf{TRA}(K_{2,2})$ are related to the correlations $P(ab|xy)$ corresponding to the CHSH game. In that game, two players, Alice and Bob, receive binary questions $x,y \in \{0,1\}$ and provide binary answers $a,b \in \{0,1\}$. In the CHSH game, the players win if their answers satisfy $a \oplus b = x\cdot y$, but this rule is irrelevant to us; we focus here only on possible $4$-tuples $x,y,a,b$ and their probability distribution. Alice and Bob cannot communicate during the game, hence the conditional probabilities they generate must obey the \emph{non-signaling relations}: the marginal over Bob's answer $b$ of $P$ cannot depend on Bob's question $y$, and vice-versa
\begin{align}
 \label{eq:non-signaling-A}   \forall a,x, \qquad \sum_{b=0,1}P(a,b|x,0) &= \sum_{b=0,1}P(a,b|x,1) \\
\label{eq:non-signaling-B}    \forall b,y, \qquad \sum_{a=0,1}P(a,b|0,y) &= \sum_{a=0,1}P(a,b|1,y).
\end{align}
This means that the four bivariate probability distributions $P(\cdot, \cdot|x,y)$, for $x,y \in \{0,1\}$ have compatible marginals as described above, so one can attach them to the edges of the bipartite graph $K_{2,2}$, as below: 

\smallskip

\begin{figure}[htb]
\begin{tikzpicture}[auto]

    \tikzstyle{node_style}=[rectangle, draw=blue, fill=blue!10!, ultra thick]
    \tikzstyle{edge_style}=[draw=black, ultra thick]
    \tikzstyle{plus}=[]

    \node[node_style] (x0) at (-4, 2) {$x=0$};
    \node[node_style] (x1) at (-4, -2) {$x=1$};
    \node[node_style] (y1) at (4, -2) {$y=1$};
    \node[node_style] (y0) at (4, 2) {$y=0$};
    
    \draw[edge_style]  (x0) edge node[above]{$P(\cdot, \cdot| 0,0)$} (y0);
    \draw[edge_style]  (x0) edge node[above left]{} (y1);
    \draw[edge_style]  (x1) edge node[above right]{} (y0);
    \draw[edge_style]  (x1) edge node[below]{$P(\cdot, \cdot| 1,1)$} (y1);

    \coordinate [distance=0cm,label={$P(\cdot, \cdot| 0,1)$}] (O) at (3,-1.2);
    \coordinate [distance=0cm,label={$P(\cdot, \cdot| 1,0)$}] (O) at (3, 0.4);

\end{tikzpicture}
\end{figure}

\smallskip

Hence, to a conditional probability $P$ satisfying the non-signaling equations \eqref{eq:non-signaling-A}-\eqref{eq:non-signaling-B}, we associate the following element of the polytope $\mathsf{TRA}(K_{2,2})$:
\begin{align}\label{eq:vector-TRA}
    \Bigg[&\overbrace{P(a=1|0, \cdot), P(a=1|1, \cdot), P(b=1|\cdot, 0), P(b=1|\cdot, 1)}^{\text{vertices}}, \\
\nonumber    &\qquad\qquad\qquad\qquad \underbrace{P(1,1|0,0), P(1,1|0,1), P(1,1|1,0), P(1,1|1,1)}_{\text{edges}} \Bigg] \in \mathbb R^{4+4}.
\end{align}
Importantly, from the data above, using the non-signaling equations, one can recover the whole probability distribution $P$. Given now some fixed set of marginals 
\begin{align*}
    p^A_0 = P(a=1|x=0, y = \cdot) \\
    p^A_1 = P(a=1|x=1, y = \cdot) \\
    p^B_0 = P(b=1|x=\cdot, y = 0) \\
    p^B_1 = P(b=1|x=\cdot, y = 1),
\end{align*}
one can consider the slice $\mathsf{TRA}\big( (p^A_0,p^A_1,p^B_0,p^B_1), K_{2,2})$, consisting of 4-tuples 
$$\Big(P(1,1|0,0), P(1,1|0,1), P(1,1|1,0), P(1,1|1,1)\Big).$$ 
From this information, one can recover the whole set of joint game probabilities $P$ having fixed vertex-marginals $(p^A_0,p^A_1,p^B_0,p^B_1)$. This type of slice will be the main focus of our work. 

Having related the (slices of the) transportation polytope $\mathsf{TRA}(K_{2,2})$ to conditional probabilities appearing in the CHSH game and satisfying the non-signaling conditions (i.e.~$\mathsf N(K_{2,2})$), let us now describe the correlation polytope $\mathsf{COR}(K_{2,2})$ in terms of the local probabilities $P$ appearing in the CHSH game. To this end, recall that a game strategy $P$ is called \emph{local} if there exists a \emph{hidden variable} $\lambda$ with probability distribution $Q$, and local probabilities $P_A$, $P_B$ such that
$$P(a,b|x,y) = \sum_{\lambda} Q(\lambda) P_A(a|x, \lambda) P_B(b|y, \lambda).$$
The convex set of local probabilities $P$ is the convex hull of its extremal points, the set of \emph{deterministic} strategies, where $Q$ is a delta function, and, for each pair $(\text{player}, \text{question})$ (i.e.~for each vertex of the graph $K_{2,2}$), there is a deterministic answer, $0$ or $1$. This choice is encoded precisely in the function $f$ from \cref{def:correlation-polytope}, while the corresponding vector \cref{eq:vector-TRA} is given by $u_f$. This shows how the correlation polytope $\mathsf{COR}(K_{2,2})$ is related to the local polytope $\mathsf L(K_{2,2})$. One can rephrase this using \emph{Fine's theorem} \cite{fine1982hidden}: a conditional probability $P$ is local iff it can be written as a convex mixture of local deterministic processes:
\begin{equation}
    P(ab|xy) = \sum_{j, k}r_{j,k} \mathds 1_{a=j(x)} \mathds 1_{b=k(y)}
\end{equation}
Finally, let us point out that non-local games where Alice and Bob receive $m$, respectively $n$ questions, and must provide binary answers, can be easily reformulated in terms of the transportation body and correlation polytope for the \emph{bipartite complete graph} $K_{m,n}$. 

Let us now put forward the connection between the transportation body and correlation polytope \emph{contextuality theory}. In the Bell scenario, we are only concerned with correlations of outcomes of measurements which are spatially separated, hence the notion of locality and no-signaling. Contextuality is the generalization of Bell scenarios to include correlations among all compatible observables. The set of these compatible observables forms a \emph{context}. Hence, for a set of observables 
 $X = \{X_0, X_1, \ldots , X_n\}$, a subset $\mathsf c \subseteq X$ forms a context if $\forall i, j : X_i, X_j \in \mathsf c$, $X_i$ is compatible with $X_j$.

Following \cite{chaves2012entropic} (we refer the reader to \cite{scarani2019bell} for a detailed description of Bell scenarios), we start with a hidden-variable $\lambda$ which completely defines the process of obtaining outcomes corresponding to any observable $X_i$. Thus, the probabilities $\rho(\lambda)$ associated with different processes must follow $\rho(\lambda) \geq 0$  and $\sum_\lambda \rho(\lambda) = 1$. Completeness implies that the distribution $P(x_i|X_i\lambda)$ is independent of all other observables.

Clearly, for two compatible observables $X_i$ and $X_j$, we have,
\begin{equation}
    P(x_ix_j|X_iX_j) = \sum_\lambda \rho(\lambda)P(x_i|X_i)P(x_j|X_j)
\end{equation}

This is equivalent to saying that the observables are non-contextual (compatible) iff it can be written as a convex mixture of deterministic processes. \cite{fine1982hidden}
\begin{equation} \label{eq:Fine-1}
    P(x_ix_j|X_iX_j) = \sum_{k, l}r_{k,l} \mathds 1_{x_i=k(X_i)} \mathds 1_{x_j=l(X_j)}
\end{equation}
where, $r_{j, k}\geq 0$ and $\sum_{j, k} r_{j, k} = 1$. This is true in general for any number of compatible observables. We will only be studying cases with cardinality of contexts less than or equal to 2. 

Notice that each row in the truth table for the non-contextual polytope denotes one of these deterministic processes and the convex sum of the rows yields any point in the correlation polytope. Hence, the correlation polytope and the non-contextual polytope have the same mathematical structure.

Now, consider dichotomic outputs $\{0, 1\}$ for all observables and define \(P(1|X_1)\) as \(p_{a_1}\) and \(P(1|X_2)\) as \(p_{b_1}\) where $X_1$ and $X_2$ form a context. Clearly, \(p_{a_0} = 1-p_{a_1}\) and \(p_{b_0} = 1-p_{b_1}\). Notice that these probabilities exactly mimic the behaviour described by inequalities \cref{eq:trans-ineq} and hence, the no-disturbance body has the same structure as that of the transportation body. This establishes the connection between the objects studied in this paper and contextuality. Finally, let us point out that there exists another theoretical framework for contextuality, based on (hyper-)graphs \cite{cabello2014graph,acin2015combinatorial}, where vertices correspond to outcomes and hyperedges to measurements. It is argued in \cite[Appendix D]{acin2015combinatorial} that the observable based approach and the hypergraph based approach are equivalent. Note that in both the setting of the current paper and in \cite{acin2015combinatorial} (the probabilistic model), to each vertex of a graph one associates a number $p(v) \in [0,1]$. The meaning of this assignment is completely different: in our setting, there are no other vertices $v_2, \ldots, v_k$ such that the sum $p(v) + p(v_2) + \cdots + p(v_k) = 1$, whereas in \cite[Definition 2.4.1]{acin2015combinatorial} these vertices appear explicitly in the graph. 

From now on, for the sake of brevity, and in order to emphasize further the connection with quantum information theory, we shall use the notation $\mathsf L$ (resp.~$\mathsf N$) to denote the correlation polytope $\mathsf{COR}$ (resp.~the transportation polytope $\mathsf{TRA}$) and their slices:
    \begin{align*}
        \mathsf L(G) &:= \mathsf{COR}(G)\\
        \mathsf N(G) &:= \mathsf{TRA}(G)\\
        \mathsf L(p,G) &:= \mathsf{COR}(p,G)\\
        \mathsf N(p,G) &:= \mathsf{TRA}(p,G),
    \end{align*}
where $G=(V,E)$ is any graph and $p$ is a vector of proabilities $p \in [0,1]^{|V|}$.

To conclude, we have shown in this section that the marginal problem we introduced in \cref{sec:marginal-problem}, in the case of the square graph $C_4 = K_{2,2}$, is intimately related to the CHSH non-local game, with the constraint of fixed marginals for Alice and Bob. This situation can be naturally generalized to all bipartite complete graphs $K_{m,n}$, which correspond to non-local games with two answers per player and $m$, respectively $n$, questions. The volume ratio question for the $K_{2,2}$ graph will be discussed at length in \cref{sec:square}. Other graphs, such as the triangle graph, cyclic graphs, or the $K_4$ graph, are discussed in \cref{sec:triangle,sec:cycle,sec:K4}.

\section{Random marginal problems} \label{sec:random-marginal}

Motivated by questions in quantum information theory, we have introduced in the previous sections, for a given graph $G=(V,E)$ and a given vector of probabilities $p \in [0,1]^{|V|}$, the \emph{local slice} (the {correlation slice}), resp.~the \emph{non-signaling slice} (the {transportation slice})
$$\mathsf{L}(p,G) \subseteq \mathsf{N}(p,G) \subset \mathbb R^{|E|}.$$

Given its cartesian product structure, the set $\mathsf N(p,G)$ comes equipped with a natural probability measure that is easy to compute with, the normalized Lebesgue measure:
\begin{equation}
    \eta(p,G) := \bigotimes_{(i,j) \in E} \frac{\mathds 1_{[\max(0,p_i+p_j-1), \min(p_i,p_j)]}}{\min(p_i,p_j) - \max(0,p_i+p_j-1)}\,\mathrm{d}x
\end{equation}

The probability measure $\eta$ defined above is precisely the uniform (normalized) volume measure on the (scaled) hypercube introduced in Definition \ref{def:transportation-slice}.

In what follows, we shall provide a partial answer to the following fundamental question: 

$$\mathbb P_{q \sim \eta(p,G)} \Big[ q= (q_{ij})_{(i,j) \in E} \in \mathsf L(p,G) \Big] = \, ?$$

{
Equivalently, this quantity can be also understood as the volume ratio 
$$\frac{\vol(\mathsf L(p, G))}{\vol(\mathsf N(p, G))}$$
between corresponding slices of the local and the non-signaling polytopes. As mentioned in the introduction, a similar question for the non-sliced bodies has been studied in the context of (approximating) the boolean quadric polytope \cite{padberg1989boolean,ko1997volume,lee2020volume}. 
}

To quantify the properties of the volume ratio for sliced bodies, we define the following parameters:

\begin{definition}\label{def:fall-off}
    For a given graph $G(V, E)$ and symmetric marginals $p_i = t$ $\forall i \in V$, we define the \emph{fall-off value} $\tau(G)$ as the value of $t$ after which the volume ratio $\vol(\mathsf L(G)_t)/\vol(\mathsf N(G)_t)$ is no longer constant:
    
    $$\tau(G):= \sup\bigg\{ t\, : \, \frac{\vol(\mathsf L(G)_t)}{\vol(\mathsf N(G)_t)} \text{ is constant on }(0,t) \bigg\} \in [0,1/2].$$
    Above, we write 
    $$\mathsf{S}(G)_t:= \mathsf S(p=(t,t,\ldots, t), G) \qquad \text{ for a set } \mathsf S = \mathsf L, \mathsf N.$$
\end{definition}

For trees $T$, we have seen in \cref{prop:tree} that $\tau(T) = 1/2$

\begin{proposition}
    For any graph $G$, $\tau(G) > 0$. In other words, the volume ratio $\vol(\mathsf L(G)_t)/\vol(\mathsf N(G)_t)$ is constant on some non-empty interval $t \in (0, \tau)$. 
\end{proposition}
\begin{proof}
    This follows since we have finitely many inequalities that can cut at some \(t>0\), so the polytopes $\mathsf L(G)_t$ and $\mathsf N(G)_t$ are equal up to some time $\tau$.
\end{proof}

We introduce next two important values of the volume ratio: the one of the constant portion $t \in (0,\tau)$ and the one in the middle, $t=1/2$. 
\begin{definition}
    For a given graph $G(V, E)$ and associated marginals $p_i = t$ $\forall i \in V$, we define the \emph{
    initial ratio} 
    $\rho_{0+}(G)$, as the volume ratio for all values $t \in (0, \tau(G))$; in particular
    $$\rho_{0+}(G):= \frac{\vol(\mathsf L(G)_{\tau/2})}{\vol(\mathsf N(G)_{\tau/2})}.$$
     We also define 
     \emph{middle ratio} $\rho_{1/2}(G)$ as the volume ratio at $t=1/2$:
     $$\rho_{1/2}(G):= \frac{\vol(\mathsf L(G)_{1/2})}{\vol(\mathsf N(G)_{1/2})}.$$
\end{definition}

\section{The triangle graph}\label{sec:triangle}

Having set the stage, it is instructional that we consider specific examples. We start by looking at the complete graph with three vertices, $K_3$ This is also the smallest non-trivial graph which is not a tree (recall that for trees, $\mathsf L = \mathsf N$, see \cref{prop:tree}). 

The $K_3$ graph corresponds to the Bell-Wigner Polytope \cite{pitowsky1989classical}, \cite{pitowsky1989george}. The physical scenario in question is measurements performed on a bipartite state (such as the singlet state) with each component of the state being measured in one of three distinct directions (3 questions) with binary outputs (2 answers for each question) \cite{wigner1997hidden}.

\smallskip

\begin{figure}[htb]
\begin{tikzpicture}[auto, scale = 1.2]

    \tikzstyle{vertex}=[circle, draw=blue, fill=blue!10!, ultra thick]
    \tikzstyle{edge}=[draw=black, ultra thick]
    
    \node[vertex] (v1) at (-2,2) {$p_1$};
    \node[vertex] (v2) at (2,2) {$p_2$};
    \node[vertex] (v3) at (0, 0) {$p_3$};
    
    \draw[edge]  (v1) edge node{$q_{12}$} (v2);
    \draw[edge]  (v2) edge node{$q_{23}$} (v3);
    \draw[edge]  (v3) edge node{$q_{13}$} (v1);
    
\end{tikzpicture}
\end{figure}

\smallskip

The $V$-representation of $\mathsf L(K_3) = \mathsf{COR}(K_3)$ for this graph has already been given in \cref{tbl:VE-K3}. The corresponding $H$-representation reads:

\begin{equation}\label{eq:K3-N}
\begin{aligned}
  0 \leq q_{ij} \leq \min{(p_{i}, p_{j})} \\
  p_i + p_j -q_{ij} \leq 1
\end{aligned}
\end{equation}
\vspace{5px}
\begin{equation}\label{eq:K3-L}
    \begin{aligned}
        p_{1} + p_{2} + p_{3} - q_{12} - q_{13} - q_{23} &\leq 1 \\
        -p_{1} + q_{12} + q_{13} - q_{23} &\leq 0\\
         -p_{2} + q_{12} - q_{13} + q_{23} &\leq 0\\
         -p_{3}- q_{12} + q_{13} + q_{23} &\leq 0
     \end{aligned}
\end{equation}

Note that we have in this case 6 variables: 3 corresponding to the vertices of $K_3$ ($p_1,p_2,p_3$) and 3 corresponding to the edges ($q_{12},q_{13},q_{23}$). Hence, polytopes of interest have 6 coordinates: 
$$\mathsf L(K_3) = \mathsf{COR}(K_3) \subseteq \mathsf N(K_3) = \mathsf{TRA}(K_3) \subset \mathbb R^6.$$

Inequalities in \cref{eq:K3-N} are precisely those appearing in  \cref{def:correlation-slice} and thus define the set $\mathsf N(K_3)$. The inequalities \cref{eq:K3-L} are the additional constraints distinguishing $\mathsf L(K_3)$ from $\mathsf N(K_3)$. The volumes of these bodies have been computed respectively in \cite[Theorems 16 and 19]{lee2020volume}:
$$\vol(\mathsf N(K_3)) = \frac{1}{120}\qquad \text{and} \qquad \vol(\mathsf L(K_3)) = \frac{1}{180} = \frac 2 3 \vol(\mathsf N(K_3)).$$

As stated in the introduction, the main philosophy of our work is to understand the inclusion $\mathsf L(K_3) \subseteq \mathsf N(K_3)$ via \emph{slices} of these polytopes, and to see the slice inclusion problem as a a classical marginal problem. 

In this section we shall consider two such slices, studied separately in the following two subsections. The slices are obtained by fixing the values of the vertex parameters $p$ and studying the 3-dimensional polytopes of the $q$ variables.

\subsection{Symmetric slices \texorpdfstring{$(p_1, p_2, p_3) = (t, t, t)$}{ttt}}

We first consider the most symmetric case, where all the vertex variables $p_{1,2,3}$ of the correlation polytope of the triangle graph are equal:
$$p_1 = p_2 = p_3 =: t \in [0,1].$$

Plugging in these marginals, we obtain the $H$-representation defining our slice $\mathsf L(p=(t,t,t), K_3)$:
\begin{equation}\label{eq:K3-Np}
\begin{aligned}
  \max\{0, 2t-1\} \leq q_{ij} \leq t
\end{aligned}
\end{equation}

\begin{equation}\label{eq:K3-Lp}
    \begin{aligned}
        q_{12} + q_{13} + q_{23} &\geq 3t-1 \\
        q_{12} + q_{13} - q_{23} &\leq t \\
        q_{12} - q_{13} + q_{23} &\leq t\\
        -q_{12} + q_{13} + q_{23} &\leq t.
     \end{aligned}
\end{equation}
As before, \cref{eq:K3-Np} corresponds to the no-signaling slice $\mathsf N(p = (t,t,t), K_3)$. In what follows, we shall denote the polytopes of interest by
\begin{align*}
    \mathsf N_t &:= \mathsf N(p=(t,t,t), K_3) = \{(q_{12},q_{13},q_{23}) \in \mathbb R^3 \, : \, \text{ \cref{eq:K3-Np} holds}\}\\
    \mathsf L_t &:= \mathsf L(p=(t,t,t), K_3) = \{(q_{12},q_{13},q_{23}) \in \mathbb R^3 \, : \, \text{ \cref{eq:K3-Np} and \cref{eq:K3-Lp} hold}\}.
\end{align*}

Since our measurements are dichotomic (with outcomes say, $0$ and $1$), the assignment of possible outcomes to the random variables is symmetric with respect to swapping the two possible outcomes (bit-flip). Hence, we have the following obvious symmetry. 

\begin{proposition} \label{prop:symVol}
    The involution $(q_{12},q_{13},q_{23}) \mapsto (1-q_{12},1-q_{13},1-q_{23})$ maps isometrically $\mathsf N_t \leftrightarrow \mathsf N_{1-t}$ and $\mathsf L_t \leftrightarrow \mathsf L_{1-t}$ for all $t \in [0,1]$.
\end{proposition}

Thus, we will keep our study limited to $t \in [0,1/2]$. This ensures that \cref{eq:K3-Np} simplifies to \(0 \leq q_{ij} \leq t\). Note that the case $t=0$ (or, equivalently, $t=1$), corresponding to the deterministic scenario, is degenerate: 
$$\mathsf N_0 = \mathsf L_0 = \{(0,0,0)\}.$$
We shall thus assume from now on $t \in (0,1/2]$.

Before we explicitly calculate the volumes of $\mathsf L_t$ and $\mathsf N_t$, we will do one final simplification. Substituting \(q_{ij} \to tx_{ij}\), the inequalities \eqref{eq:K3-Np} and \eqref{eq:K3-Lp} become:

\begin{equation} \label{eq:K3-Npx}
    0 \leq x_{ij} \leq 1
\end{equation}

\begin{subequations} \label{eq:K3-Lpx}
\begin{align}
    \label{eq:K3-Lpx-1} x_{12} + x_{13} + x_{23} &\geq 3-\frac{1}{t} \\
    \label{eq:K3-Lpx-2} x_{12} + x_{13} - x_{23} &\leq 1 \\
    \label{eq:K3-Lpx-3} x_{12} - x_{13} + x_{23} &\leq 1 \\
    \label{eq:K3-Lpx-4} -x_{12} + x_{13} + x_{23} &\leq 1
\end{align}
\end{subequations}

We introduce the new, scaled, polytopes: 
\begin{align*}
    \tilde{\mathsf N}_t &:= \{(x_{12},x_{13},x_{23}) \in \mathbb R^3 \, : \, \text{ \cref{eq:K3-Npx} holds}\}\\
    \tilde{\mathsf L}_t &:= \{(x_{12},x_{13},x_{23}) \in \mathbb R^3 \, : \, \text{ \cref{eq:K3-Npx} and \cref{eq:K3-Lpx} hold}\}.
\end{align*}
Clearly, one has $\vol(\mathsf N_t) = t^3 \vol(\tilde{\mathsf N}_t)$, and similarly for the local polytope; hence 
$$\frac{\vol(\mathsf L_t)}{\vol(\mathsf N_t)} = \frac{\vol(\tilde{\mathsf L}_t)}{\vol(\tilde{\mathsf N}_t)}.$$

From inequality \eqref{eq:K3-Npx}, it becomes easy to see that $\tilde{\mathsf N}_t$ is the unit cube $[0,1]^3$ and thus $\vol(\tilde{\mathsf N}_t) =1$, independently of the value of $t$.

The computation for the scaled local polytope slice $\tilde{\mathsf L}_t$ is more involved. First, we note that \cref{eq:K3-Lpx-1} depends on $t$, while \crefrange{eq:K3-Lpx-2}{eq:K3-Lpx-4} are independent of $t$. Importantly, \cref{eq:K3-Lpx-1} is trivially true for all $t \in (0, 1/3]$. Hence, the volume of the local slice is constant in this region. 

We assume for now $t \leq 1/3$. In this regime, to explicitly compute the volume, we go to the $V$-representation of $\tilde{\mathsf L}_t$. The polytope $\tilde{\mathsf L}_t$ has the following vertices 
$$(0, 0, 0), (1, 0, 0), (0, 1, 0), (0, 0, 1), (1, 1, 1).$$

The corresponding 3-dimensional body is depicted in Figure \ref{fig:K3-t-small}. We partition it into two disjoint bodies:
\begin{itemize}
    \item a triangular pyramid 
    $$\operatorname{conv} \Big\{  (0, 0, 0), (1, 0, 0), (0, 1, 0), (0, 0, 1)\Big\}$$
    which has volume $1/6$
    \item a regular tetrahedron
    $$\operatorname{conv} \Big\{  (1, 1, 1), (1, 0, 0), (0, 1, 0), (0, 0, 1)\Big\}$$
    with side length $\sqrt 2$ and volume $1/3$.
\end{itemize}

\begin{figure}[htb]
\hspace{-2cm}\begin{minipage}[l]{.45\textwidth}
    \includegraphics[scale=0.8]{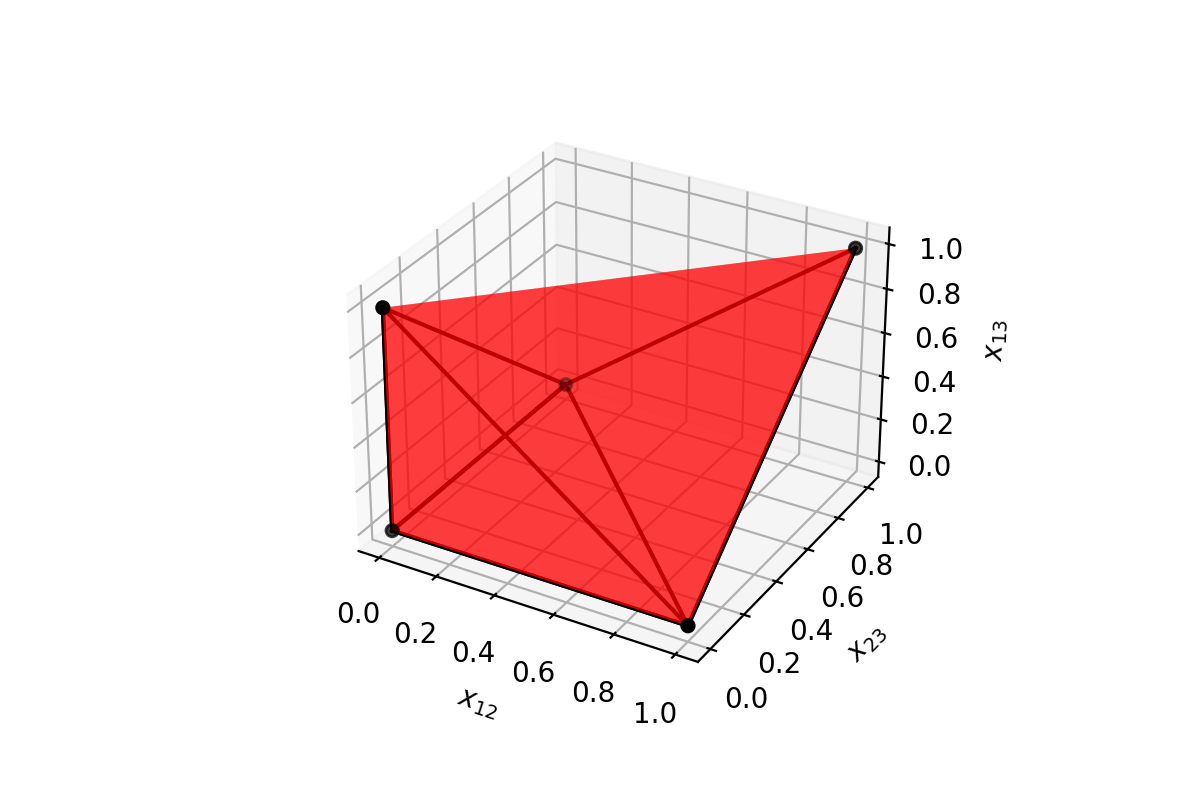}
\end{minipage}\hfill
\begin{minipage}[r]{.45\textwidth}
    \includegraphics[scale=0.5]{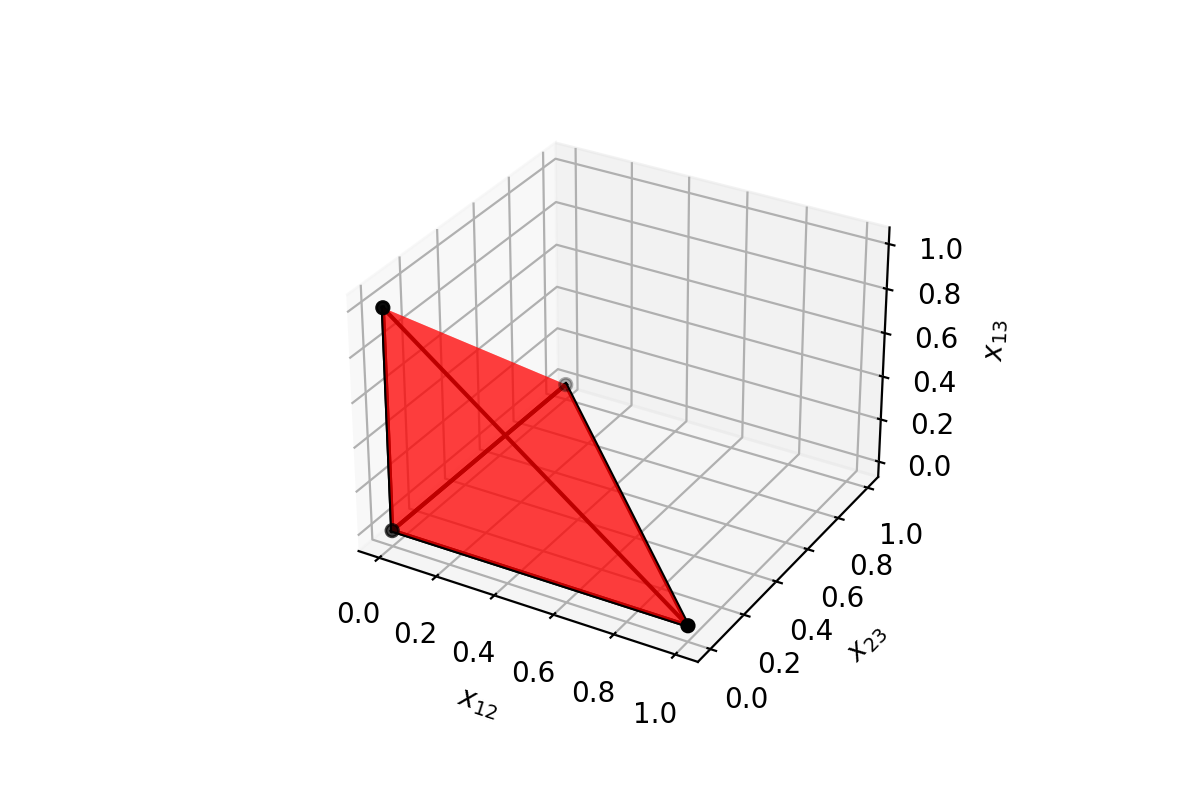} \\
    \includegraphics[scale=0.5]{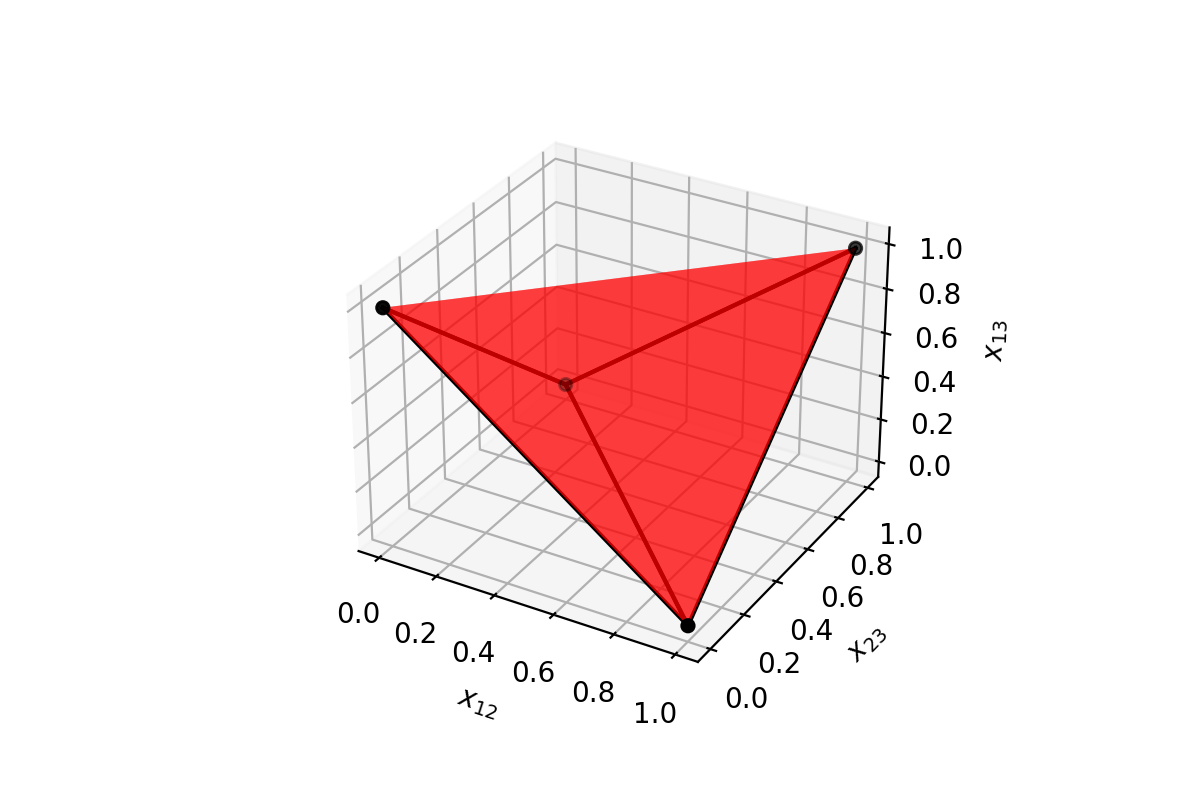}
\end{minipage}

\caption{The scaled local polytope slice $\tilde{\mathsf L}_t$, for $0 < t\leq1/3$ (left). In the right panel, we partition it into a triangular pyramid and a regular tetrahedron.}
\label{fig:K3-t-small}
\end{figure}

Hence, the volume of the combined region is:
\begin{equation}
    \forall t \in (0,1/3] \qquad \vol(\tilde{\mathsf L}_t) = \frac{1}{3}+\frac{1}{6} = \frac{1}{2}.
\end{equation}

\medskip

We now move on to the case $t \in (1/3, 1/2]$. The inequality \eqref{eq:K3-Lpx-1} is non-trivial, and the body $\tilde{\mathsf L}_t$ will depend on the actual value of the parameter $t$. The $V$-representation for $\tilde{\mathsf L}_t$ in this parameter region is 
$$\operatorname{conv} \Big\{(s, 0, 0),(0, s, 0), (0, 0, s),(0, 1, 0), (0, 0, 1), (1, 1, 1)\Big\},$$ 
where \(s := 3-1/t\). The 3-dimensional polytope spanned is shown in \cref{fig:K3-t-large}.

\begin{figure}[htb]
    \centering
    \includegraphics[scale=1]{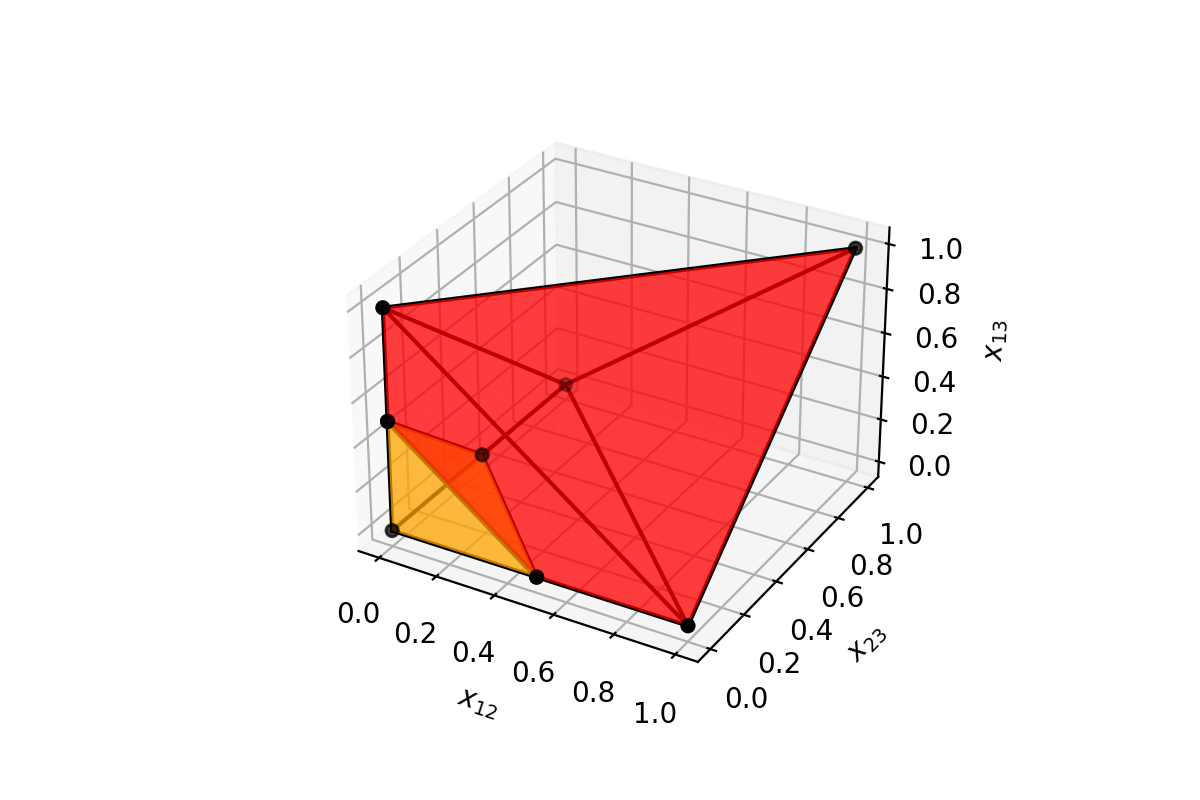}
    \caption{The scaled local polytope slice $\tilde{\mathsf L}_t$, for $1/3 < t\leq 1/2$ (in red). The yellow region represents the difference between this case and the case when $t \leq 1/3$, see \cref{fig:K3-t-small}.}
    \label{fig:K3-t-large}
\end{figure}

The enclosed region has the same volume as the region in \cref{fig:K3-t-small}, left panel, minus the volume of region \(\operatorname{conv} \Big\{(s, 0, 0),(0, s, 0), (0, 0, s),(0, 0, 0)\Big\}\) which is represented in yellow in \cref{fig:K3-t-large}. This again is just the region displayed in \cref{fig:K3-t-small}, top-right panel, scaled by $s$:
\[\vol(\text{bottom-left, yellow region in \cref{fig:K3-t-large}}) = \frac{s^3}{6}.\]
Hence, we have:
\begin{equation}
   \forall t \in (1/3,1/2] \qquad  \vol(\tilde{\mathsf L}_t) = \frac{1}{2} - \frac{s^3}{6} = \frac{1}{2} - \frac{(3-1/t)^3}{6}.
\end{equation}

Putting everything together and using \cref{prop:symVol}, we obtain the main result of this section.

\begin{proposition}\label{prop:K3-vol}
The volume ratio of symmetric slices ($p_1 = p_2 = p_3 = t$) between the local and the non-signaling polytopes corresponding to the triangle graph $K_3$ is given by: 
$$
  \frac{\vol(\mathsf L_t)}{\vol(\mathsf N_t)} = \begin{cases}
  \frac{1}{2}, & t \in(0, \frac{1}{3}] \\
  \frac{1}{2} - \frac{(3-1/t)^3}{6}, & t \in (\frac{1}{3}, \frac{1}{2}]\\
  \frac{1}{2} - \frac{(3-1/(1-t))^3}{6}, & t \in (\frac{1}{2}, \frac{2}{3}]\\
  \frac{1}{2}, &  t \in(\frac{1}{3}, 1).
\end{cases}
$$

    The various parameters as defined in \cref{sec:random-marginal} are:
    \begin{align*}
        \text{The fall-off value} \quad \tau(K_3) &= \frac{1}{3}\\
        \text{The initial ratio} \quad \rho_{0+}(K_3) &= \frac{1}{2}\\
        \text{The middle ratio} \quad \rho_{1/2}(K_3) &= \frac{1}{3}.
    \end{align*}
\end{proposition}

We plot the analytical results alongside numerical calculations using the \texttt{cddlib} library \cite{cddlib} in \cref{fig:vol-ratio-K3-symmetric-slice}.

\begin{figure}[htb]
    \centering
    \includegraphics[width=0.8\textwidth]{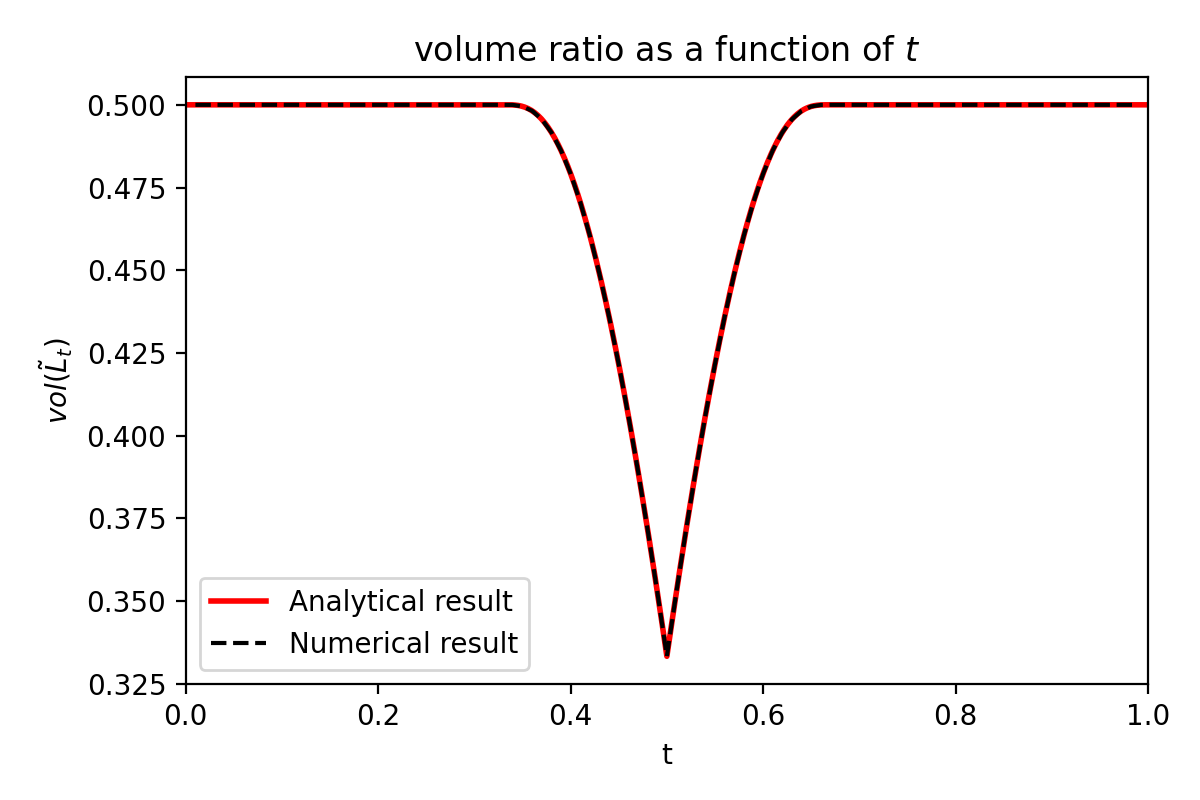}
    \caption{Volume ratio between the local polytope slice and the non-signaling polytope slice, in the symmetric ($p_i = t$) case. The curve is symmetric with respect to $t=1/2$, and constant ($=1/2$) for $t \leq 1/3$.}
    \label{fig:vol-ratio-K3-symmetric-slice}
\end{figure}

\subsection{Skewed slices \texorpdfstring{$(p_1, p_2, p_3) = (t, t, 1/2-t)$}{tt1/2-t}}

We now consider a different, non-symmetric, slice through the local and non-signaling polytopes of the triangle graph. For a parameter $t \in [0, 1/2]$, we consider the slice corresponding to the vertex probabilities 
$$p_1 = p_2 = t, \qquad p_3 = 1/2 - t.$$
Note that this is the only situation of this type that we shall consider in this paper. The previous subsection, as well as all of the subsequent examples will focus on symmetric slices, where all the probabilities associated to vertices are equal.

Using the same symmetry argument as in \cref{prop:symVol}, this slice is isometrically equivalent to the one corresponding to 
$$p_1 = p_2 = 1-t, \qquad p_3 = t-1/2, \qquad \text{for } t \in [1/2,1].$$
As in the previous case, we denote by $\mathsf N_t$ and $\mathsf L_t$ the respective slices. The case $t=0$ yields degenerate sets, so we assume that $t \in (0, 1/2]$. 

Replacing this parameterization for the marginals in \cref{eq:K3-N} and \cref{eq:K3-L}, we have the $H$-representation defining our slice $\mathsf L_t$ as:
\begin{equation}\label{eq:K3-NpC2}
\begin{aligned}
    0 \leq q_{12} &\leq t \\
    0 \leq q_{13}, q_{23} &\leq \min{\{t, \frac{1}{2}-t\}}
\end{aligned}
\end{equation}

\begin{equation}\label{eq:K3-LpC2}
\begin{aligned}
    q_{12} + q_{13} + q_{23} &\geq t - \frac{1}{2} \\
    q_{12} + q_{13} - q_{23} &\leq t \\
     q_{12} - q_{13} + q_{23} &\leq t \\
     -q_{12} + q_{13} + q_{23} &\leq \frac{1}{2}-t
\end{aligned}
\end{equation}
where inequalities \eqref{eq:K3-NpC2} form the $H$-representations of $\mathsf N_t$ and inequalities \eqref{eq:K3-LpC2} are the additional inequalities constraining $\mathsf L_t$. Following a similar approach as in the symmetric case, we substitute \(q_{ij} = tx_{ij}\), obtaining the inequalities for the scaled bodies $\tilde{\mathsf N}_t, \tilde{\mathsf L}_t$:

\begin{equation} \label{eqn:K3-NpxC2}
\begin{aligned}
    0 \leq x_{12} &\leq 1 \\
    0 \leq x_{13}, x_{23} &\leq \min\big{\{1, \frac{1}{2t}-1\big\}}
\end{aligned}
\end{equation}
\vspace{3px}
\begin{subequations} \label{eqn:K3-LpxC2}
\begin{align}
 \label{eq:K3-LpxC2-1}   x_{12} + x_{13} + x_{23} &\geq 1 - \frac{1}{2t} \\
  \label{eq:K3-LpxC2-2}   x_{12} + x_{13} - x_{23} &\leq 1 \\
  \label{eq:K3-LpxC2-3}    x_{12} - x_{13} + x_{23} &\leq 1 \\
  \label{eq:K3-LpxC2-4}    -x_{12} + x_{13} + x_{23} &\leq \frac{1}{2t}-1.
\end{align}
\end{subequations}

Unlike the symmetric slice case discussed previously, the present situation is more involved since $\vol(\tilde{\mathsf N}_t)$ is not equal to $1$ in the entire domain here. We note that inequality \eqref{eq:K3-LpxC2-1} is always satisfied for $t \in (0, 1/2)$ and hence it can be dropped. Inequalities \eqref{eq:K3-LpxC2-2} and \eqref{eq:K3-LpxC2-3} are independent of $t$, while \eqref{eq:K3-LpxC2-4} is trivially satisfied in the range $t \in (0, 1/6]$.

Hence, we start by first looking at the region $t \in (0, 1/6]$. The $H$-representation of $\tilde{\mathsf L}_t$ is given by the relevant inequalities:
\begin{equation} \label{eqn:K3-NpxC2-r1}
    0 \leq x_{ij} \leq 1 \\
\end{equation}
\vspace{2px}
\begin{equation} \label{eqn:K3-LpxC2-r1}
\begin{aligned}
     x_{12} + x_{13} - x_{23} \leq 1 \\
     x_{12} - x_{13} + x_{23} \leq 1
\end{aligned}
\end{equation}
with \cref{eqn:K3-NpxC2-r1} corresponding to $\tilde{\mathsf N}_t$. Hence, $\vol(\tilde{\mathsf N}_t) = 1$. The $V$-representation for $\tilde{\mathsf L}_t$ is easily obtained: 
$$\tilde{\mathsf L}_t = \operatorname{conv}\Big\{(0, 0, 0), (1, 1, 1), (1, 0, 0), (0, 1, 0), (0, 0, 1), (0, 1, 1)\Big\}.$$ 
This region is plotted in Figure 4.

\begin{figure}[htb]
\hspace{-2cm}
\begin{minipage}[b]{.45\textwidth}
    \includegraphics[scale=0.8]{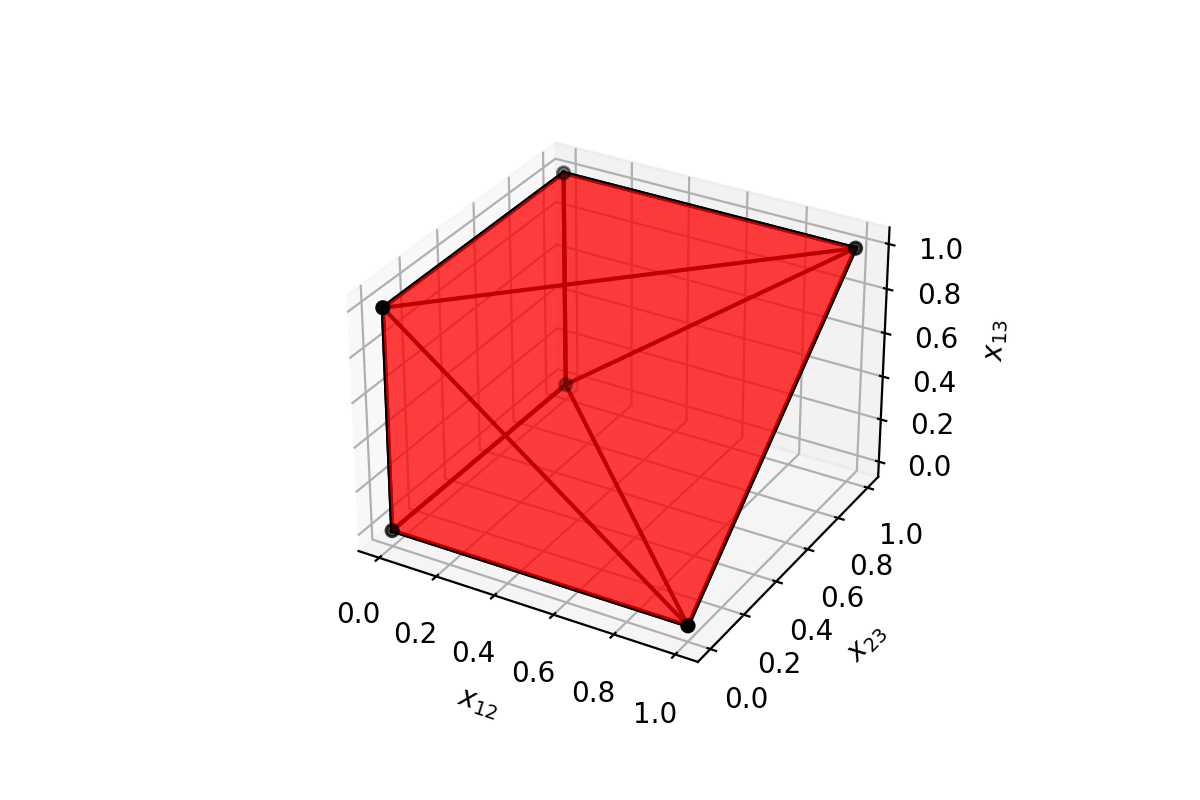}
\end{minipage}\hfill
\begin{minipage}[b]{.45\textwidth}
    \includegraphics[scale=0.5]{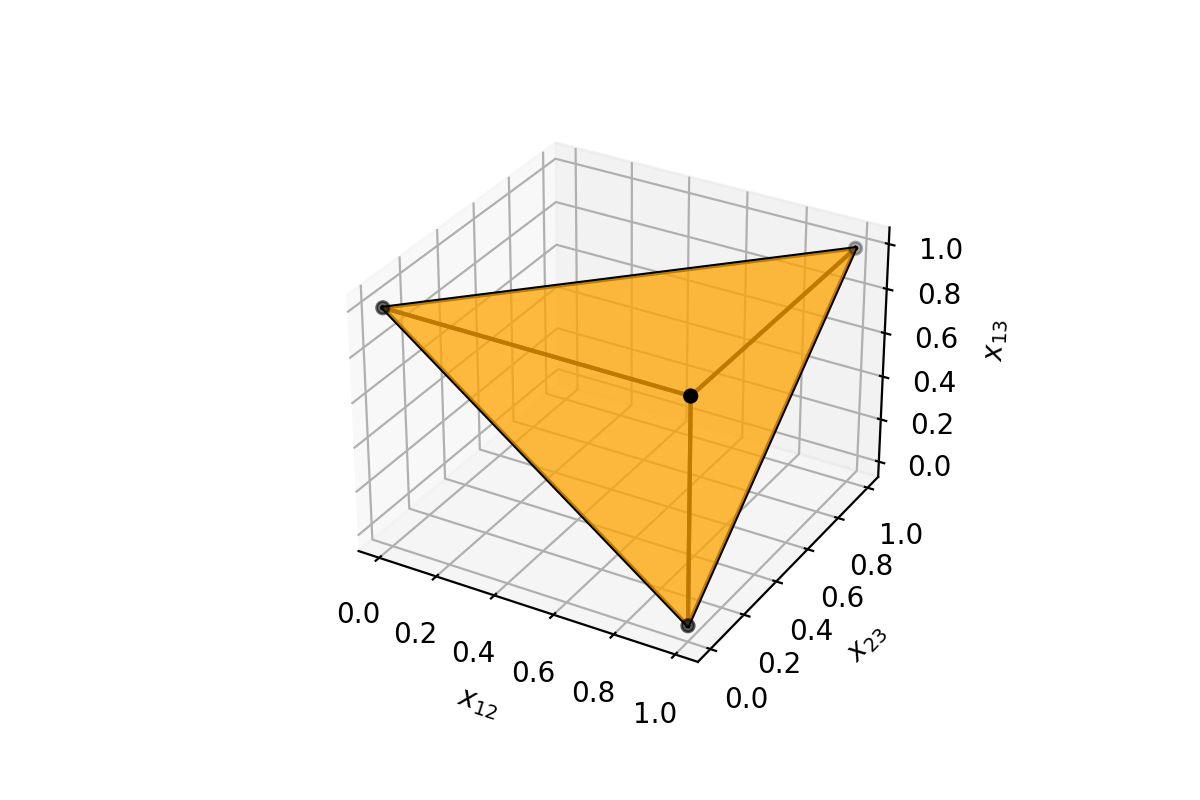} \\
    \includegraphics[scale=0.5]{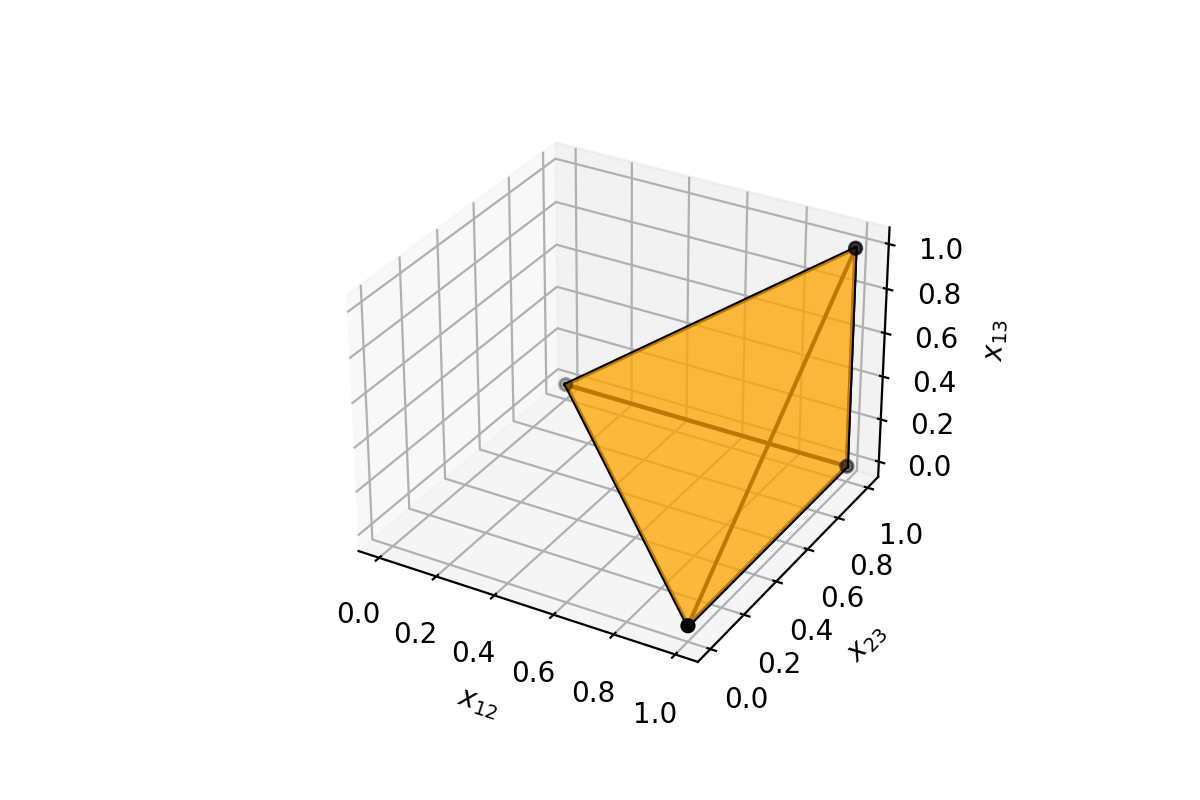}
\end{minipage}
\caption{The scaled local polytope slice $\tilde{\mathsf L}_t$, for $0 < t\leq1/6$ (left). In the right panel, we show the two triangular pyramids removed from the unit cube to obtain figure on the left.}
\label{fig:K3skew-t-small}
\end{figure}

Both $\operatorname{conv}\{(1, 0, 0), (0, 0, 1), (1, 1, 1), (1, 0, 1)\}$ and $\operatorname{conv}\{(1, 0, 0), (0, 0, 1), (1, 1, 1), (1, 1, 0)\}$ are triangular pyramids with volume $1/6$ as can be seen in figure \ref{fig:K3skew-t-small}. The scaled slice $\tilde{\mathsf L}_t$ is just the unit cube with these two regions removed. Hence,
\begin{equation}
    \forall t \in (0,1/6] \qquad \vol(\tilde{\mathsf L}_t) = 1 - 2\cdot\frac{1}{6} = \frac{2}{3}
\end{equation}

Next, we look at the region $t \in (1/6, 1/4]$. The relevant inequalities for the $H$-representation of $\tilde{\mathsf N}_t$ are given by \cref{eqn:K3-NpxC2-r1}. Hence, $\vol(\tilde{\mathsf N}_t) = 1$. 

The $H$-representation of the local slice $\tilde{\mathsf L}_t$ are the inequalities in \cref{eqn:K3-NpxC2-r1} and \eqref{eq:K3-LpxC2-2}-\eqref{eq:K3-LpxC2-4}. The corresponding $V$-representation is 
$$\operatorname{conv}\{(0, 0, 0), (1, 1, 1), (1, 0, 0), (0, 1, 0), (0, 0, 1), (0, 1, k_1), (0, k_1, 1), (k_2, 1, 1)\},$$
where \(k_1 = 1/(2t)-2\) and \(k_2 = 3 - 1/(2t)\). This region is plotted in \cref{fig:K3skew-t-mid}.

\begin{figure}[htb]
    \centering
    \includegraphics[width=1\textwidth]{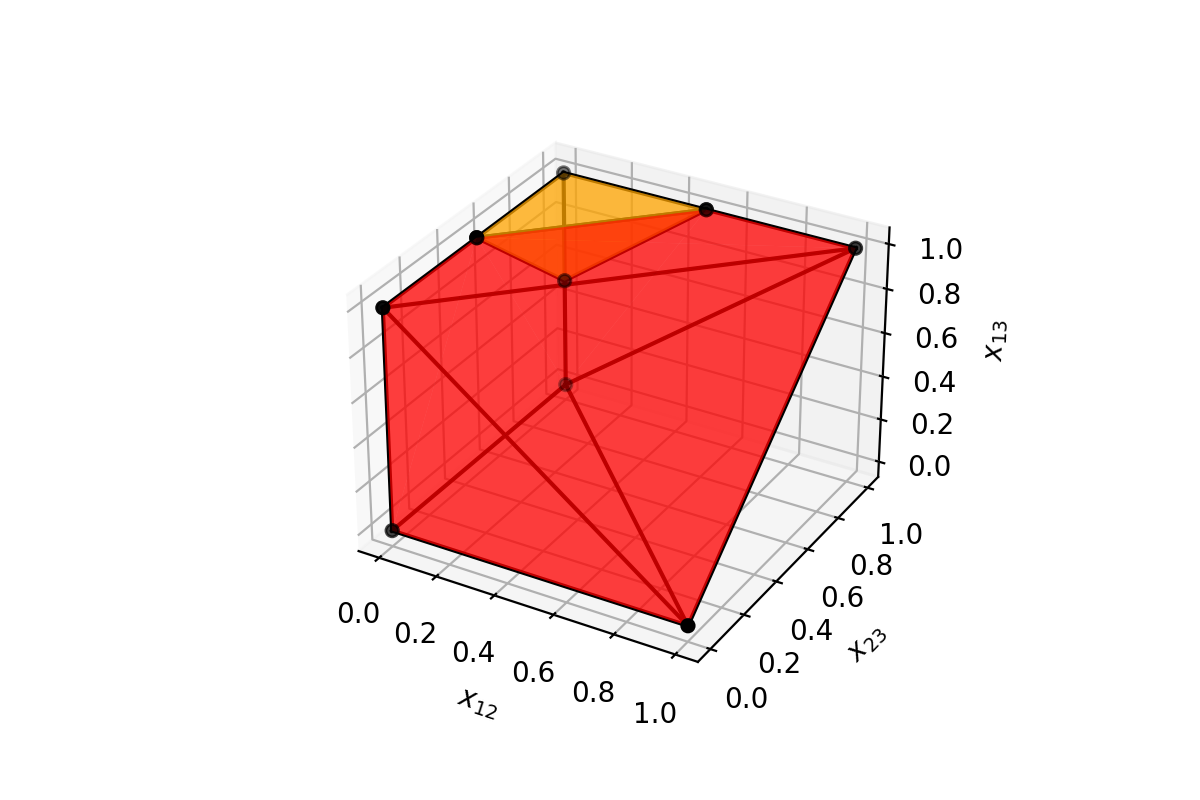}
    \caption{The scaled local polytope slice $\tilde{\mathsf L}_t$, for $1/6 < t\leq 1/4$ (in red). The yellow region represents the difference between this case and the case when $t \leq 1/6$, see \cref{fig:K3skew-t-small}.}
    \label{fig:K3skew-t-mid}
\end{figure}

The volume of the local slice now is just the volume of the local slice obtained for \(t \in (0, 1/6]\) minus the volume of the yellow region. The yellow region is a triangular pyramid scaled by \(3-1/(2t)\) and hence the volume is given by:
\[V_{yellow} = \frac{(3-1/(2t))^3}{6}\]
Hence, we have:
\begin{equation}
    \forall t \in (1/6,1/4] \qquad \vol(\tilde{\mathsf L}_t) = \frac{2}{3} - \frac{(3-1/(2t))^3}{6}
\end{equation}

For the values $t \in (1/4, 1/2]$, we look back at inequalities in \cref{eqn:K3-NpxC2,eqn:K3-LpxC2} and use the substitution \(q_{12} = tx_{12}\), \(q_{23} = (1/2-t)x_{23}\) and \(q_{13} = (1/2-t)x_{13}\). The resulting inequalities are:

\begin{equation} \label{eqn:K3-NpxC2-r3}
    0 \leq x_{ij} \leq 1 \\
\end{equation}
\vspace{2px}
\begin{equation} \label{eqn:K3-LpxC2-r3}
\begin{aligned}
    x_{12} + \Big(\frac{1}{2t}-1\Big)x_{13} - \Big(\frac{1}{2t}-1\Big)x_{23} \leq 1 \\
     x_{12} - \Big(\frac{1}{2t}-1\Big)x_{13} + \Big(\frac{1}{2t}-1\Big)x_{23} \leq 1 \\
     -x_{12} + \Big(\frac{1}{2t}-1\Big)x_{13} + \Big(\frac{1}{2t}-1\Big)x_{23} \leq \frac{1}{2t}-1
\end{aligned}
\end{equation}

In this new substitution, $\vol(\tilde{\mathsf N}_t)$ is unity again. The $V$-representation of $\tilde{\mathsf L}_t$ is \\ 
\begin{align*}
	\operatorname{conv}\Big\{&(0, 0, 0), (1, 1, 1), (1, 0, 0), (0, 1, 0), (0, 0, 1), (1, 1, 1), \\
	&(1/(2t)-1, 1, 1), (2-1/(2t), 1, 0), (2-1/(2t), 0, 1)\Big\}
\end{align*}
We plot this region in \cref{fig:K3skew-t-large}.

\begin{figure}[htb]
\hspace{-2cm}
\begin{minipage}[c]{.45\textwidth}
    \includegraphics[scale=0.8]{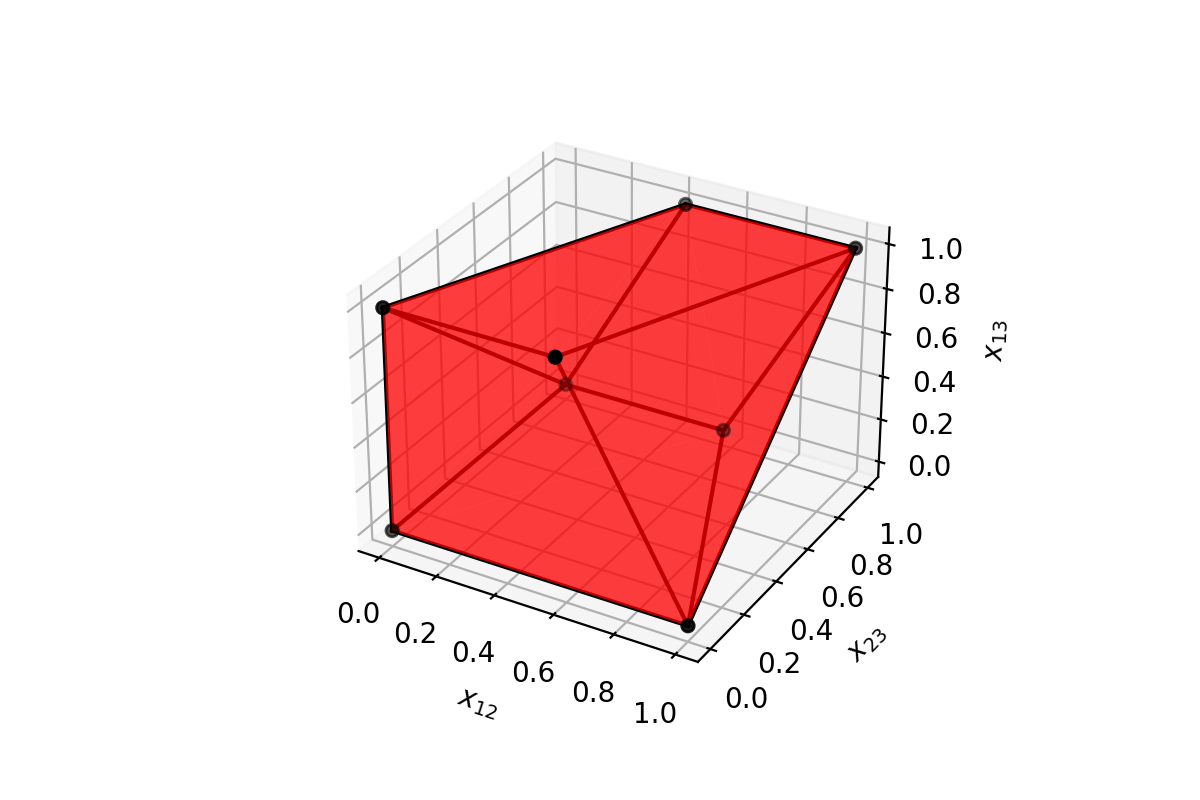}
\end{minipage}\hfill
\begin{minipage}[c]{.45\textwidth}
    \includegraphics[scale=0.4]{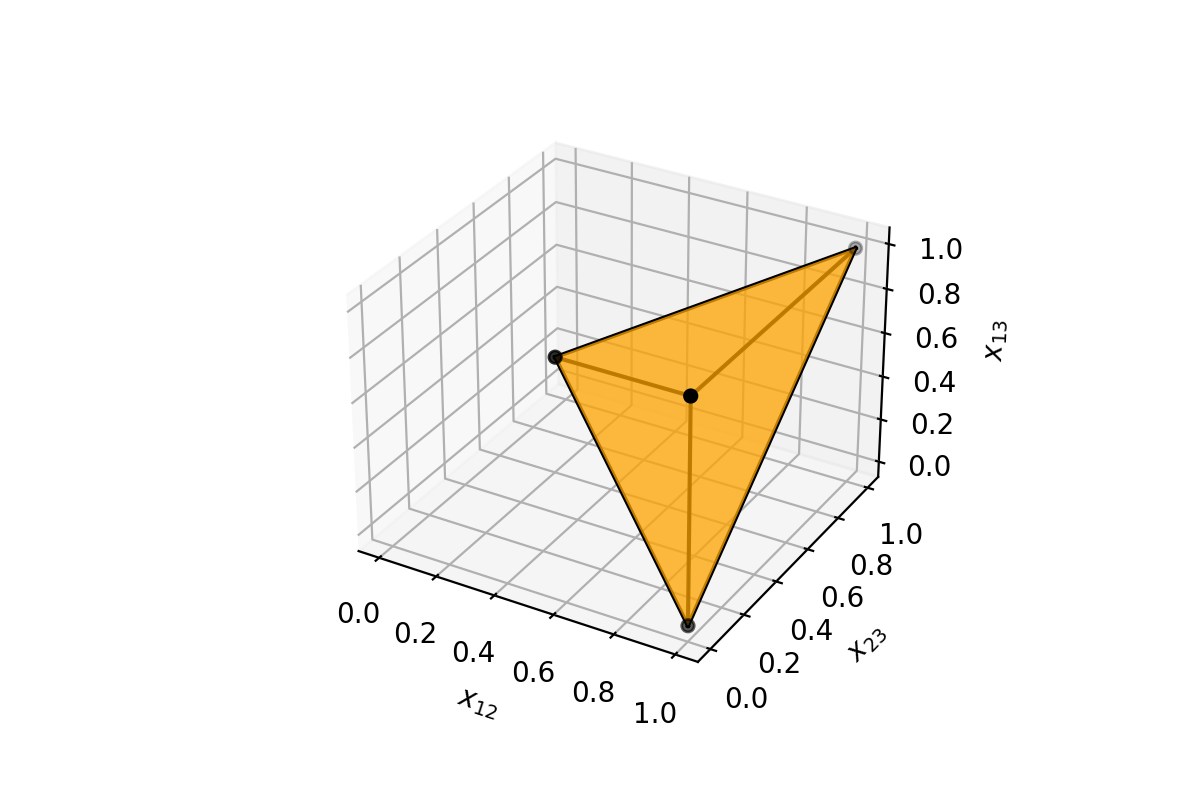} \\
    \includegraphics[scale=0.4]{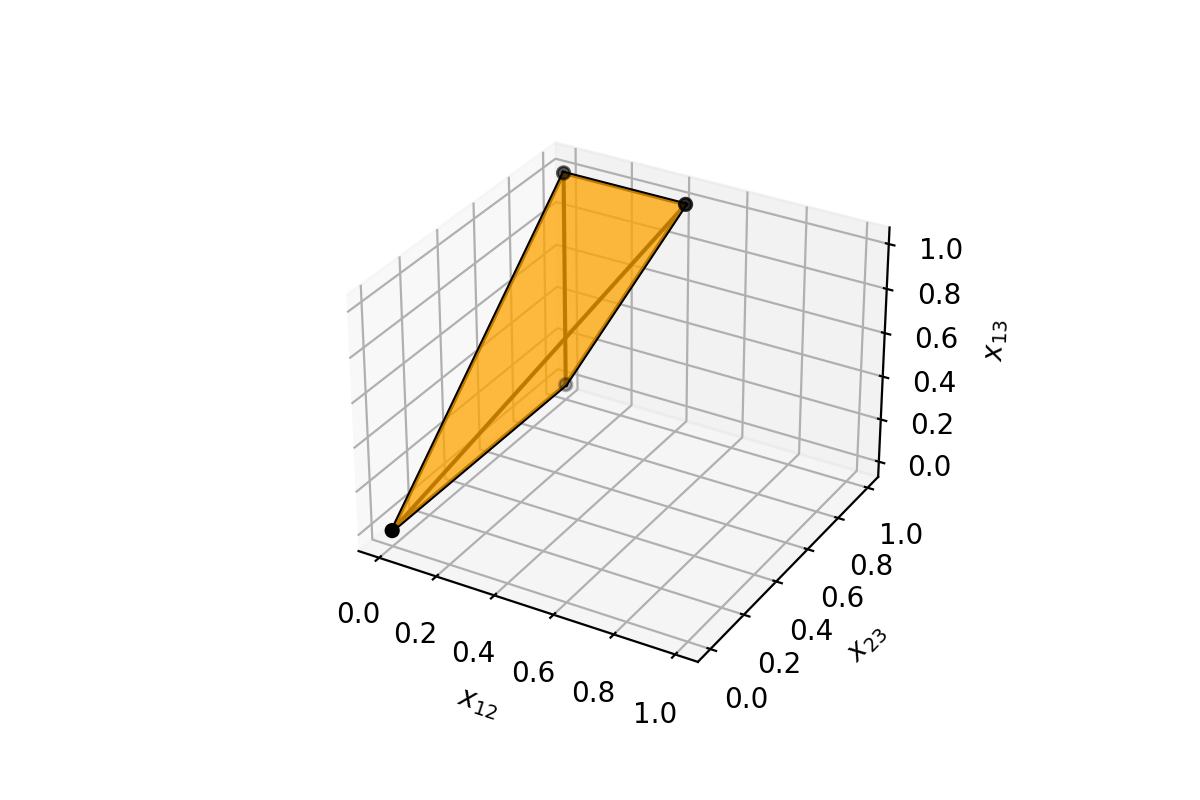} \\
    \includegraphics[scale=0.4]{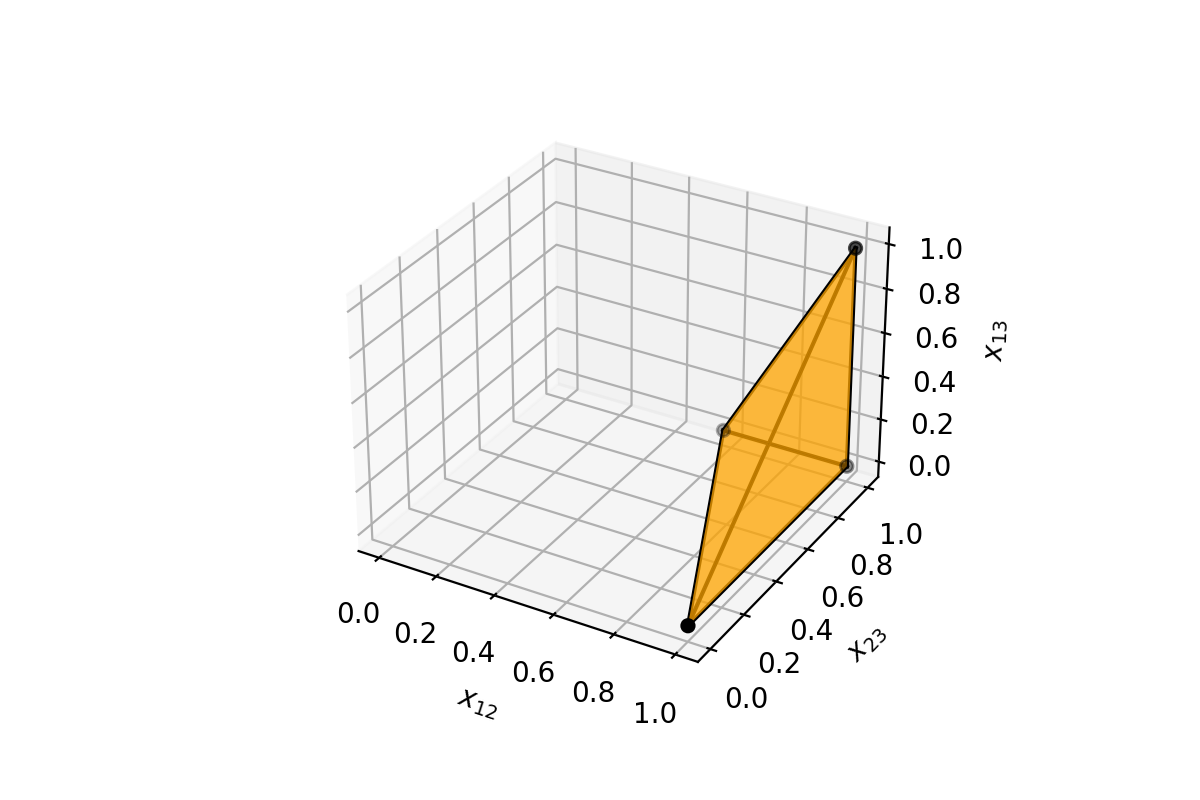}
\end{minipage}
\caption{The scaled local polytope slice $\tilde{\mathsf L}_t$, for $1/4 < t\leq1/2$ (left). In the right panel, we show the three triangular pyramid regions removed from the unit cube to get figure on the left.}
\label{fig:K3skew-t-large}
\end{figure}

It is easy to see that the local slice $\tilde{\mathsf L}_t$ is the unit cube minus the 3 triangular pyramids (shaded yellow) of base area \(1/2\), height \(1/(2t)-1\) and hence of volume,
\[V_{\text{yellow} }= \frac{1}{3}\cdot\frac{1}{2}\cdot(\frac{1}{2t}-1) = \frac{1/(2t)-1}{6}\]

Hence, we have:
\begin{equation}
\forall t \in (1/4,1/2] \qquad \vol(\tilde{\mathsf L}_t) = 1 - \frac{(1/(2t)-1)}{2}
\end{equation}

\begin{proposition}\label{prop:K3-skew-vol}
    The volume ratio of skewed slices ($p_1 = p_2 = t$, $p_3 = 1/2-t$) between the local and no-signaling polytopes corresponding to the triangle graph $K_3$ are given by (see \cref{fig:K3-skew-vol}):
$$
  \frac{\vol(\mathsf L_t)}{\vol(\mathsf N_t)} = \begin{cases}
  \frac{2}{3}, & t \in(0, \frac{1}{6}] \\
  \frac{2}{3} - \frac{(3-1/(2t))^3}{6}, & t \in (\frac{1}{6}, \frac{1}{4}]\\
  1 - \frac{(1/(2t)-1)}{2}, & t \in (\frac{1}{4}, \frac{1}{2})
\end{cases}
$$  
\end{proposition}

\begin{figure}[!htb]
    \centering
    \includegraphics[width=0.8\textwidth]{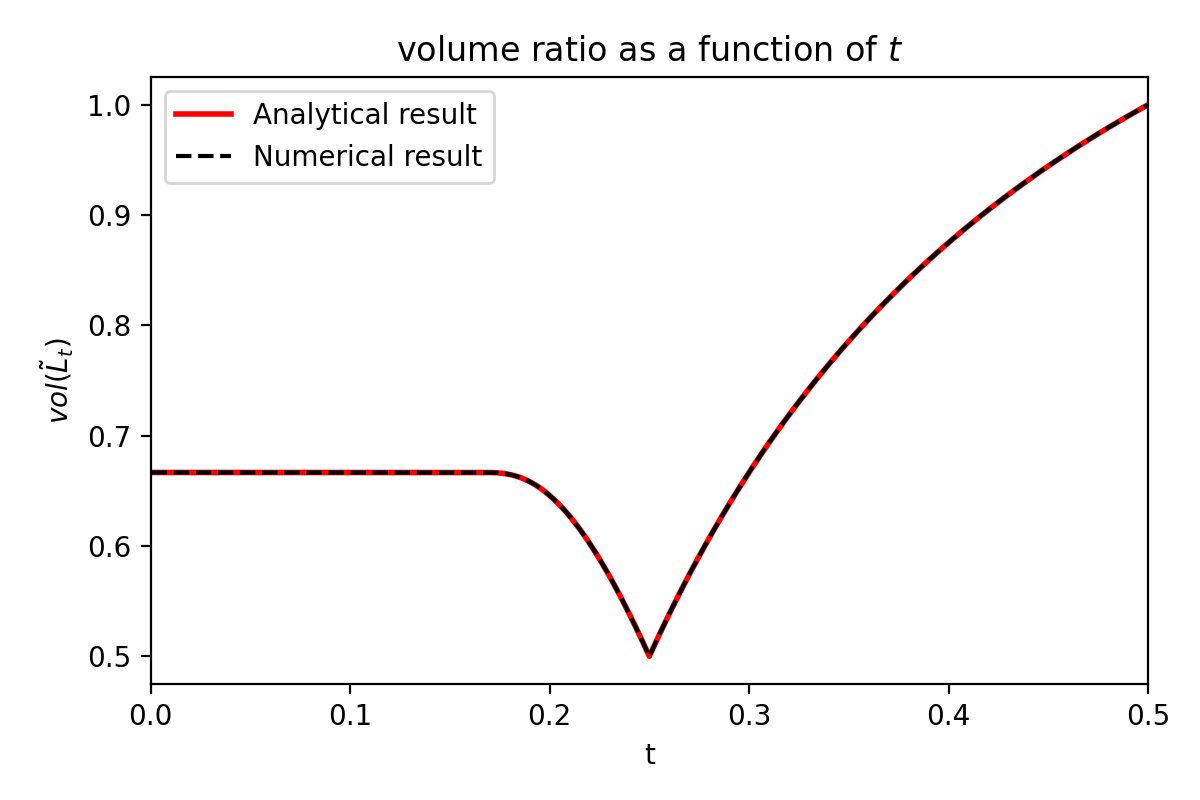}
    \caption{Volume ratio between the local polytope slice and the non-signaling polytope slice, in the skewed ($p_1 = p_2 = t, p_3 = 1/2-t$) case.}
    \label{fig:K3-skew-vol}
\end{figure}

\section{The square graph}\label{sec:square}

In this section, we analyze in detail the case of $K_{2, 2}$ which is the complete-bipartite graph. As explained in \cref{sec:quantum-info}, this case corresponds to the famous CHSH scenario \cite{clauser1969proposed}, and more generally 2-player non-local games with 2 questions and 2 answers. In this work we are focusing only on the local and non-signaling sets of correlations; for an in-depth study of the set of \emph{quantum} correlations, we refer the reader to the excellent work \cite{le2023quantum}. Note that we embed below the graph $K_{2,2}$ in the plane as a square, see also \cref{sec:cycle} for the discussion of arbitrary cycle graphs.

\begin{figure}[htb]
  \centering
  \begin{tikzpicture}[auto, scale = 1.4]

    \tikzstyle{vertex1}=[circle, draw=blue, fill=blue!10!, ultra thick]
    \tikzstyle{vertex2}=[circle, draw=red, fill=red!10!, ultra thick]
    \tikzstyle{edge1}=[draw=black, ultra thick]
    \tikzstyle{edge2}=[draw=red, ultra thick]
    \tikzstyle{plus}=[]

    \node[vertex1] (v1) at (-1.5, 1.5) {$p_1$};
    \node[vertex1] (v2) at (1.5, 1.5) {$p_2$};
    \node[vertex1] (v3) at (1.5, -1.5) {$p_3$};
    \node[vertex1] (v4) at (-1.5, -1.5) {$p_4$};

    \draw[edge1] (v1) edge node{$q_{12}$} (v2);
    \draw[edge1] (v2) edge node{$q_{23}$} (v3);
    \draw[edge1] (v3) edge node{$q_{34}$} (v4);
    \draw[edge1] (v1) edge node{$q_{14}$} (v4);
    
    \end{tikzpicture}
  \label{fig:square-graph}
\end{figure}

The $V$-representation of $\mathsf L(K_{2,2}) = \mathsf{COR} (K_{2,2})$ is given by the convex hull of all rows of the truth table presented in \cref{tab:K22}. This is then used to obtain the $H$-representation below: 

\begin{equation}\label{eq:K22-N}
\begin{aligned}
  0 \leq q_{ij} \leq \min{(p_{i}, p_{j})} \\
  p_i + p_j -p_{ij} \leq 1
\end{aligned}
\end{equation}
\begin{equation}\label{eq:K22-L}
    \begin{aligned}
        0 \leq p_{3} + p_{4} + q_{12} - q_{23} - q_{34} - q_{14}  &\leq 1 \\
        0 \leq p_{1} + p_{4} - q_{12} + q_{23} - q_{34} - q_{14}  &\leq 1 \\
        0 \leq p_{1} + p_{2} - q_{12} - q_{23} + q_{34} - q_{14}  &\leq 1 \\
        0 \leq p_{2} + p_{3} - q_{12} - q_{23} - q_{34} + q_{14}  &\leq 1 \\
     \end{aligned}
\end{equation}

\begin{table}[htb]
   \bgroup
\def\arraystretch{1.5}
\begin{minipage}[b]{0.45\textwidth}
    \begin{center} 
        \begin{tabular}{|cccc|cccc|}
        \hline
\rowcolor[HTML]{C0C0C0}      
        \multicolumn{4}{|c|}{$V$-part}                 & \multicolumn{4}{|c|}{$E$-part}                            \\ \hline
        \rowcolor[HTML]{EFEFEF} 
        1 & 2 & 3 & 4 & 12 & 23 & 34 & 14 \\ \hline
        \hline \rowcolor[HTML]{FFFFFF} 
        0 & 0 & 0 & 0 & 0 & 0 & 0 & 0 \\ \hline
        0 & 0 & 0 & 1 & 0 & 0 & 0 & 0 \\ \hline
        0 & 0 & 1 & 0 & 0 & 0 & 0 & 0\\ \hline
        0 & 0 & 1 & 1 & 0 & 0 & 1 & 0 \\ \hline
        0 & 1 & 0 & 0 & 0 & 0 & 0 & 0 \\ \hline
        0 & 1 & 0 & 1 & 0 & 0 & 0 & 0 \\ \hline
        0 & 1 & 1 & 0 & 0 & 1 & 0 & 0 \\ \hline
        0 & 1 & 1 & 1 & 0 & 1 & 1 & 0 \\ \hline
        \end{tabular}
  \end{center}
\end{minipage}\hfill
\begin{minipage}[b]{0.45\textwidth}
    \begin{center} 
        \begin{tabular}{|cccc|cccc|}
        \hline
\rowcolor[HTML]{C0C0C0}      
        \multicolumn{4}{|c|}{$V$-part}                 & \multicolumn{4}{|c|}{$E$-part}                            \\ \hline
        \rowcolor[HTML]{EFEFEF} 
        1 & 2 & 3 & 4 & 12 & 23 & 34 & 14 \\ \hline
        \hline \rowcolor[HTML]{FFFFFF} 
        1 & 0 & 0 & 0 & 0 & 0 & 0 & 0 \\ \hline
        1 & 0 & 0 & 1 & 0 & 0 & 0 & 1 \\ \hline
        1 & 0 & 1 & 0 & 0 & 0 & 0 & 0 \\ \hline
        1 & 0 & 1 & 1 & 0 & 0 & 1 & 1 \\ \hline
        1 & 1 & 0 & 0 & 1 & 0 & 0 & 0 \\ \hline
        1 & 1 & 0 & 1 & 1 & 0 & 0 & 1 \\ \hline
        1 & 1 & 1 & 0 & 1 & 1 & 0 & 0 \\ \hline
        1 & 1 & 1 & 1 & 1 & 1 & 1 & 1 \\ \hline
        \end{tabular}
  \end{center}
\end{minipage}
\caption{$V$-representation of $\mathsf L(K_{2, 2})$, split in two tables.}\label{tab:K22}
\egroup

\end{table}

The polytopes in question in this case are $8$-dimensional corresponding to $4$ each of vertices ($p_i$) and edges ($q_{ij}$).

\[\mathsf L(K_{2, 2}) = \mathsf{COR}(K_{2,2}) \subseteq \mathsf N(K_{2,2}) = \mathsf{TRA}(K_{2,2}) \subset \mathbb R^8\]

Inequalities in \cref{eq:K22-N} correspond to $\mathsf N(K_{2,2})$ while inequalities in \cref{eq:K22-L} are the additional constraints defining $\mathsf L(K_{2, 2})$; these are the \emph{Bell inequalities}. The volume of these polytopes is given by \cite{lee2020volume}:

\[\vol(\mathsf N(K_{2,2})) = \frac{17}{10800}\qquad \text{and} \qquad \vol(\mathsf L(K_{2,2})) = \frac{1}{630} = \frac{16}{17} \vol(\mathsf N(K_{2,2})).\]

We set $p_1 = p_2 = p_3 = p_4 = t \in (0, 1]$ and study the slices $\mathsf L(p=(t,t,t,t), K_{2,2})$ and $\mathsf N(p=(t,t,t,t), K_{2,2})$. The corresponding $H$-representation for $\mathsf L(p=(t,t,t,t), K_{2,2})$ reads:

\begin{equation}\label{eq:K22-Np}
\begin{aligned}
  \max\{0, 2t-1\} \leq q_{ij} \leq t
\end{aligned}
\end{equation}
\vspace{3px}
\begin{equation}\label{eq:K22-Lp}
    \begin{aligned}
        -2t \leq q_{12} - q_{23} - q_{34} - q_{14}  &\leq 1 - 2t \\
        -2t \leq - q_{12} + q_{23} - q_{34} - q_{14}  &\leq 1 - 2t \\
        -2t \leq - q_{12} - q_{23} + q_{34} - q_{14}  &\leq 1 - 2t \\
        -2t \leq - q_{12} - q_{23} - q_{34} + q_{14}  &\leq 1 - 2t\\
     \end{aligned}
\end{equation}

As before, we define our polytopes of interest by:
\begin{align*}
    \mathsf N_t &:= \mathsf N(p=(t,t,t,t), K_{2,2}) = \{(q_{12},q_{23},q_{34}, q_{14}) \in \mathbb R^4 \, : \, \text{ \cref{eq:K22-Np} holds}\}\\
    \mathsf L_t &:= \mathsf L(p=(t,t,t,t), K_{2,2}) = \{(q_{12},q_{23},q_{34}, q_{14}) \in \mathbb R^4 \, : \, \text{ \cref{eq:K22-Np} and \cref{eq:K22-Lp} hold}\}.
\end{align*}

We'll focus only in the range $t\in(0, 1/2]$ owing to \cref{prop:symVol}. Setting $q_{ij} = tx_{ij}$, we finally have:

\begin{equation}\label{eq:K22-Npx}
\begin{aligned}
  0 \leq x_{ij} \leq 1
\end{aligned}
\end{equation}
\vspace{3px}
\begin{equation}\label{eq:K22-Lpxt}
    \begin{aligned}
        x_{12} - x_{23} - x_{34} - x_{14}  &\leq \frac{1}{t} - 2 \\
        - x_{12} + x_{23} - x_{34} - x_{14}  &\leq \frac{1}{t} - 2 \\
        - x_{12} - x_{23} + x_{34} - x_{14}  &\leq \frac{1}{t} - 2 \\
        - x_{12} - x_{23} - x_{34} + x_{14}  &\leq \frac{1}{t} - 2\\
     \end{aligned}
\end{equation}
\vspace{3px}
\begin{equation}\label{eq:K22-Lpx}
    \begin{aligned}
        - x_{12} + x_{23} + x_{34} + x_{14}  &\leq 2 \\
         x_{12} - x_{23} + x_{34} + x_{14}  &\leq 2 \\
         x_{12} + x_{23} - x_{34} + x_{14}  &\leq 2 \\
         x_{12} + x_{23} + x_{34} - x_{14}  &\leq 2\\
     \end{aligned}
\end{equation}

Following the previous section we define our scaled polytopes: 
\begin{align*}
    \tilde{\mathsf N}_t &:= \{(x_{12},x_{23},x_{34}, x_{14}) \in \mathbb R^4 \, : \, \text{ \cref{eq:K22-Npx} holds}\}\\
    \tilde{\mathsf L}_t &:= \{(x_{12},x_{23},x_{34}, x_{14}) \in \mathbb R^4 \, : \, \text{ \cref{eq:K22-Npx} and \cref{eq:K22-Lpx} hold}\}.
\end{align*}
Since, $\vol(\mathsf N_t) = t^4 \vol(\tilde{\mathsf N}_t)$, and similarly for the local polytope; we have,
$$\frac{\vol(\mathsf L_t)}{\vol(\mathsf N_t)} = \frac{\vol(\tilde{\mathsf L}_t)}{\vol(\tilde{\mathsf N}_t)}.$$

From inequalities in \cref{eq:K22-Npx}, we see $\tilde{\mathsf N}_t$ is the $4$-D unit cube $[0, 1]^4$ with volume given by $\vol (\tilde{\mathsf N}_t) = 1$ independent of $t$. Next, we notice that inequalities \eqref{eq:K22-Lpxt} are trivially satisfied for all $t \in \Big(0, 1/3]$ while the inequalities in \cref{eq:K22-Lpx} are independent of $t$.

Hence, we will start by looking at $t \in (0, 1/3]$, where the only relevant inequalities in the $H$-representation of $\tilde{\mathsf L}_t$ are \eqref{eq:K22-Npx} and \eqref{eq:K22-Lpx}. The $V$-representation is thus given by the convex hull of the following 12 vertices:
\begin{align*}
    (0, 0, 0, 0), (1, 0, 0, 0), (0, 1, 0, 0), (0, 0, 1, 0), (0, 0, 0, 1),(1, 1, 0, 0),\\ (1, 0, 1, 0), (1, 0, 0, 1), (0, 1, 1, 0), (0, 1, 0, 1), (0, 0, 1, 1), (1, 1, 1, 1)
\end{align*}

We note that these are all the vertices of a tesseract but with permutations of $(1, 1, 1, 0)$ removed. Removing each of these permutation amounts to the removal of a $4$-simplex. Using the fact that the volume of an $n$-simplex is given by $1/n!$, we have:
\begin{equation}
    \forall t \in \Big(0,\frac{1}{3}\Big] \qquad \vol(\tilde{\mathsf L}_t) = 1-4\cdot\frac{1}{4!} = \frac{5}{6}.
\end{equation}

Next, looking in the range $t \in (1/2,1/3]$, the inequalities \eqref{eq:K22-Lpxt} start cutting in and the $V$-representation for $\tilde{\mathsf L}_t$ is given by,
\begin{gather*}
        \operatorname{conv}\Big\{(0, 0, 0, 0)\\
        \Big(\frac{1}{t}-2, 0, 0, 0\Big), \Big(1, 3-\frac{1}{t}, 0, 0\Big), \Big(1, 0, 3-\frac{1}{t}, 0\Big), \Big(1, 0, 0, 3-\frac{1}{t}\Big)\\
        \Big(0, \frac{1}{t}-2, 0, 0\Big), \Big(3-\frac{1}{t}, 1, 0, 0\Big), \Big(0, 1, 3-\frac{1}{t}, 0\Big), \Big(0, 1, 0, 3-\frac{1}{t}\Big)\\
        \Big(0, 0, \frac{1}{t}-2, 0\Big), \Big(3-\frac{1}{t}, 0, 1, 0\Big), \Big(0, 3-\frac{1}{t}, 1, 0\Big), \Big(0, 0, 1, 3-\frac{1}{t}\Big)\\
        \Big(0, 0, 0, \frac{1}{t}-2\Big), \Big(3-\frac{1}{t}, 0, 0, 1\Big), \Big(0, 3-\frac{1}{t}, 0, 1\Big), \Big(0, 0, 3-\frac{1}{t}, 1\Big)\\
        (1, 1, 0, 0), (1, 0, 1, 0), (1, 0, 0, 1), (0, 1, 1, 0), (0, 1, 0, 1), (0, 0, 1, 1), (1, 1, 1, 1)\Big\}
\end{gather*}

We observe that each of the inequalities in \eqref{eq:K22-Lpxt} splits each permutation of $(1, 0, 0, 0)$ into $4$ vertices amounting to a removed volume equivalent to a 4-simplex scaled by $3-1/t$. Thus, we have:

\begin{equation}
    \forall t \in \Big(\frac{1}{3},\frac{1}{2}\Big] \qquad \vol(\tilde{\mathsf L}_t) = \frac{5}{6} - 4\cdot\frac{1}{4!}\cdot\Big(3-\frac{1}{t}\Big)^4 = \frac{5-(3-1/t)^4}{6}
\end{equation}

\begin{figure}[htb]
\centering
  \includegraphics[width=.4\linewidth]{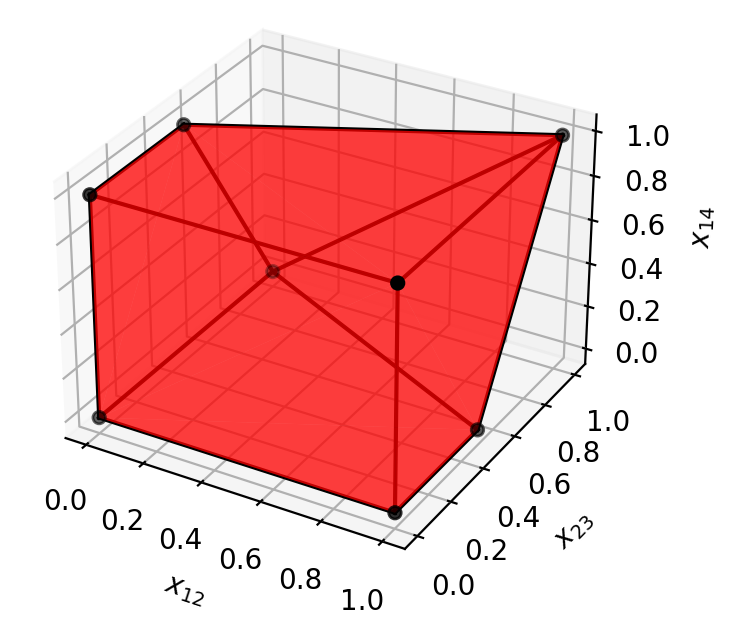}\qquad\qquad
  \includegraphics[width=.4\linewidth]{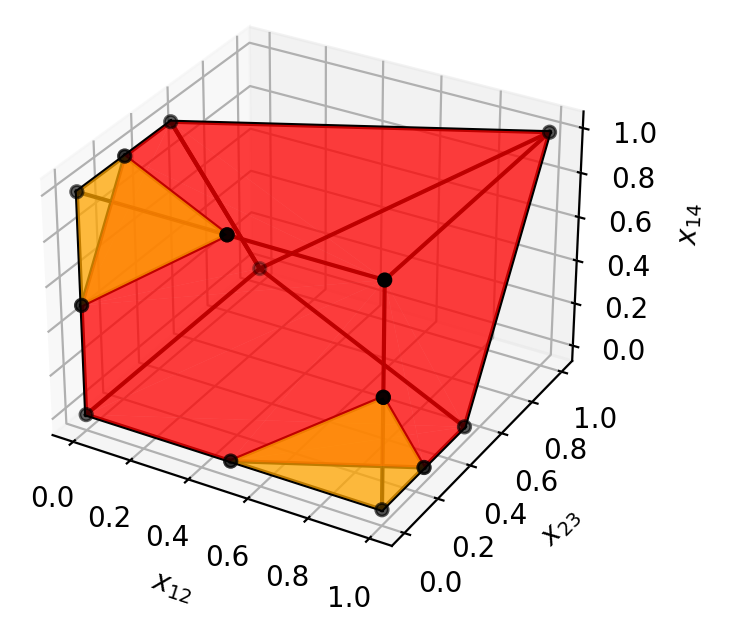}
\caption{The $x_{23} = x_{34}$ slices of $\mathsf L(p, K_{2,2})$. Left: $t \in \Big(0, \frac{1}{3}\Big]$. Right: $t = 0.4$; the yellow sections are the volumes removed due to inequalities \eqref{eq:K22-Lpxt}.}
\label{fig:K22}
\end{figure}

Thus, combining the results of this section, we obtain our main result. This is plotted against numerical result in Figure \ref{fig:vol-ratio-K22}.

\begin{proposition}\label{prop:K22-vol}
The volume ratio of slices ($p_1 = p_2 = p_3 = p_4 = t$) between the local and the non-signaling polytopes corresponding to the complete-bipartite graph $K_{2,2}$ is given by: 
$$
  \frac{ \vol(\mathsf L_t)}{\vol(\mathsf N_t)} = \begin{cases}
  \frac{5}{6}, & t \in(0, \frac{1}{3}] \\
  \frac{5-(3-1/t)^4}{6}, & t \in (\frac{1}{3}, \frac{1}{2}]\\
  \frac{5-(3-1/(1-t))^4}{6}, & t \in (\frac{1}{2}, \frac{2}{3}]\\
  \frac{5}{6}, &  t \in(\frac{1}{3}, 1]
\end{cases}
$$
    The various parameters as defined in \cref{sec:random-marginal} are:
    \begin{align*}
        \text{The fall-off value} \quad \tau(K_{2, 2}) &= \frac{1}{3}\\
        \text{The initial ratio} \quad \rho_{0+}(K_{2, 2}) &= \frac{5}{6}\\
        \text{The middle ratio} \quad \rho_{1/2}(K_{2, 2}) &= \frac{2}{3}.
    \end{align*}
\end{proposition}

\begin{figure}[htb]
    \centering
    \includegraphics[width=0.8\textwidth]{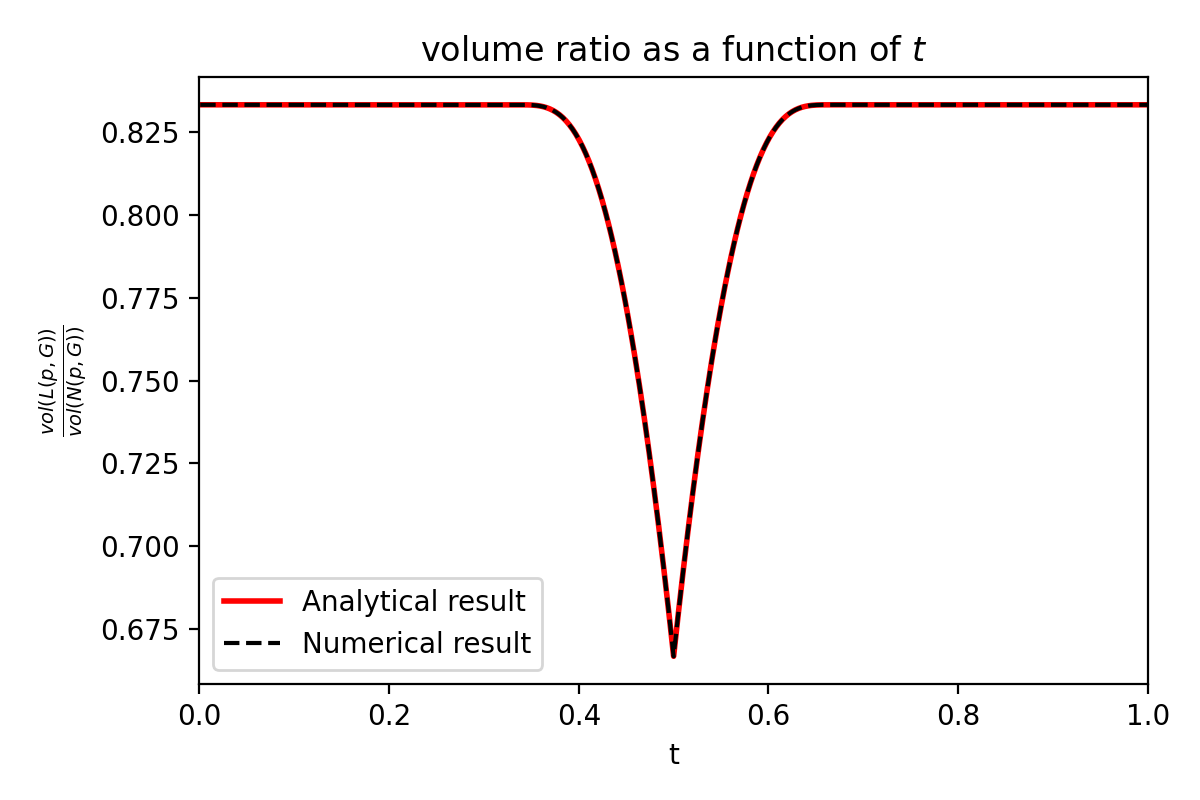}
    \caption{\(\frac{\vol(\mathsf L(p, K_{2,2}))}{\vol(\mathsf N(p, K_{2,2}))}\) as a function of $t$}
    \label{fig:vol-ratio-K22}
\end{figure}

In terms of the CHSH game, \cref{prop:K22-vol} can be interpreted as follows. Consider Alice and Bob sharing a randomly sampled no-signaling box with the condition that the marginal distribution of each of their own questions are fixed to be $t$. In such a case, the highest probability of the randomly sampled box being non-local occurs when $t=1/2$. In fact, the $t=1/2$ value is interesting in its own right, as many important behaviours such as the PR boxes \cite{popescu1994quantum} and the boxes corresponding to maximal quantum violations (where Alice and Bob share a maximally entangled Bell state) all lie on this slice.

\section{Cycle graphs}\label{sec:cycle}

In the previous two sections, we have discussed in great detail the cases of the triangle graph $K_3$ and that of the complete bipartite graph $K_{2,2}$. These are \emph{cycle graphs}, $K_3 = C_3$ and $K_{2,2} = C_4$. Another well-studied cyclic scenario often found in literature is the KCBS \cite{klyachko2008simple} scenario which corresponds to $C_5$. Physically this is equivalent to measuring the two components of a bipartite system along five distinct measurement directions (5 questions) each yielding two possible outcomes (2 answers each). We would not study $C_5$ specifically but rather in this section, we shall establish some general results for cycle graphs of arbitrary order $C_n$. Our analysis builds on the work of Ara\'ujo, T\'ulio Quintino, Budroni, Terra Cunha, and Cabello \cite{araujo2013all}, in which the facets of $L(C_n)$ have been described. and we will use them to derive general results for the volume of slices $L(C_n, p=(t,\ldots, t))$ and $N(C_n, p=(t,\ldots, t))$.
 
\begin{proposition}
    [{{\cite[Theorem 1]{araujo2013all}}}]\label{prop:Cn-contextIneq} For a set of observables $\{X_0, X_1,...,X_{n-1}\}$, all the $2^{n-1}$ tight noncontextuality inequalities for
the n-cycle noncontextual polytope are 
\begin{equation} \label{eq:araujo-X}
    \sum_{i=0}^{n-1}\gamma_i \langle X_i X_{i+1}\rangle \leq n-2,
\end{equation}
where $\gamma_i = \pm 1$ are such that there are odd number of $-1$'s.
\end{proposition}

Here, $\langle X_iX_{i+1} \rangle = 4q_{ij}-2p_i-2p_j+1$, the non-contextual inequalities in the $H$-representation of $\mathsf L(C_n)$ are given by:
\begin{equation} \label{eq:araujo-pq}
    4\sum_{i=0}^{n-1} \gamma_i q_{i,i+1} - 2\sum_i^{n-1} \gamma_i (p_i + p_{i+1}) + \sum_i^{n-1}\gamma_i\leq n-2
\end{equation}

We can get the corresponding noncontextuality inequalities for the slice $\mathsf L(C_n)_t := \mathsf L(p=\{t,\ldots,t\}, C_n)$ by just substituting $p=\{t,\ldots,t\}$ in ($\ref{eq:araujo-pq}$).

\begin{definition}
    For $p=\{t,\ldots,t\}$, we define the $m$-negative inequality for $\mathsf L(p=\{t,\ldots,t\}, C_n)$ as the inequality with $\gamma_i = -1$ for $m$ values. Note that $m$ can only take odd values for valid noncontextuality inequalities. Thus, the $m$-negative inequalities, after substituting $q_{ij} = tx_{ij}$ are given by:
    \begin{equation} \label{eq:m-neg}
        \sum_{i=0}^{n-1} \gamma_i x_{ij} \leq \frac{m-1}{2t}+n-2m
    \end{equation}
\end{definition}

The remaining inequalities in the $H$-representation of $\mathsf L(C_n)_t$ are of the form given in \cref{eq:K3-N} corresponding to $\mathsf N(C_n)_t$. Under the substitution $p_{ij} = tx_{ij}$ and $t \in \Big(0, \frac{1}{2}\Big]$ these are again of the form:
\begin{equation}\label{eq:Nx}
    0 \leq x_{ij} \leq 1
\end{equation}

Denoting the polytopes of interest again by:
\begin{align*}
    \tilde{\mathsf N}_t &:= \{(x_{12},x_{23},\ldots,x_{1n-1}) \in \mathbb R^n \, : \, \text{ \cref{eq:Nx} holds}\}\\
    \tilde{\mathsf L}_t &:= \{(x_{12},x_{23},\ldots,x_{1n-1}) \in \mathbb R^n \, : \, \text{ \cref{eq:Nx} and \cref{eq:m-neg} hold}\}.
\end{align*}
Again, it is straightforward to see that $\tilde{\mathsf N}_t$ is the unit cube $[0, 1]^n$. We will now look at how the m-inequalities cut into this cube to give shape to $\tilde{\mathsf L}_t$. We start by noting the the $1$-negative inequalities are independent of $t$. The RHS of inequality \eqref{eq:m-neg} becomes $n-2$ for $m=1$. This implies that the vertices of the cube with $n-1$ number of $1$s are no longer part of $\tilde{\mathsf L}_t$. In fact, the resultant body formed after imposing $1$-negative constraint on $\tilde{\mathsf N}_t$ is just the convex hull of the remaining vertices of the cube.

\begin{proposition}
    The volume removed by $1$-negative inequalities from $\tilde{\mathsf N}_t$ is $\frac{1}{n!}$.
\end{proposition}

This is because each $1$-negative inequality removes a $n$-simplex from the cube.

\begin{proposition}
    Any $m$-negative inequality is trivially satisfied for all $t \in \Big[0, \frac{m-1}{2m}\Big]$.
\end{proposition}

This is clear if one notices that maximum value the right hand side of eqn \ref{eq:m-neg} can take is $n-m$.

\begin{proposition} \label{prop:m-splitting}
    For $m\geq 2$, an m-inequality starts cutting into $\tilde{\mathsf N}_t$ by splitting the points of form \[\perm ({1,1,\ldots,1\smash{\llap{$\underbrace{\phantom{1,1,\ldots,1}}_{\text{$n-m$ times}}$}}},
    {0,0,\ldots,0\smash{\llap{$\underbrace{\phantom{0,0,\ldots,0}}_{\text{$m$ times}}$}}})\]\\
    into $n$ points. These n points correspond to either one $0$ being replaced by $m-\frac{m-1}{2t}$ or one $1$ being replaced by $1-m+\frac{m-1}{2t}$.
\end{proposition}
    
\begin{proposition}
    For $m\geq 2$  and $t\geq\frac{m-1}{2m}$, the volume removed by an $m$-inequality from $\tilde{\mathsf N}_t$ is \[\frac{1}{n!}\cdot\Big(m-\frac{m-1}{2t}\Big)^n\].
\end{proposition}

This follows from \cref{prop:m-splitting} as the region removed is just a $n$-simplex scaled by $m-\frac{m-1}{2t}$. Next, we notice that for $t=\frac{1}{2}$, $m-\frac{m-1}{2t}=1$  and hence, an $m$-inequality takes points with $m$ (odd) number of $1$s to $m-1$ (even) or $m+1$ (even) number of $1$s.

\begin{proposition} \label{prop:Cn-minimum}
    For $t = 1/2$, $\tilde{\mathsf L}_t$ is  demicube with volume given by:
    \[\vol(\tilde{\mathsf L}_t) = 1-\frac{2^{n-1}}{n!}.\]
\end{proposition}

These results completely define the structure of $\vol(\tilde{\mathsf L}_{t=1/2})/\vol(\tilde{\mathsf N}_t)$.

\begin{proposition} \label{prop:Cn_full}
The ratio of volumes of symmetric slices for $C_n$ is given by the following results.
for $t \in (0, 1/2]$
    \begin{align*}
    &\frac{\vol(\mathsf L(p=(t,\ldots,t),  C_n))}{\vol(\mathsf L(p=(t,\ldots,t),  C_n))} = \frac{\vol(\tilde{\mathsf L}_t)}{\vol(\tilde{\mathsf N}_t)} =\\
    &\qquad \qquad \begin{cases}
     1-\sum_{k=1}^{m} \frac{1}{n!}\binom{n}{n-k}\Big(k-\frac{k-1}{2t}\Big)^n & n-m>1, t \in \Big(\frac{m-1}{2m}, \frac{m+1}{2(m+2)}\Big]\\
    1-\sum_{k=1}^{m} \frac{1}{n!}\binom{n}{n-k}\Big(k-\frac{k-1}{2t}\Big)^n & n-m\leq 1, t \in \Big(\frac{m-1}{2m}, \frac{1}{2}\Big]
    \end{cases}
    \end{align*}
    where $m$ must be an odd number not greater than $n$ and $k$ only takes odd values in the sum.
\end{proposition}

\begin{proposition}\label{prop:cycle-graph-parameters}
    For the cycle graphs $C_n$ ($n \geq 3$), the various volume ratio parameters defined in \cref{sec:random-marginal} are:
    \begin{align*}
        \text{The fall-off value} \quad &\tau(C_n) = \frac{1}{3}\\
        \text{The initial ratio} \quad &\rho_{0+}(C_n) = 1-\frac{1}{(n-1)!}\\
        \text{The middle ratio} \quad &\rho_{1/2}(C_n) = 1-\frac{2^{n-1}}{n!}.
    \end{align*}
\end{proposition}

The ratio $\vol \mathsf L(C_n)/\vol\mathsf N(C_n)$ goes to unity in the $n \to \infty$ as already shown in \cite{lee2020volume}. From \ref{prop:Cn_full}, we see the same happens for the symmetric slices and the ratio $\vol \mathsf L_t/\vol\mathsf N_t$ goes to $1$ as $n \to \infty$. We plot this volume ratio for $C_n$ for different values of $n$ in \cref{fig:Cn}. 
\begin{figure}
    \centering
    \includegraphics[width=0.8\textwidth]{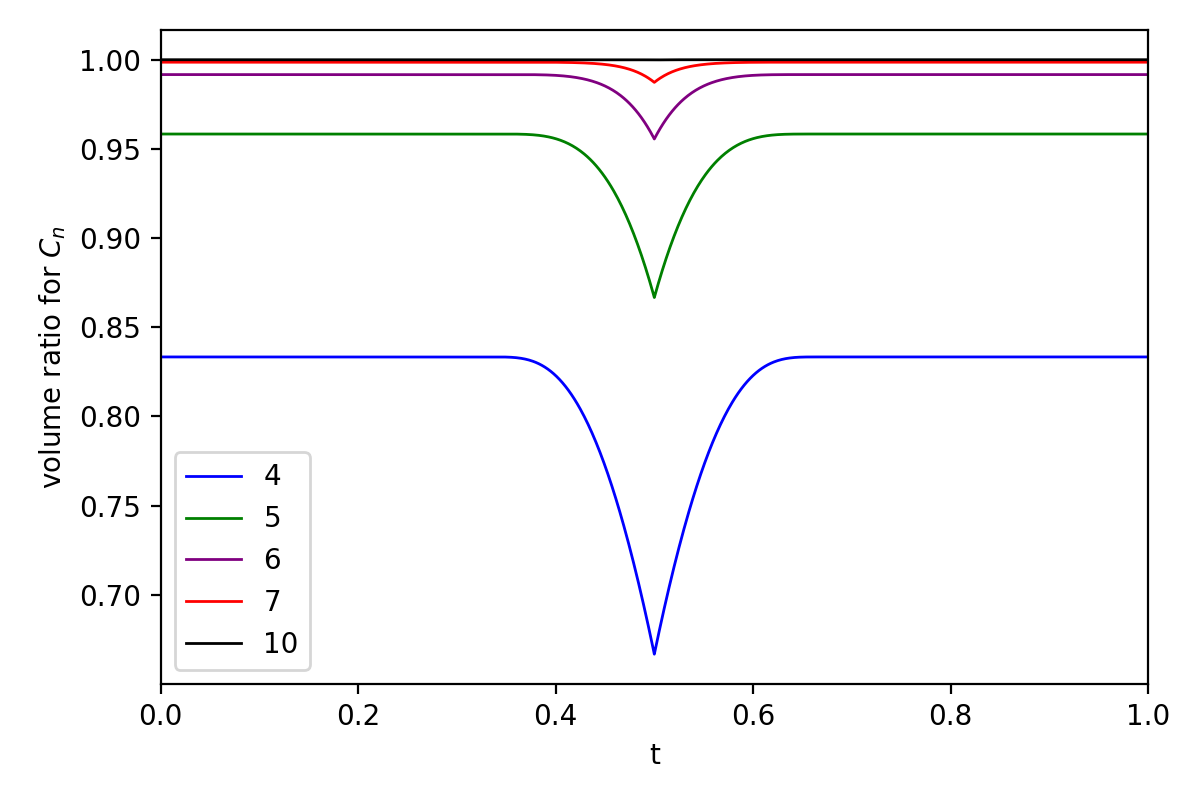}
    \caption{The volume ratio $\vol \mathsf L_t/\vol \mathsf N_t$ plotted as a function of $t$ for $C_n$. The legend shows the value of $n$ for each plot. We see that as $n$ increases, the volume ratio tends to $1$.}
    \label{fig:Cn}
\end{figure}

\section{The complete graph on four vertices}\label{sec:K4}

We now shift our attention to $K_4$, the complete graph on four vertices, see \cref{fig:K4-graph}. In \cref{sec:FM}, we shall see that the inequalities for any graph with $|V|$ vertices can be obtained from the inequalities of $K_{|V|}$, hence the importance of the study of complete graphs. 

\begin{figure}[htb]
  \centering
  \begin{tikzpicture}[auto, scale = 1.4]

    \tikzstyle{vertex1}=[circle, draw=blue, fill=blue!10!, ultra thick]
    \tikzstyle{edge1}=[draw=black, ultra thick]
    \tikzstyle{plus}=[]

    \node[vertex1] (v1) at (-1.5, 1.5) {$p_1$};
    \node[vertex1] (v2) at (1.5, 1.5) {$p_2$};
    \node[vertex1] (v3) at (1.5, -1.5) {$p_3$};
    \node[vertex1] (v4) at (-1.5, -1.5) {$p_4$};

    \draw[edge1] (v1) edge node{$q_{12}$} (v2);
    \draw[edge1] (v2) edge node{$q_{23}$} (v3);
    \draw[edge1] (v3) edge node{$q_{34}$} (v4);
    \draw[edge1] (v1) edge node{$q_{14}$} (v4);
    \draw[edge1] (v2) edge node{} (v4);
    \draw[edge1] (v1) edge node{} (v3);

    \coordinate [distance=0cm,label={$q_{13}$}] (O) at (0.4,-0.95);
    \coordinate [distance=0cm,label={$q_{24}$}] (O) at (0.4, 0.6);

    \end{tikzpicture}
  \caption{The $K_4$ graph}
  \label{fig:K4-graph}
\end{figure}
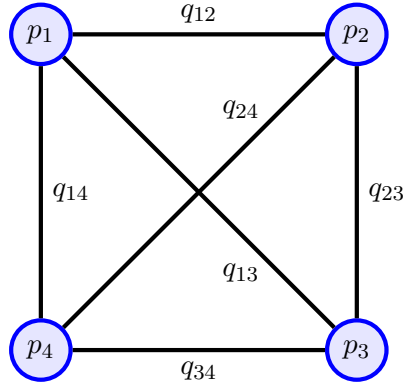

Before we start, let us recall the \emph{Inclusion-Exclusion} inequalities. For a set of events $\{A_i \, : \,  1\leq i\leq n\}$ we have associated probabilities $\{p_i = P(A_i) \, : \,  1\leq i\leq n\}$ and, more generally, 
$$\{p_{i_1i_2\cdots i_k} = P(A_{i_1}\cup A_{i_2}\cup\ldots A_{i_k}) \, : \,  1\leq i_1,\ldots,  i_k \leq n, N\leq n\}$$ as well as  
$$\{q_{i_1i_2\cdots i_k} = P(A_{i_1}\cap A_{i_2}\cap\ldots A_{i_k}) \, : \,  1\leq i_1,\ldots,  i_k \leq n\}.$$
Then the Inclusion-Exclusion inequalities are given by:
\begin{equation}
    p_{i_1i_2\cdots i_k} = \sum_{r=1}^N (-1)^{r+1}\sum_{j_1, \ldots,  j_r \in \{i_1, \ldots, i_k\}} q_{j_1\cdots j_r} \geq 0,
\end{equation}
with $q_i = p_i$ in the expansion. Now, since in our case, we do not deal with hypergraphs, we are limited to an intersection of maximum two events (recall that contexts are of maximum size 2). Thus, we modify the Inclusion-Exclusion inequality in our case as :
\begin{equation}
    p_{i_1i_2\cdots i_k} = \sum_{r=1}^2 (-1)^{r+1}\sum_{j_1, \ldots, j_r \in \{i_1, \ldots,  i_k\}} q_{j_1\cdots j_r} \geq 0.
\end{equation}
We start be recalling a general result about complete graphs.

\begin{proposition}\label{k4:incl-excl}
    \cite[Theorem 2.2]{pitowsky1991correlation} For any graph complete graph $K_{|V|}$ with $|V|\geq 2$, the Inclusion-Exclusion inequalities form facets of $\mathsf L(K_{|V|})$.
\end{proposition}

These are not all the facets of $K_{|V|}$ in general as pointed in \cite{pitowsky1991correlation} with the Chung inequalities \cite{chung1941probability} being an example of other possible facets.

\begin{proposition} \label{prop:k4-ineq}
    For $K_{4}$, the Inclusion-Exclusion inequalities along with the no-disturbance inequalities give all the facets.
\end{proposition}

This was conjectured by Pitowsky himself and we have checked this with \texttt{cdd} \cite{cddlib}. Hence, once we have all the inequalities in the $H$-representation of $\mathsf L(K_4)$ from \cref{prop:k4-ineq}, we can start studying our sliced polytopes.

Let us list down the inequalities of $\mathsf L(K_4, p = (t, t, t, t))$ ordering them on the basis of the value of $t$, they become active at. In what proceeds, \(i, j, k, l \in \{1, 2, 3, 4\}\) are distinct.

The inequalities cutting in from $t=0$ are:
\begin{equation}\label{eq:K4-nb}
\begin{aligned}
    -q_{ij} \leq 0 \\
    q_{ij} \leq t \\
\end{aligned}
\end{equation}
which we get from the no-disturbance condition. We also have the inequality \(-p_{ij} \leq 1-2t\) but that is not relevant for $0 \leq t \leq 1/2$. The remaining inequalities can be obtained from the Inclusion-Exclusion conditions.
\begin{equation}\label{eq:K4-0}
\begin{aligned}
    - q_{jk} + q_{ij} + q_{ik} \leq t \\
    q_{ij} +q_{ik} + q_{il} -q_{jk} - q_{jl} - q_{kl} \leq t \\
    - q_{ij} - q_{kl} + q_{ik} + q_{il} +q_{jk} + q_{jl} \leq 2t
\end{aligned}
\end{equation}

The inequalities cutting in from $t=1/4$ are:
\begin{equation} \label{eq:K4-1/4}
\begin{aligned}
    -q_{ij}-q_{ik}-q_{il}-q_{jk}-q_{jl}-q_{kl}\leq 1-4t \\
      - q_{ij} -q_{ik} -q_{jk} + q_{jl} + q_{il}+q_{kl} \leq 1-t
\end{aligned}
\end{equation}

The inequalities cutting in from $t=1/3$ are:
\begin{equation} \label{eq:K4-1/3}
\begin{aligned}
    -q_{ij} - q_{ik} -q_{jk} \leq 1-3t
\end{aligned}
\end{equation}

The inequalities cutting in from $t=3/8$ are:
\begin{equation} \label{eq:K4-3/8}
\begin{aligned}
    - q_{ij} -q_{ik} -q_{jk} - q_{jl} - q_{il}-q_{kl} \leq 3-8t
\end{aligned}
\end{equation}

We define the polytopes of interest: 

\begin{align*}
    \mathsf N_t &:= \mathsf N(p=(t,t,t,t), K_4) = \{(q_{12},q_{23},q_{34}, q_{14}, q_{13}, q_{24}) \in \mathbb R^6 \, : \, \text{ \cref{eq:K4-nb} holds}\}\\
    \mathsf L_t &:= \mathsf L(p=(t,t,t,t), K_4) = \{(q_{12},q_{23},q_{34}, q_{14}, q_{13}, q_{24}) \in \mathbb R^6 \, : \, \text{ \cref{eq:K4-nb}-\eqref{eq:K4-3/8} hold}\}.
\end{align*}

Let us now look at the various parameters for $K_4$.

\begin{proposition}
    Since, the first inequalities cutting in at a non-zero $t$ value are given by the inequalities \ref{eq:K4-1/4}, the fall of value for $K_4$ is given by:
    $$\tau(K_4) = \frac{1}{4}.$$
\end{proposition}

In the region $0 \leq t \leq 1/4$, the relevant inequalities in the $H$-representation of $\mathsf L_t$ are \ref{eq:K4-nb} and \ref{eq:K4-0}. The corresponding $V$-representation is:
\begin{gather*}
        \operatorname{conv}\Big\{(0, 0, 0, 0, 0, 0), (t, t, 0, t, 0, 0), (t, 0, t, 0, t, 0), (0, t, t, 0, 0, t), (0, 0, 0, t, t, t)\\
        (0, t, 0, 0, t, 0), (0, 0, t, t, 0, 0), (t, 0, 0, 0, 0, t), (t, 0, 0, 0, 0, 0), (0, t, 0, 0, 0, 0) \\
        (0, 0, t, 0, 0, 0), (0, 0, 0, t, 0, 0), (0, 0, 0, 0, t, 0), (0, 0, 0, 0, 0, t), (t, t, t, t, t, t)\Big\}
\end{gather*}

The volume of $N_t$ is simply $t^6$. The volume of $L_t$ can be computed (numerically, using \texttt{cdd}) from the convex hull.

\begin{proposition}
    The initial ratio for $K_4$ is given by :
    $$\rho_{0+}(K_4) = \frac{5}{36}.$$
\end{proposition}

Finally, for $t=1/2$, the $H$-representation of $L_t$ has all the inequalities in \eqref{eq:K4-nb}-\eqref{eq:K4-3/8}. The $V$-representation reads:
\begin{align*}
        \operatorname{conv}\Big\{ &(0.5, 0, 0, 0, 0, 0.5), (0, 0.5, 0, 0, 0.5, 0), (0, 0, 0.5, 0.5, 0, 0), (0.5, 0.5, 0, 0.5, 0, 0), \\ 
        &(0.5, 0, 0.5, 0, 0.5, 0), (0.5, 0.5, 0, 0.5, 0, 0), (0.5, 0, 0.5, 0, 0.5, 0), (0, 0.5, 0.5, 0, 0, 0.5), \\
        &(0, 0, 0, 0.5, 0.5, 0.5), (0.5, 0.5, 0.5, 0.5, 0.5, 0.5)\Big\}.
\end{align*}

The volume of this region is $1/1440$ while the volume of $N_t$ is just $1/2^6$.

\begin{proposition}
    The middle ratio for $K_4$ is given by:
    $$\rho_{1/2}(K_4) = \frac{2}{45}.$$
\end{proposition}

Finally, we provide the computationally obtained plot for $t$ against the volume ratio in \cref{fig:4-vertices}.

\section{Operations on general graphs} \label{sec:FM}

To obtain the $H$-representation of the local polytope $\mathsf L(G) = \mathsf{COR}(G)$, the obvious method is to list down the $V$-representation using the truth-table approach from \cref{def:correlation-polytope} and solve the convex hull problem. However, there are other approaches to getting the $H$-representation of the local polytope as shown in \cite{budroni2012bell}. For the sake of completeness, we will mention them.

Consider the $V$-representation of the local polytope $\mathsf L(G)$ corresponding to a graph $G(V, E)$. Removing an edge, say \(\{ij\}\) from $E$ is equivalent to removing the corresponding column from the truth-table. Hence, the local polytope $\mathsf L(G')$ corresponding to $G'(V, E\setminus\{ij\})$ is just the projection of $\mathsf L(G)$ from $\mathbb R^{|V|+|E|}$ onto  $\mathbb R^{|V|+|E\setminus\{ij\}|}$. This can be achieved by \emph{Fourier-Motzkin elimination}.

\begin{proposition} \label{prop:collapsing-fm}
    The $H$-representation of the local polytope for the graph $G'(V, E\setminus\{ij\})$ can be obtained by applying Fourier-Motzkin elimination to remove $q_{ij}$ from the $H$-representation of the local polytope of the graph $G(V, E)$, and then throwing away the redundant inequalities.
\end{proposition}

As an example, let us derive the $H$-representation of $\mathsf L(C_4)$ from the $H$-representation $\mathsf L(K_4-e)$ which is given by:
\begin{equation} \label{eqn:K4-e-N}
\begin{aligned}
    0 \leq q_{ij} \leq \min{(p_{i}, p_{j})} \\
    p_i + p_j -p_{ij} \leq 1
\end{aligned}
\end{equation}

\begin{subequations} \label{eqn:K4-e-Lneg}
\begin{align}
\label{eq:K4-e-Lneg-1}    p_1 + p_2 + p_3 - q_{12} - q_{23} - q_{13} \leq 1\\
\label{eq:K4-e-Lneg-2}    p_1 + p_3 + p_4 - q_{14} - q_{34} - q_{13} \leq 1\\
\label{eq:K4-e-Lneg-3}    -p_2 + q_{12} + q_{23} - q_{13} \leq 0\\
\label{eq:K4-e-Lneg-4}    -p_4 + q_{14} + q_{34} - q_{13} \leq 0 
\end{align}
\end{subequations}

\begin{subequations} \label{eqn:K4-e-Lpos}
\begin{align}
\label{eq:K4-e-Lpos-1}    -p_3 - q_{12} + q_{23} + q_{13} \leq 0\\
\label{eq:K4-e-Lpos-2}    -p_1 - q_{23} + q_{12} + q_{13} \leq 0\\
\label{eq:K4-e-Lpos-3}    -p_3 - q_{14} + q_{34} + q_{13} \leq 0\\
\label{eq:K4-e-Lpos-4}    -p_1 - q_{34} + q_{14} + q_{13} \leq 0
\end{align}
\end{subequations}

The equations \eqref{eqn:K4-e-N} are the trivial facets corresponding to $\mathsf N(K_4-e)$. Equations \eqref{eqn:K4-e-Lneg} and \eqref{eqn:K4-e-Lpos} are the non-trivial facets for $\mathsf L(K_4-e)$ with negative and positive unity as coefficients of $q_{13}$ respectively. To eliminate $q_{13}$, we just add the opposite signed inequalities. The resulting set of inequalities will have many redundant inequalities which can be removed by checking against a linear program \cite{boyd2004convex}. The minimal set of these inequalities forming the $H$-representation of $\mathsf L(K_4-e)$ are:

\begin{table}[htb]
\begin{center}
\bgroup
\def\arraystretch{1.5}
\begin{tabular}{|cc|c|}
\hline
\rowcolor[HTML]{C0C0C0} 
\multicolumn{2}{|c|}{Inequalities of $\mathsf L(K_4-e)$} & Resultant Inequality for $\mathsf L(C_4)$ \\ 
\hline
\rowcolor[HTML]{EFEFEF} 
\ref{eq:K4-e-Lpos-1} & \ref{eq:K4-e-Lneg-2} & \(p_1+p_4-q_{14}-q_{34}-q_{12}+q_{23} \leq 1\) \\
\rowcolor[HTML]{EFEFEF} 
\ref{eq:K4-e-Lpos-2} & \ref{eq:K4-e-Lneg-4} & \(-p_1-p_4+q_{14}+q_{34}+q_{12}-q_{23} \leq 0\) \\
\ref{eq:K4-e-Lpos-2} & \ref{eq:K4-e-Lneg-2} & \(p_3+p_4-q_{14}-q_{34}+q_{12}-q_{23} \leq 1\) \\
\ref{eq:K4-e-Lpos-1} & \ref{eq:K4-e-Lneg-4} & \(-p_3-p_4+q_{14}+q_{34}-q_{12}+q_{23} \leq 0\) \\
\rowcolor[HTML]{EFEFEF} 
\ref{eq:K4-e-Lpos-3} & \ref{eq:K4-e-Lneg-1} & \(p_1+p_2-q_{14}+q_{34}-q_{12}-q_{23} \leq 1\) \\
\rowcolor[HTML]{EFEFEF} 
\ref{eq:K4-e-Lpos-4} & \ref{eq:K4-e-Lneg-3} & \(-p_1-p_2+q_{14}-q_{34}+q_{12}+q_{23} \leq 0\) \\
\ref{eq:K4-e-Lpos-4} & \ref{eq:K4-e-Lneg-1} & \(p_2+p_3+q_{14}-q_{34}-q_{12}-q_{23} \leq 1\) \\
\ref{eq:K4-e-Lpos-3} & \ref{eq:K4-e-Lneg-3} & \(-p_2-p_3-q_{14}+q_{34}+q_{12}+q_{23} \leq 0\) \\
\hline
\end{tabular}
\egroup
\end{center}
\caption{The inequalities of $\mathsf L(C_4)$ obtained from applying Fourier-Motzkin on inequalities for $\mathsf L(K_4-e)$. The first column shows the inequalities which are added to obtain the resultant inequalities in column 2}
\label{tbl:fm-C4}
\end{table}

Aside from these $8$ inequalities, the other inequalities in the $H$-representation of $\mathsf L(C_4)$ are the $16$ inequalities in \eqref{eqn:K4-e-N} that do not contain $q_{13}$. Comparing to \cref{eq:K22-N,eq:K22-L}, we have the complete $H$-representation of $\mathsf L(C_4)$.

Thus, from (\ref{prop:collapsing-fm}), if the list of facets of $K_n$ is known, the facets of any subgraph with the same set of vertices can be obtained.

Another approach towards unravelling the $H$-representation for a graph is the technique of gluing smaller graphs to get a larger graph as shown in \cite{budroni2012bell}. For the sake of completeness, we will briefly mention this procedure. 

\begin{proposition}\label{prop:lifting-glue}
    Consider two graphs $G_1(V_1, E_1)$ and $G_2(V_2, E_2)$ with corresponding $H$-representations of the local polytope given by $\mathcal{H}_k(p_i \in V_k, q_{ij}\in E_k)$. Now, let the graph formed by gluing $G_1$ and $G_2$ on some vertices be $G(V_1\cup V_2, E_1 \cup E_2)$, such that the induced subgraph on these common vertices is the identical for both graphs. Then, the $H$-representation of $\mathsf L(G)$ is the just $\mathcal{H}(p_i \in V_1 \cup V_2, q_{ij}\in E_1 \cup E_2) \equiv \mathcal{H}_1(p_i \in V_1, q_{ij}\in E_1) \cup \mathcal{H}_2(p_i \in V_2, q_{ij}\in E_2)$.
\end{proposition}

This is because gluing two graphs in such a way effectively results in a tree structure as shown in (\ref{fig:fmtree}). As long as the two subgraphs admit a probability assignment and the values at the intersection coincide, the overall tree always admits a joint probability assignment. Thus, studying the geometry of the resulting graph involves studying the geometry of the cartesian product of the subgraphs. 
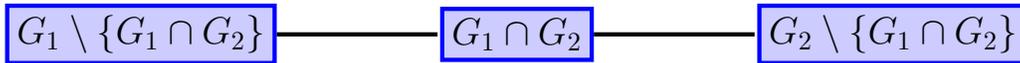
\begin{figure}[htb]
\centering
\begin{tikzpicture}[shorten >=1pt, auto, node distance=3cm, ultra thick,
   node_style/.style={,draw=blue,fill=blue!20!,font=\sffamily\Large\bfseries},
   edge_style/.style={draw=black, ultra thick}]

    \node[node_style] (v1) at (-5,2) {$G_1\setminus\{G_1 \cap G_2\}$};
    \node[node_style] (v2) at (0, 2) {$G_1 \cap G_2$};
    \node[node_style] (v3) at (5, 2) {$G_2\setminus\{G_1 \cap G_2\}$};
    \draw[edge_style]  (v1) edge node{} (v2);
    \draw[edge_style]  (v2) edge node{} (v3);
    \end{tikzpicture}
    \caption{The tree structure generated on gluing graphs}
    \label{fig:fmtree}
\end{figure}

As an example, notice that $K_4-e$ can be formed by gluing two $K_3$ along an edge. Infact, the two sets of inequalities (\ref{eq:K4-e-Lneg-1}, \ref{eq:K4-e-Lneg-3}, \ref{eq:K4-e-Lpos-1}, \ref{eq:K4-e-Lpos-2}) and (\ref{eq:K4-e-Lneg-2}, \ref{eq:K4-e-Lneg-4}, \ref{eq:K4-e-Lpos-3}, \ref{eq:K4-e-Lpos-4}) are isomorphic to (\ref{eq:K3-L}) over some vertex relabellings. Hence, $H$-representation of $\mathsf L(K_4-e)$ can be obtained from the $H$-representation of $\mathsf L(K_3)$.

\begin{center}
\begin{figure}[htb]
\begin{subfigure}{.5\textwidth}
  \centering
  \begin{tikzpicture}[auto]

    \tikzstyle{vertex1}=[circle, draw=blue, fill=blue!10!, ultra thick]
    \tikzstyle{vertex2}=[circle, draw=red, fill=red!10!, ultra thick]
    \tikzstyle{edge1}=[draw=black, thick]
    \tikzstyle{edge2}=[draw=red, thick]
    \tikzstyle{plus}=[]
    
    \node[vertex2] (v1) at (-1.5, 1.5) {$1$};
    \node[vertex1] (v2) at (-1.5, -1.5) {$2$};
    \node[vertex2] (v3) at (1.5, -1.5) {$3$};

    \node[vertex2] (v1') at (0, 1.5) {$\tilde{1}$};
    \node[vertex2] (v3') at (3, -1.5) {$\tilde{3}$};
    \node[vertex1] (v4) at (3, 1.5) {$\tilde{4}$};

    \coordinate [distance=0cm,label={\faPlus}] (O) at (0.75,-0.3);
    
    \draw[edge1]  (v1) edge node{} (v2);
    \draw[edge1]  (v2) edge node{} (v3);
    \draw[edge2]  (v3) edge node{} (v1);

    \draw[edge2]  (v1') edge node{} (v3');
    \draw[edge1]  (v3') edge node{} (v4);
    \draw[edge1]  (v4) edge node{} (v1');
    
    \end{tikzpicture}
  \caption{}
  \label{fig:glue-graph-1}
\end{subfigure}
\begin{subfigure}{.5\textwidth}
  \centering
  \begin{tikzpicture}[auto]

    \tikzstyle{vertex1}=[circle, draw=blue, fill=blue!10!, ultra thick]
    \tikzstyle{vertex2}=[circle, draw=red, fill=red!10!, ultra thick]
    \tikzstyle{edge1}=[draw=black, thick]
    \tikzstyle{edge2}=[draw=red, thick]
    \tikzstyle{plus}=[]
    
    \node[vertex2] (v1) at (-1.5, 1.5) {$1$};
    \node[vertex1] (v2) at (-1.5, -1.5) {$2$};
    \node[vertex2] (v3) at (1.5, -1.5) {$3$};
    \node[vertex1] (v4) at (1.5, 1.5) {$4$};
    
    \draw[edge1]  (v1) edge node{} (v2);
    \draw[edge1]  (v2) edge node{} (v3);
    \draw[edge1]  (v3) edge node{} (v4);
    \draw[edge1]  (v4) edge node{} (v1);
    \draw[edge2]  (v1) edge node{} (v3);
    
    \end{tikzpicture}
  \caption{}
  \label{fig:glue-graph-2}
\end{subfigure}
\caption{$K_4 - e$ formed by joining two $K_3$ as shown in (A). The red components in (B) show the common induced subgraph.}
\label{fig:glue-graph}
\end{figure}
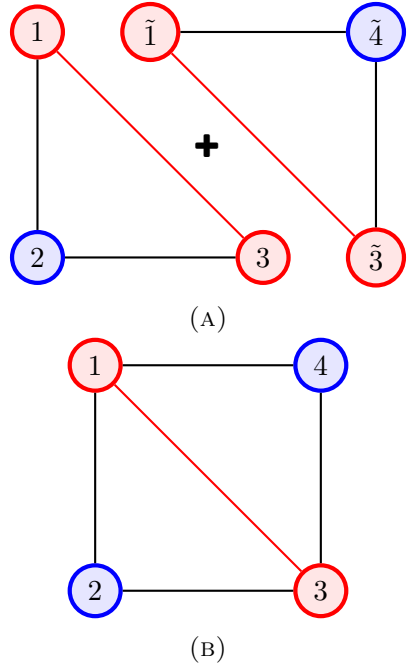  
\end{center}

A direct consequence of (\ref{prop:collapsing-fm}) and (\ref{prop:lifting-glue}) is that all the non-trivial inequalities (\ref{eq:araujo-X}) for $\mathsf L(C_n)$ can be obtained just by iteratively gluing together $K_3$ along an edge and removing it as demonstrated in \cite{araujo2013all}.

Lets start by studying, how \cref{prop:collapsing-fm} and \cref{prop:lifting-glue} affect the volume ration properties defined at the end of \cref{sec:random-marginal}.

\begin{lemma} \label{prop:tau-fm}
    Consider a graph $G'(V, E\setminus\{ij\})$ obtained by removing an edge from the graph $G(V, E)$. Then, $\tau(G')\geq \tau(G)$.
\end{lemma}
\begin{proof}
    Consider an inequality labeled by $k$ to be written in the form:
\begin{equation*}
    [v_1 \quad v_2 \quad\ldots\quad v_{|V|} \quad w_1 \quad w_2 \quad \ldots \quad w_{|E|}] \times [p_i \quad q_{ij}]^\dag \leq C_k
\end{equation*}

Let $m_k$ be the number of positive coefficients $v_i$ and $n_k$ be the number of positive coefficients $w_i$. If this inequality forms the $H$-representation of $\mathsf L(G(V, E))$, then for the symmetric slices $p_i = t \quad \forall i\in V$, the value of $t$ at which the inequality becomes active is given by $C_k/(m_k + n_k)$. 

Now, consider adding two inequalities labelled by $1$ and $2$, to obtain a new inequality labelled by $3$. Then,
\begin{equation}
\begin{aligned}
    m_3 \leq m_1 + m_2 \\
    n_3 \leq n_1 + n_2 \\
    C_3 = C_1+C_2
\end{aligned}
\end{equation}

This new inequality becomes active at $C_3/(m_3 + n_3)$ and we have,
\begin{equation}
    \min \bigg\{\frac{C_1}{m_1+n_1}, \frac{C_2}{m_2+n_2}\bigg\} \leq \frac{C_1+C_2}{m_1+m_2+n_1+n_2} \leq \frac{C_3}{m_3 + n_3}
\end{equation}

Thus, any inequality obtained as a result of applying Fourier-Motzkin elimination becomes active at $t$ larger than its constituents and (\ref{prop:tau-fm}) is implied.
\end{proof}

\begin{proposition}
    For any graph $G(V, E)$,
    \begin{equation*}
        \tau(G) \geq \tau(K_{|V|})
    \end{equation*}
\end{proposition}

This follows from (\ref{prop:tau-fm}) as any graph with $|V|$ vertices can be formed by iteratively removing edges from $K_{|V|}$. 

\begin{proposition} \label{prop:glue-vertex}
    Consider a graph $G(V_1 \cup V_2, E_1 \cup E_2)$ formed by gluing two graphs $G_1(V_1, E_1)$ and $G_2(V_2, E_2)$. Then,
    \begin{equation*}
        \tau(G) = \min\{\tau(G_1), \tau(G_2)\}
    \end{equation*}
\end{proposition}

This follows directly from (\ref{prop:lifting-glue}) as the set of inequalities in $H$-representation of $G$ is just the union of set of inequalities in the $H$-representation of $G_1$ and $G_2$.

Finally in this section we would like to make a conjecture about fall-off value based on the graphs we have studied (see \cref{appendix}) and add concluding remarks section.
\begin{conjecture} \label{cnj:tw}
    For a graph G of treewidth $\tw({G})$, the fall-off value is given by,
    $$\tau(G) = \frac{1}{\tw(G) + 1}$$
\end{conjecture}

\begin{proposition}\label{prop:incl-excl}
    The facet defining Inclusion-Exclusion inequalities cannot falsify \cref{cnj:tw}.
\end{proposition}

\begin{proof}
    Consider a complete graph with $N$ vertices with $p_i, i \in \{1,\ldots, N\}$ being the probabilities of the corresponding events. Consider an Inclusion-Exclusion inequality for the graph which is formed by the union of $n$ events $p_i, i \in \{1,\ldots, n\}$ and $k$ complements $p_{\bar{j}} = 1 - p_j, j \in \{n+1,\ldots, N\}$ ($n+k = N$). Such an inequality will be of the form:
    \begin{align*}
        \sum_{i=1}^n p_i + \sum_{j=n+1}^N p_{\bar{j}} - \sum_{i=1}^n\sum_{j=i+1}^n q_{ij} - \sum_{i=1}^n\sum_{j=n+1}^N q_{i\bar{j}} - \sum_{i=n+1}^N\sum_{j=i+1}^N q_{\bar{i}\bar{j}} \leq 1
    \end{align*}
    where $q_{ij}$ is the probability of the intersection of events corresponding to $p_i$ and $p_j$. Fixing our marginals $p_i = t$ for all $i \in \{1,\ldots, N\}$, this inequality becomes:
    \begin{align*}
        \sum_{i=1}^n t + \sum_{j=n+1}^N (1-t) - \sum_{i=1}^n\sum_{j=i+1}^n q_{ij} - \sum_{i=1}^n\sum_{j=n+1}^N (t-q_{ij}) - \sum_{i=n+1}^N\sum_{j=i+1}^N (1-2t+q_{ij}) \leq 1 \\
        nt + k(1-t) - \sum_{i=1}^n\sum_{j=i+1}^n q_{ij} - nkt +  \sum_{i=1}^n\sum_{j=n+1}^N q_{ij} - \binom{k}{2}(1-2t) - \sum_{i=n+1}^N\sum_{j=i+1}^N q_{ij} \leq 1\\
        - \sum_{i=1}^n\sum_{j=i+1}^n q_{ij} - \sum_{i=n+1}^N\sum_{j=i+1}^N q_{ij} + \sum_{i=1}^n\sum_{j=n+1}^N q_{ij}\leq \frac{(k-1)(k-2)}{2} + [(N+3)k-2k^2-N]t
    \end{align*}

    Now, the inequality becomes active only when the maximum value of left hand side of the inequality becomes equal to the minimum value of the right hand side. The maximum value of $q_{ij} = t$ for $p_i = p_j = t$. Thus, on simplifying the value of $t$ at which the inequality becomes active is given by: 
    \begin{equation}
        t = -\frac{(k-1)(k-2)}{2(3k-k^2-N)}
    \end{equation}

    We see for $k = 1, 2$, $t=0$. Hence, the the corresponding inequalities distinguish the body $\mathsf{L}$ and $\mathsf(N)$ for small $t$. The first inequalities to becomes active at non-zero $t$ occur for $k = 0, 3$. The corresponding value of $t = 1/N$ in accordance to \ref{cnj:tw} for complete graphs.
\end{proof}

\begin{proposition}
    \cref{cnj:tw} holds true for all graphs with 4 or less vertices. 
\end{proposition}

This simply follows from \cref{k4:incl-excl}, \cref{prop:incl-excl} and the fact that removing an edge from a complete graph lowers the treewidth.

\medskip 

Let us now show that \cref{cnj:tw} holds for \emph{series-parallel} graphs, i.e.~graphs that have treewidth 2. Note that the \cref{cnj:tw} is obviously true for forests (graphs having treewidth 1), since in that case $\mathsf L(G) = \mathsf N(G)$, see \cref{prop:forest}.

Our main tool will the following result, characterizing the facet structure of $\mathsf L(G)$, for a series-parallel graph $G$. 

\begin{proposition}[{{\cite[Theorem 10]{padberg1989boolean}}}]\label{prop:facets-series-parallel}
    For any series-parallel graph $G$, \emph{odd-cycle inequalities} define all the facets of the local polytope $\mathsf L(G)$.
\end{proposition}

Let us explain next what are odd-cycle inequalities and how to obtain them; our presentation follows closely \cite[Section 4]{padberg1989boolean}. Let $C$ be a non-trivial, simple cycle of $G$ (seen as a collection of edges) and $M \subseteq C$ a subset of the edges in $C$ of \emph{odd} cardinality $m:=|M|$. If $S\subseteq V$ is the set of vertices incident to the edges in $C$, we define
\begin{align*}
    S_0 &:= \{v \in S \, : \, \exists e \neq f \in M \text{ with } e \cap f = v\}\\
    S_2 &:= \{v \in S \, : \, \exists e \neq f \in C \setminus M \text{ with } e \cap f = v\}\\
    S_1 &:= S \setminus (S_0 \sqcup S_2).
\end{align*}

Then, the \emph{odd-cycle inequality} corresponding to the pair $(C, M)$ is: 
\begin{equation}\label{eq:odd-cycle-ineqaulity}
    \sum_{v\in S_0} p_v - \sum_{v\in S_2} p_v - \sum_{e \in M} q_e + \sum_{e \in C \setminus M} q_e  \leq \left\lfloor \frac m 2 \right\rfloor. 
\end{equation}

Such inequalities are satisfied for all the elements in $\mathsf L(G)$, see \cite{padberg1989boolean}. For example, in the case $G=C=K_3$, taking $M=\{(12)\}$ yields $S_0 = \emptyset$, $S_2 =\{3\}$, and $S_1 =\{1,2\}$, giving the inequality
$$-p_3 -q_{12} + q_{13} + q_{23} \leq 0,$$
which is the last inequality in \cref{eq:K3-L}. The other inequalities can be obtained by varying the subset $M$.  

Let us now consider the correlation slice corresponding to setting $p_v = t$ for all vertices $v$ of a series-parallel graph. \cref{eq:odd-cycle-ineqaulity} reads in this case, after the change of variables $q_e \leftarrow t x_e$,
\begin{equation}\label{eq:odd-cycle-inequality-slice}
    \sum_{e \in M} x_e - \sum_{e \in C \setminus M} x_e \geq |S_0| - |S_2| - \frac{ \lfloor m/2 \rfloor}{t}.
\end{equation}
One can easily see that for $m=1$, the inequality above does not depend on $t$. 

\begin{lemma}
    The inequality \eqref{eq:odd-cycle-inequality-slice} is trivially satisfied for all $t \leq 1/3$.
\end{lemma}
\begin{proof}
    Clearly, the minimum possible value of the left-hand-side of \eqref{eq:odd-cycle-inequality-slice} is $-|C \setminus M|$. One can easily show that for all $C,M$, we have \cite{padberg1989boolean}
    $$|S_0| - |S_2| = |M| - |C \setminus M|.$$
    Hence, \eqref{eq:odd-cycle-inequality-slice} is trivially satisfied whenever
    $$-|C \setminus M| \geq |S_0| - |S_2| - \frac{ \lfloor m/2 \rfloor}{t} \iff t \leq \frac{\lfloor m/2 \rfloor}{m}.$$
    The latter is an increasing function of $m$, attaining its minimum $1/3$ for odd $m$ at $m=3$.
\end{proof}

We have now all the elements to prove one of our main results. 

\begin{theorem}
    \cref{cnj:tw} holds for series-parallel graphs $G$: if $\operatorname{tw}(G)=2$, then $\tau(G) = 1/3$.
\end{theorem}
\begin{proof}
    Let $G$ be a series-parallel graph. By \cref{prop:facets-series-parallel}, we know that the polytope $L(G)$ is described only by odd-cycle inequalities, which, by the previous lemma, are trivially satisfied for $t \leq 1/3$; hence $\tau(G) \geq 1/3$. Moreover, since $\operatorname{tw}(G)=2$, $G$ is not a forest, so it contains a cycle. Consider the \emph{smallest} induced subgraph $H$ of $G$ that contains a cycle. The graph $H$ must actually be a cycle, since the presence of extra edges would contradict its minimality. We have shown in \cref{prop:cycle-graph-parameters} that $\tau(H)=1/3$, so $\mathsf L(H)$ has a non-trivial facet that becomes ``active'' at $t=1/3$. By \cite[Corollary 2]{padberg1989boolean}, the same holds for $G$; hence $\tau(G) \leq 1/3$, finishing the proof. 
\end{proof}

\medskip 

We consider the effect of different operations on graphs on the volume taion of the local and non-signaling polytopes. 

\begin{proposition}
    For two graphs $G_1$ and $G_2$ glued together on a single vertex, the resulting $G$ has, for symmetric slices $p = (t, \ldots, t)$,
    \begin{equation*}
        \frac{\vol \mathsf L(p, G)}{\vol \mathsf N(p, G)} = \frac{\vol \mathsf L(p, G_1)}{\vol \mathsf N(p, G_1)}\times\frac{\vol \mathsf L(p, G_2)}{\vol \mathsf N(p, G_2)}
    \end{equation*}
\end{proposition}

As an example, consider the graph in \cref{fig:K3-glue-K22-1} formed by a product of $K_3$ and $K_{2,2}$. Its volume ratio can be obtained using \cref{prop:K3-vol} and \cref{prop:K22-vol}. This has been plotted against the numerical results in \cref{fig:K3-glue-K22-2}.

\begin{center}
\begin{figure}[htb]
\begin{subfigure}{.25\textwidth}
  \centering
  \resizebox{\columnwidth}{!}{
      \begin{tikzpicture}
        \tikzstyle{vertex1}=[circle, draw=blue, fill=blue!10!, ultra thick]
        \tikzstyle{vertex2}=[circle, draw=red, fill=red!10!, ultra thick]
        \tikzstyle{edge1}=[draw=black, thick]
        \tikzstyle{edge2}=[draw=red, thick]
    
        \node[vertex1] (v1) at (-1.5, 4.5) {$1$};
        \node[vertex1] (v2) at (1.5, 4.5) {$2$};
        \node[vertex2] (v3) at (0, 3.0) {$3$};
        \node[vertex1] (v4) at (1.5, 1.5) {$4$};
        \node[vertex1] (v6) at (-1.5, 1.5) {$6$};
        \node[vertex1] (v5) at (0, 0) {$5$};
    
        \draw[edge1]  (v1) edge node{} (v2);
        \draw[edge1]  (v2) edge node{} (v3);
        \draw[edge1]  (v3) edge node{} (v1);
        \draw[edge1]  (v3) edge node{} (v4);
        \draw[edge1]  (v4) edge node{} (v5);
        \draw[edge1]  (v5) edge node{} (v6);
        \draw[edge1]  (v6) edge node{} (v3);
        
    \end{tikzpicture}
    }
  \caption{Graph formed by joining $K_3$ and $K_{2,2}$ along a vertex}
  \label{fig:K3-glue-K22-1}
\end{subfigure}
\begin{subfigure}{.7\textwidth}
  \centering
    \includegraphics[width=1.2\linewidth]{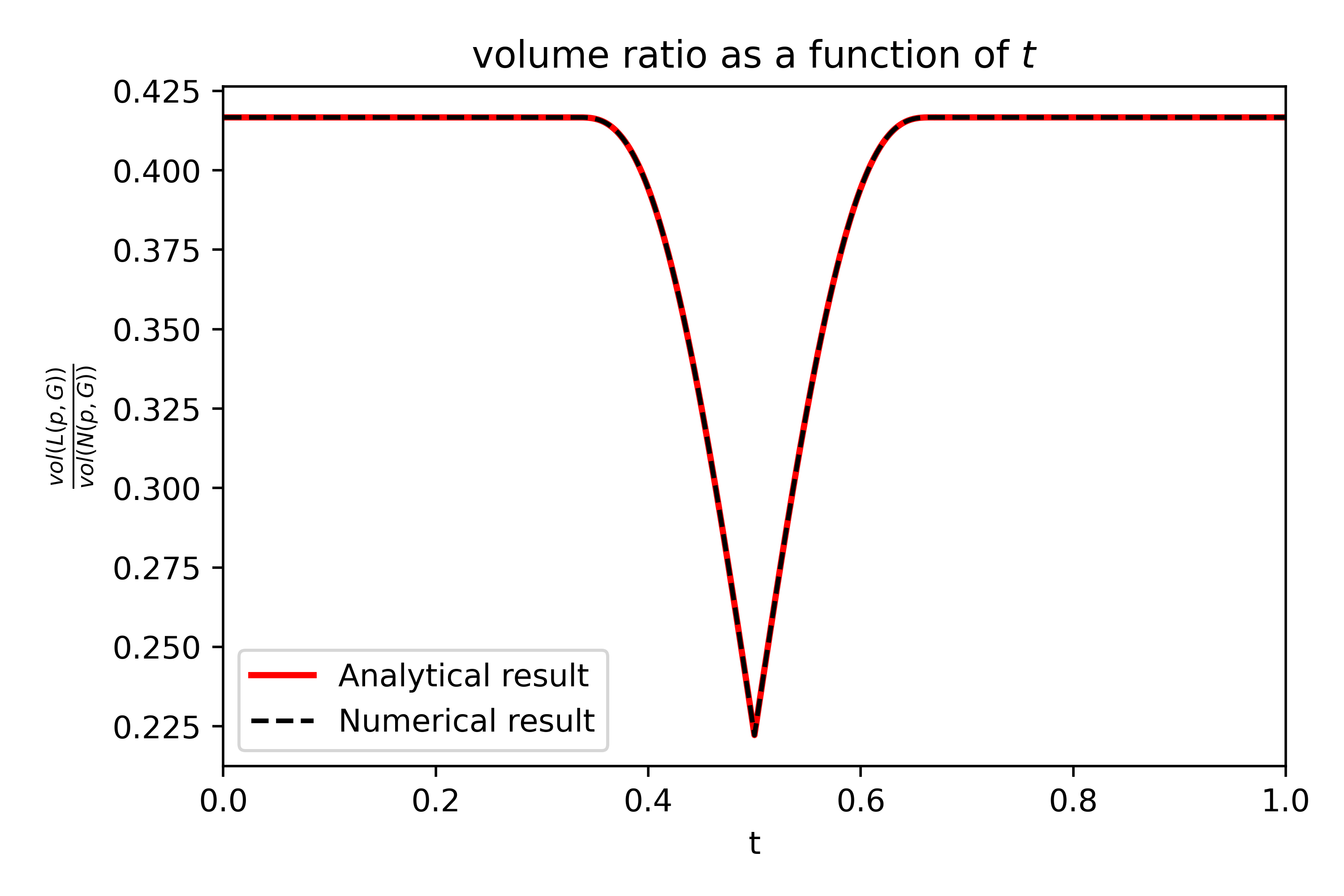}
  \caption{Analytically and numerically calculated ratio of volumes}
  \label{fig:K3-glue-K22-2}
\end{subfigure}
\caption{}
\label{fig:K3-glue-K22}
\end{figure}  
\end{center}

This is because, in terms of the $H$-representation gluing two graphs along a vertex is just taking the combined set of inequalities for both graphs. Note that these two sets would not have any identical inequalities unlike the case where two graphs are glued along an edge. Hence, the resultant polytope is just the prism product of its constituents whose volume is the  product of volumes of its constituents \cite{manning2012fourth, toth2017handbook}.

\begin{proposition}
    Consider a graph $G'$ obtained by gluing a tree to the graph $G$. Then,
    \begin{equation*}
        \frac{\vol \mathsf L(p, G')}{\vol \mathsf N(p, G')} = \frac{\vol \mathsf L(p, G)}{\vol \mathsf N(p, G)}
    \end{equation*}
\end{proposition}
\begin{proof}
    This follows from \cref{prop:glue-vertex} and \cref{prop:tree}.
\end{proof}

\section{Conclusion}\label{sec:conclusion}

Throughout this paper we have studied the geometry of correlation and transportation polytopes by looking at fixed (symmetrical) marginal slices and studying their volume ratios. These bodies are closely related to the non-contextual and no-disturbance polytopes which are of prime interest in Quantum Foundations. 

The characterization of the volume ratio for trees (treewidth = 1) is trivial. We provide a complete analysis for the $K_3$ and $K_{2,2}$ case which famously appear in literature as the Bell-Wigner and CHSH polytopes. We then generalize our analysis to all cyclic graphs. Finally, we also prove theorems on the nature of volume-ratio for all graphs of treewidth two. For graphs with three-width $3$, the example of $K_4$ is worked out explicitly and additional examples are given in the appendices.

Over the course of our analysis, we observe that the volume ratio remains constant for some initial range of the fixed marginal before it starts decreasing and reaches a minimum value. In order to quantify this remarkable fact, we look at parameters such as \emph{fall-off value}, \emph{initial-ratio} and \emph{middle-ratio}. We conjecture that the fall-off value depends on inversely on $\tw{(G)} + 1$. We then show that this holds in general for treewidth $1$ and $2$ and holds specifically for some graphs of higher treewidths. To prove the result for a general graph, we need a general way to characterize all facets of a correlation polytopes (or equivalently a cut-polytope). This is a long-standing problem in literature and results exists only for graphs with up to $7$ vertices \cite{grishukhin1990all}. 

Contextual correlations are of particular importance in quantum information theory as they are required to obtain advantage over classical information processing protocols \cite{schmid2018all,singh2023no,wagner2024inequalities}. Thus, it is important to realize scenarios/games where a random sampling yields a contextual behaviour. Through our work, we have showed that in scenarios where n-parties are involved in a contextual scenario with contexts of maximum size $2$ and the constraint that the single-variable marginal for each party is fixed, the highest chance of getting a contextual correlation with random sampling is possible when the fixed marginals are all $1/2$. What is interesting is fixing the marginals in a clever way can infact increase the odds of getting contextual correlations. For example, consider the CHSH case. When all the probabilities are free, there is a $1$ in $17$ chance, that a sampled behaviour is contextual. However, once the marginals are fixed to $t$, the odds can be reduced to $1/3$ when $t=1/2$. We already saw the significance of the $t=1/2$ slice for the CHSH case in \cref{sec:square}. For any general bell game as well, these slices are of importance as they correspond to the players sharing a locally maximally entangled state often linked to maximal quantum violations of local inequalities such as the mermin inequalities \cite{mermin1990extreme} where the n-GHZ states are used.

There is scope for a lot of future work in this direction. Analysis for larger contexts, and poly-variate correlations would offer better insight one the distribution of useful correlation boxes in general scenario. Moreover, it would be interesting to relate the ratio of volumes we consider in this paper to other measures of non-contextuality for scenarios encoded by graphs, such as the contextual fraction \cite{abramsky2017contextual,kujala2019measures}.

\bigskip

\noindent\textbf{Acknowledgments.} I.N. was supported by the ANR projects \href{https://esquisses.math.cnrs.fr/}{ESQuisses}, grant number ANR-20-CE47-0014-01, and \href{https://www.math.univ-toulouse.fr/~gcebron/STARS.php}{STARS}, grant number ANR-20-CE40-0008, as well as by the PHC program \emph{Star} (Applications of random matrix theory and abstract harmonic analysis to quantum information theory). A.K.J.~received support from the \href{https://www.inde.campusfrance.org/about-us}{French Embassy in India} through the \href{https://www.inde.campusfrance.org/france-excellence-charpak-lab-scholarship}{French Excellence Charpak Lab Scholarship} Programme. 

\bibliography{refs}
\bibliographystyle{alpha}

\appendix \label{appendix}

\section{Graphs and Plots}

In this appendix, we provide numerical computations of the volume ratio for symmetric slices of different graphs.

\subsection{Graphs with 4 vertices}

We start with graphs on 4 vertices.

\begin{figure}[H]
\centering
\begin{minipage}[c]{.49\textwidth}
\includegraphics[scale = 0.49]{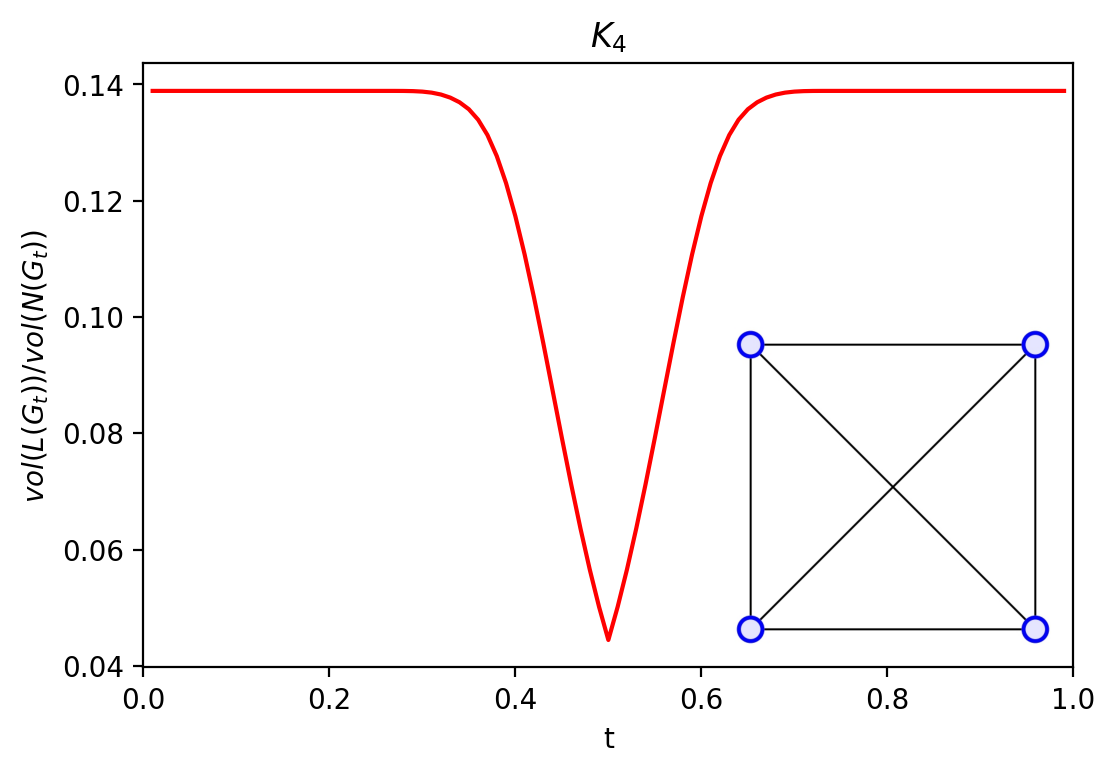}
\includegraphics[scale = 0.49]{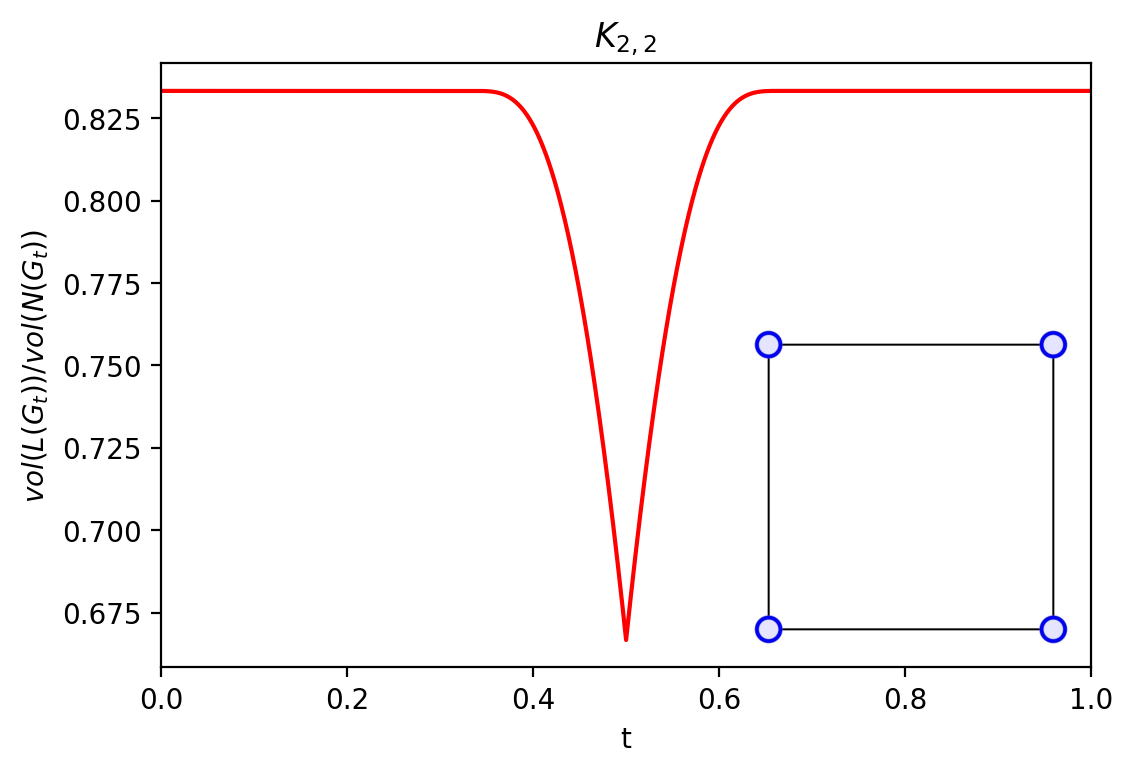}
\end{minipage}\hfill
\begin{minipage}[c]{.45\textwidth}
\includegraphics[scale = 0.49]{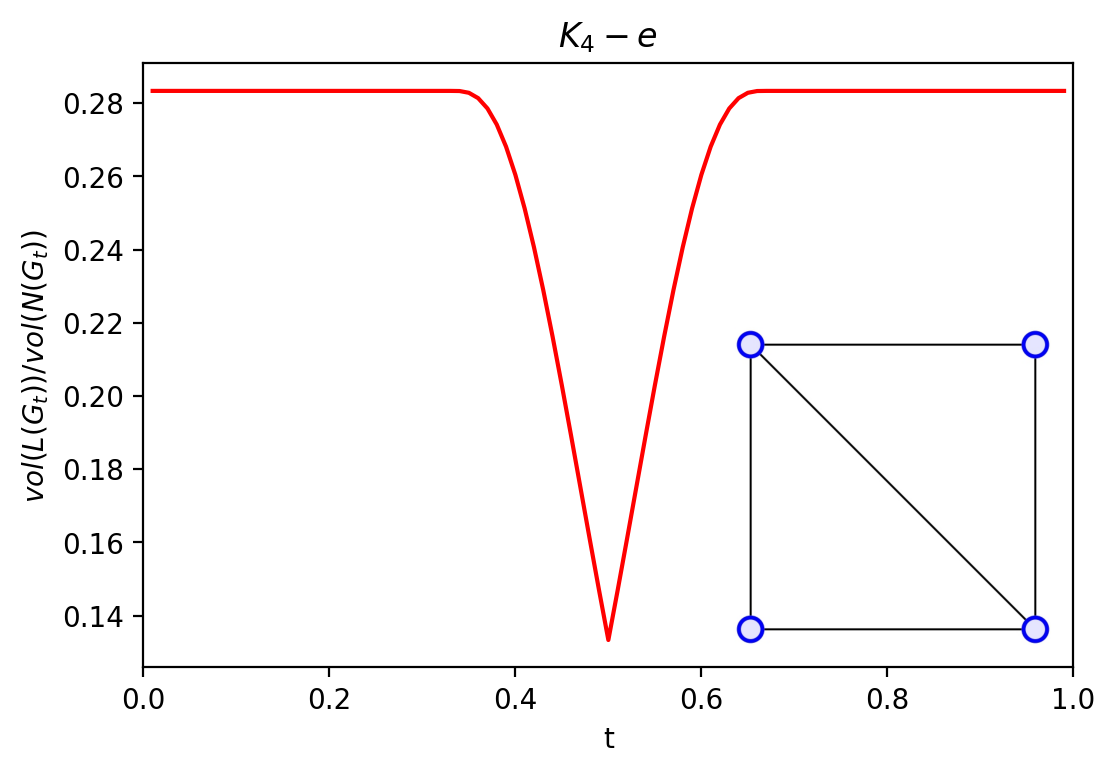}
\includegraphics[scale = 0.49]{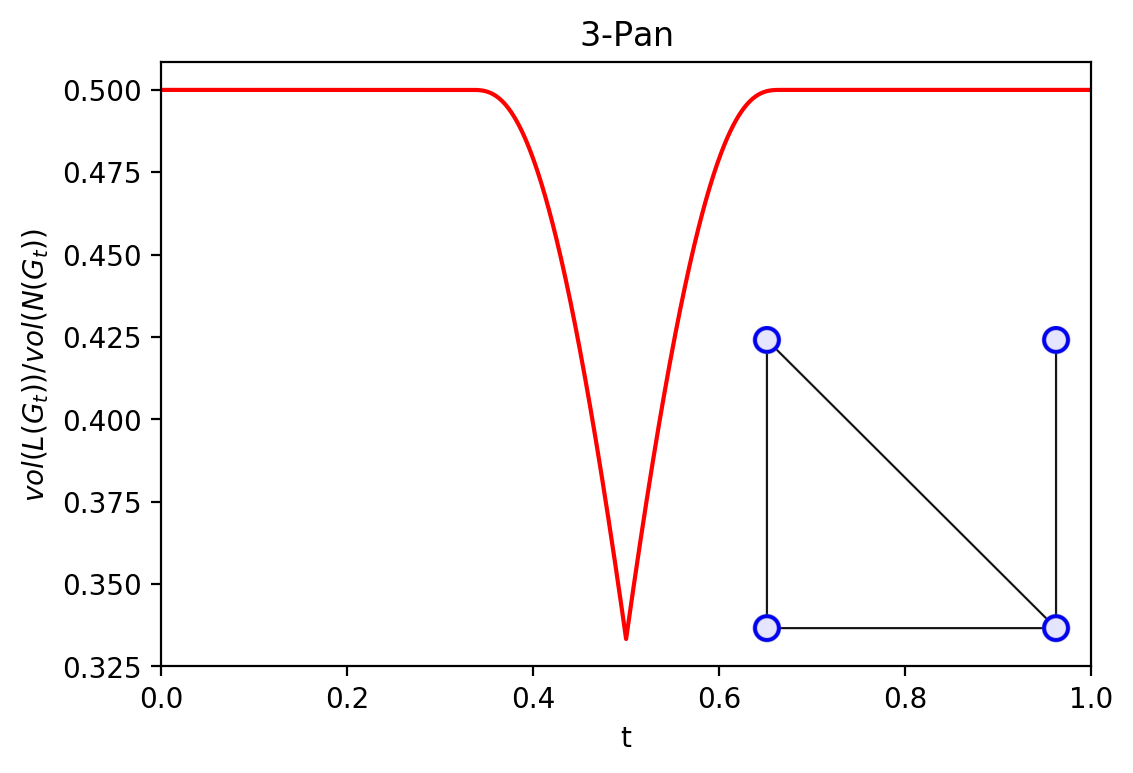}
\end{minipage}
\caption{Plots of volume ratio of symmetric slices given by, $\vol \mathsf L(p=(t,t,t,t), G)/\vol \mathsf N(p=(t,t,t,t), G)$ as a function of $t$ for different graphs with $4$ vertices. The graphs are shown on the lower right corner of each plot.}
\label{fig:4-vertices}
\end{figure}

\begin{table}[htb]
\begin{center}
\bgroup
\def\arraystretch{1.5}
\begin{tabular}{|c|c|c|c|c|}
\hline
\rowcolor[HTML]{C0C0C0} 
Graph $G$ & $\tw(G)$ & $\tau(G)$ & $\rho_{0+}(G)$ & $\rho_{1/2}(G)$\\ 
\hline
$K_4$ & 3 & $\frac{1}{4}$ & $\frac{5}{36}$ & $\frac{2}{45}$\\
\hline
$K_4-e$ & 2 & $\frac{1}{3}$ & $\frac{17}{60}$ & $\frac{2}{15}$\\
\hline
$K_{2,2}$ & 2 & $\frac{1}{3}$ & $\frac{5}{6}$ & $\frac{2}{3}$\\
\hline
pan & 2 & $\frac{1}{3}$ & $\frac{1}{2}$ & $\frac{1}{3}$\\
\hline
\end{tabular}
\egroup
\end{center}
\caption{List of parameters for graphs with 4 vertices: treewidth of the graph $\tw(G)$, the threshold $\tau(G)$ until which the volume ratio is constant, the initial (constant) volume ration $\rho_{0+}(G)$, the volume ration at $p=1/2$, $\rho_{1/2}(G)$.}
\label{tbl:4vertex}
\end{table}

From \cref{tbl:4vertex} we see that \cref{cnj:tw} holds for all graphs with 4 vertices.

\subsection{Graphs with 5 vertices}

Next, we give the numerical results for graphs with 5 vertices.

\begin{figure}[htb]
\centering
\begin{minipage}[c]{.49\textwidth}
\includegraphics[scale = 0.49]{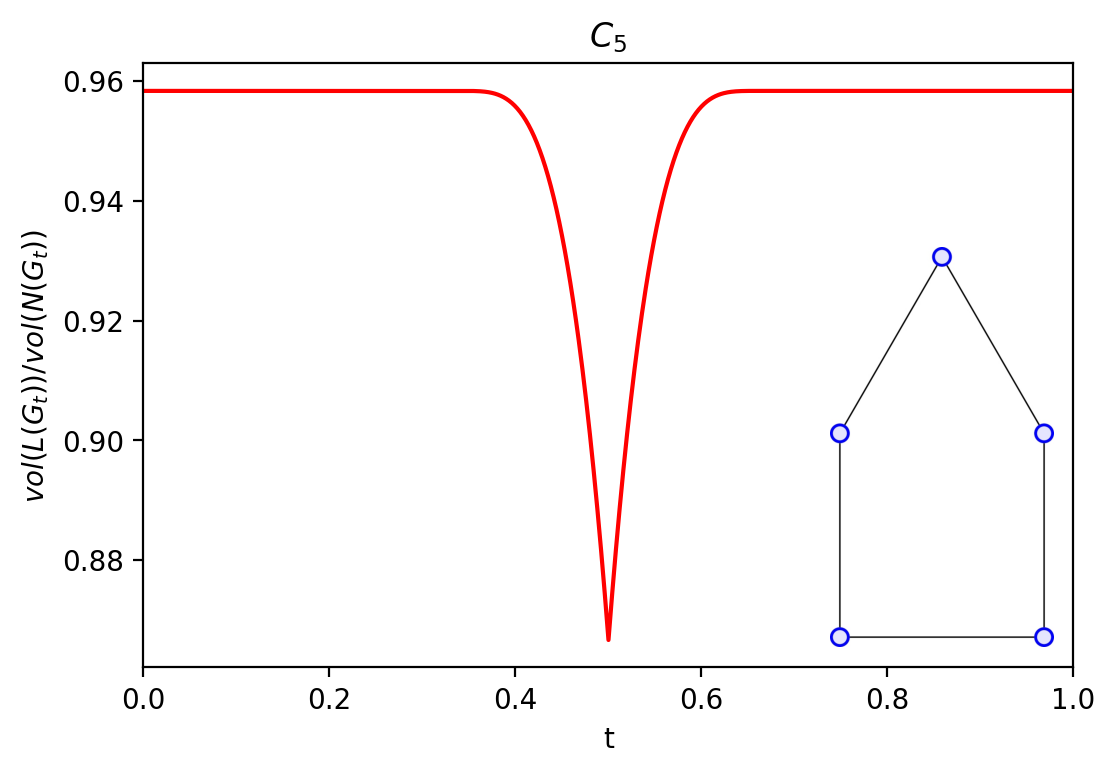}
\includegraphics[scale = 0.49]{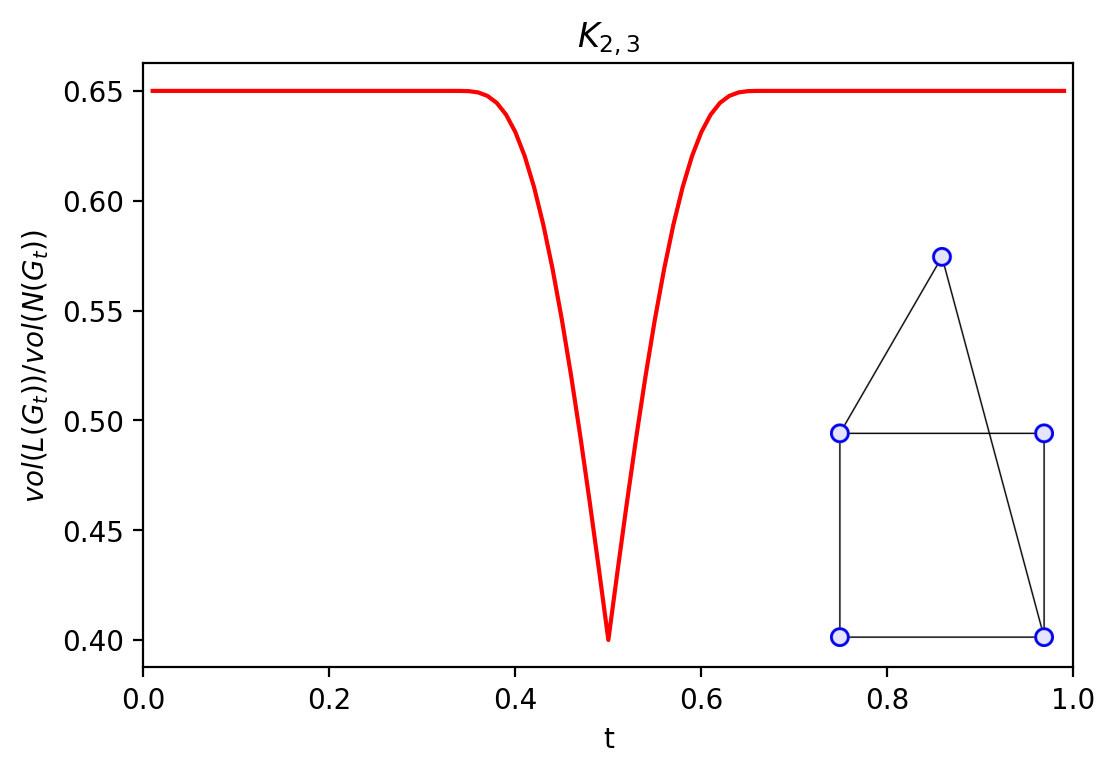}
\includegraphics[scale = 0.49]{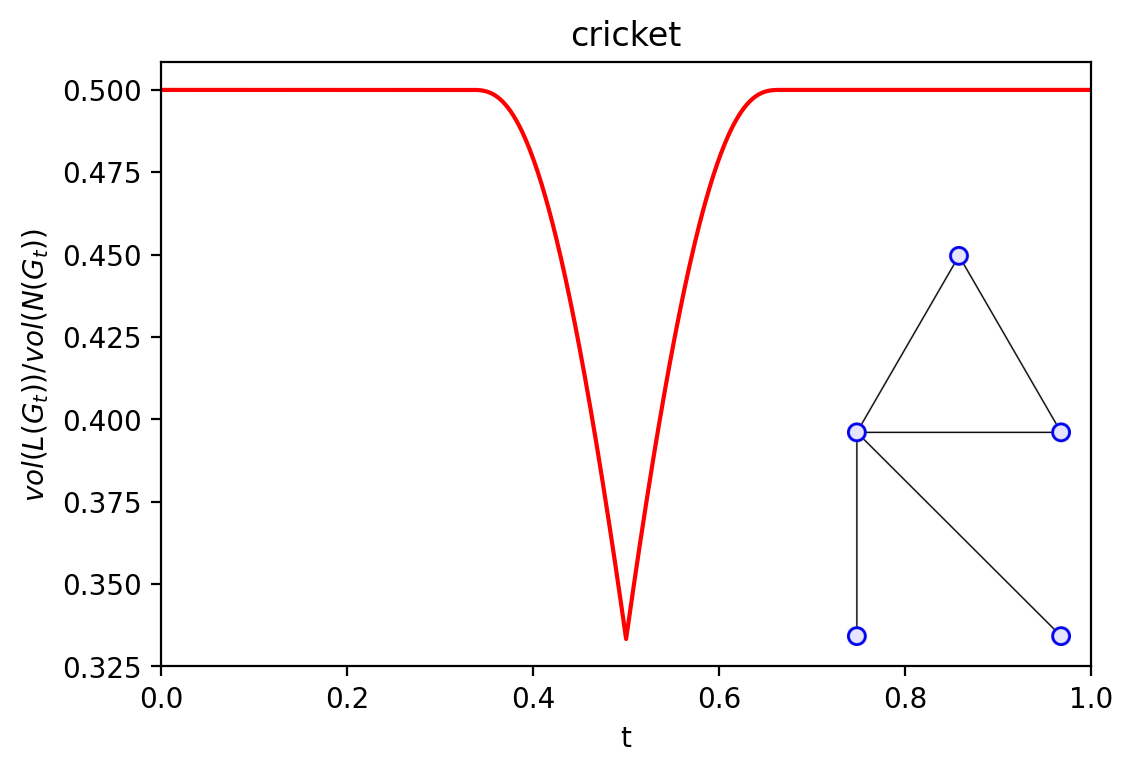}
\includegraphics[scale = 0.49]{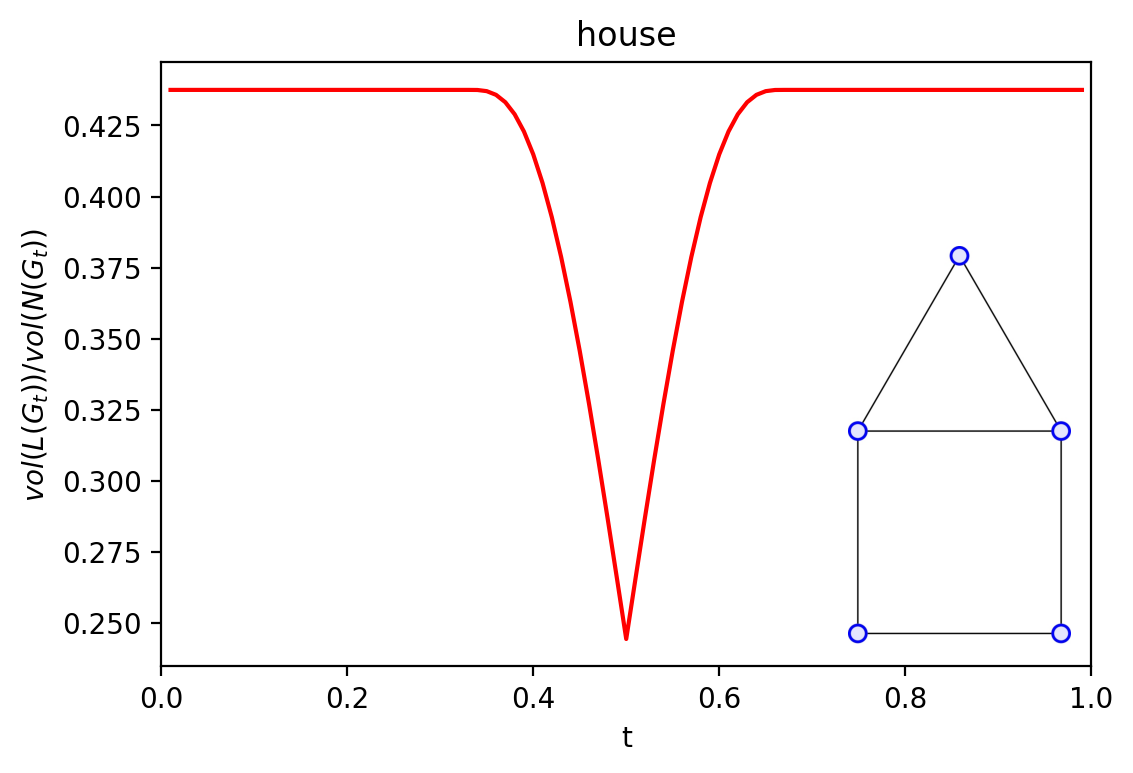}
\includegraphics[scale = 0.49]{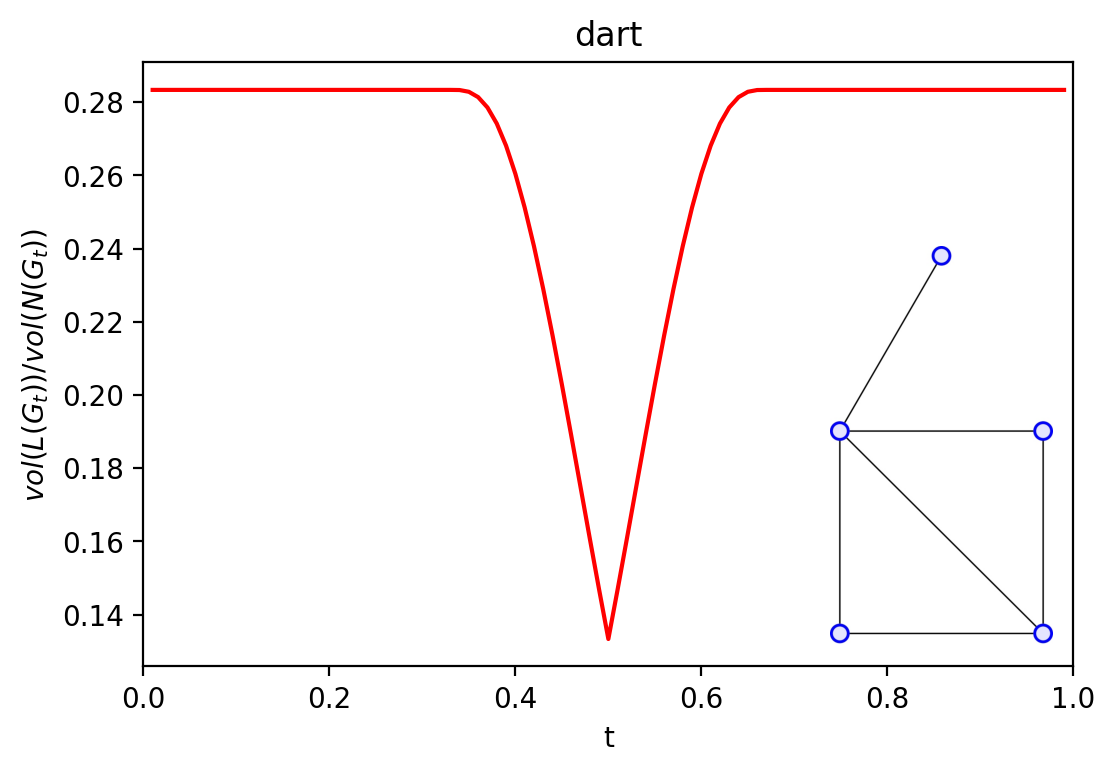}
\end{minipage}\hfill
\begin{minipage}[c]{.45\textwidth}
\includegraphics[scale = 0.49]{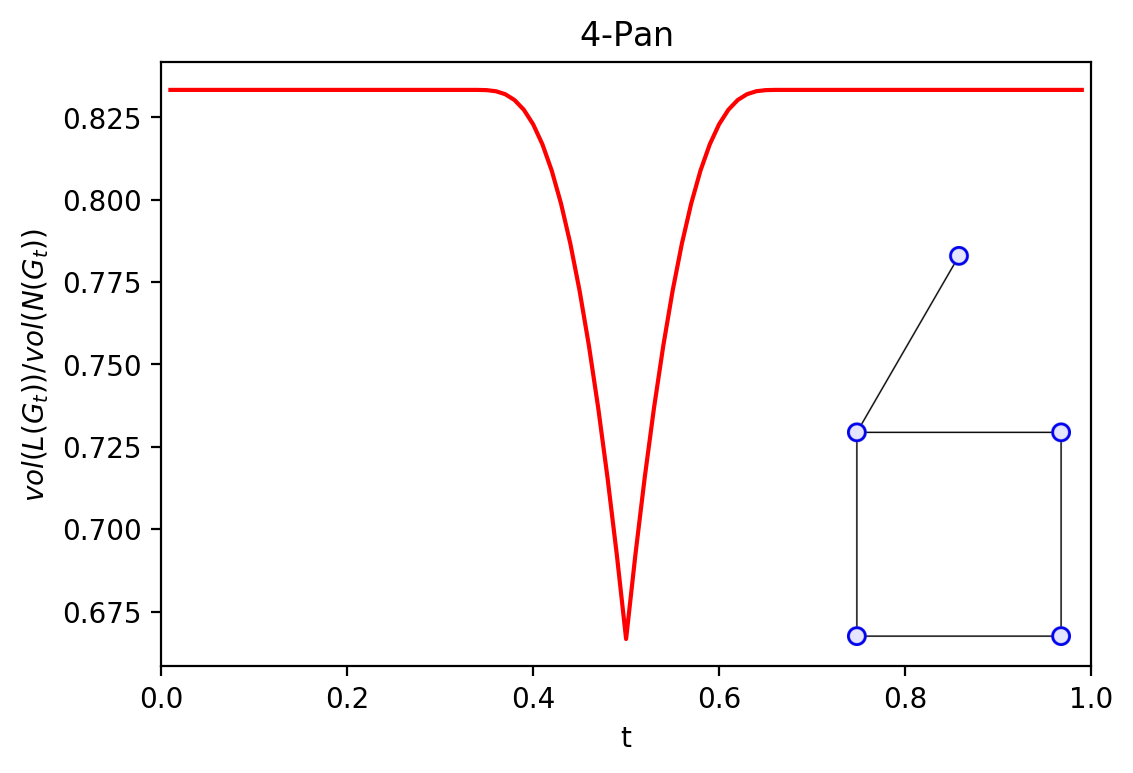}
\includegraphics[scale = 0.49]{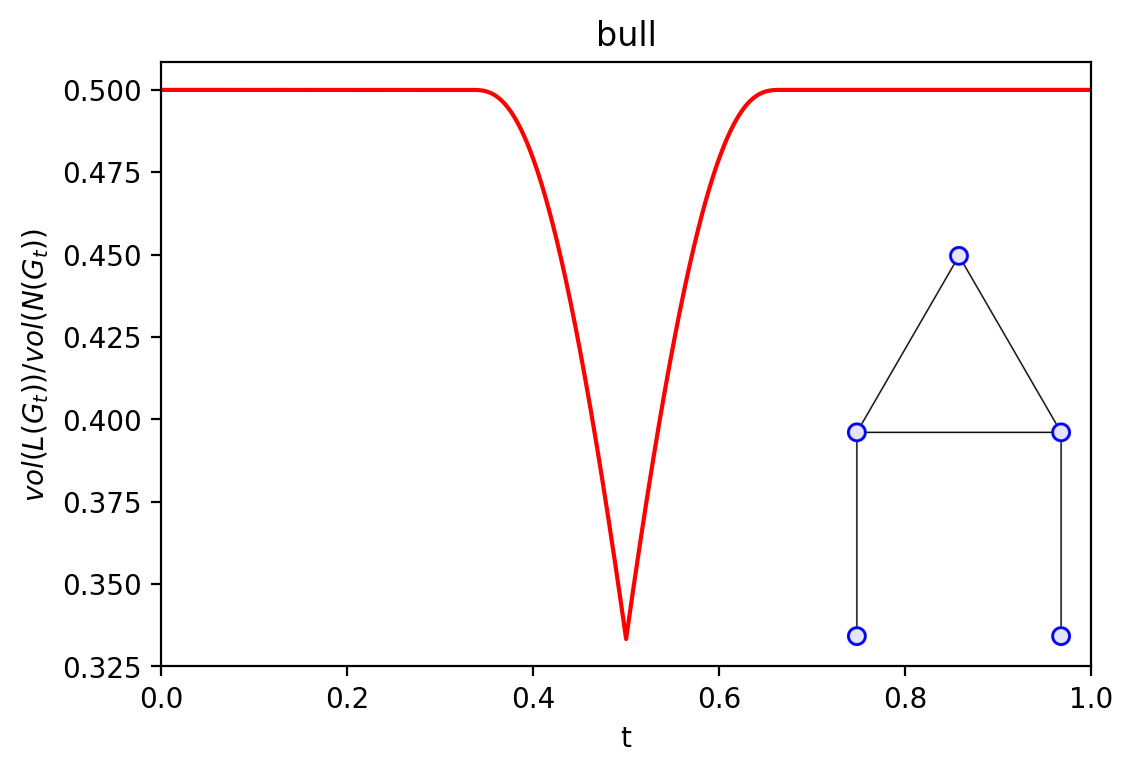}
\includegraphics[scale = 0.49]{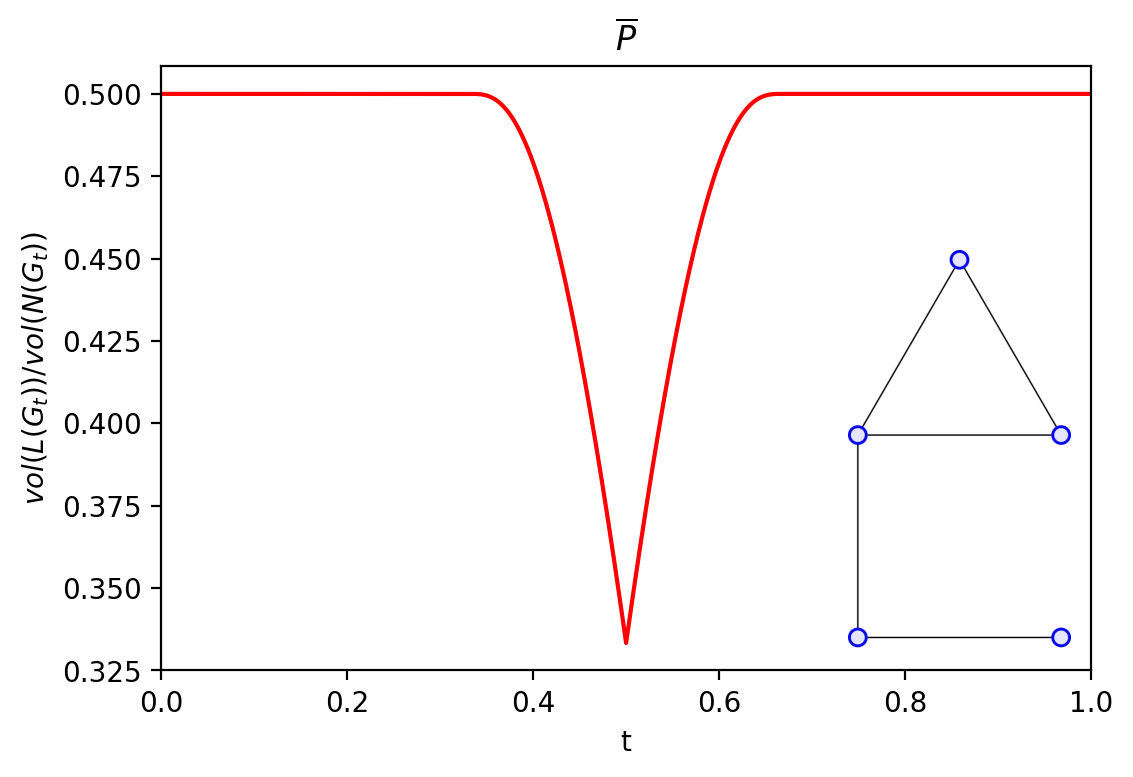}
\includegraphics[scale = 0.49]{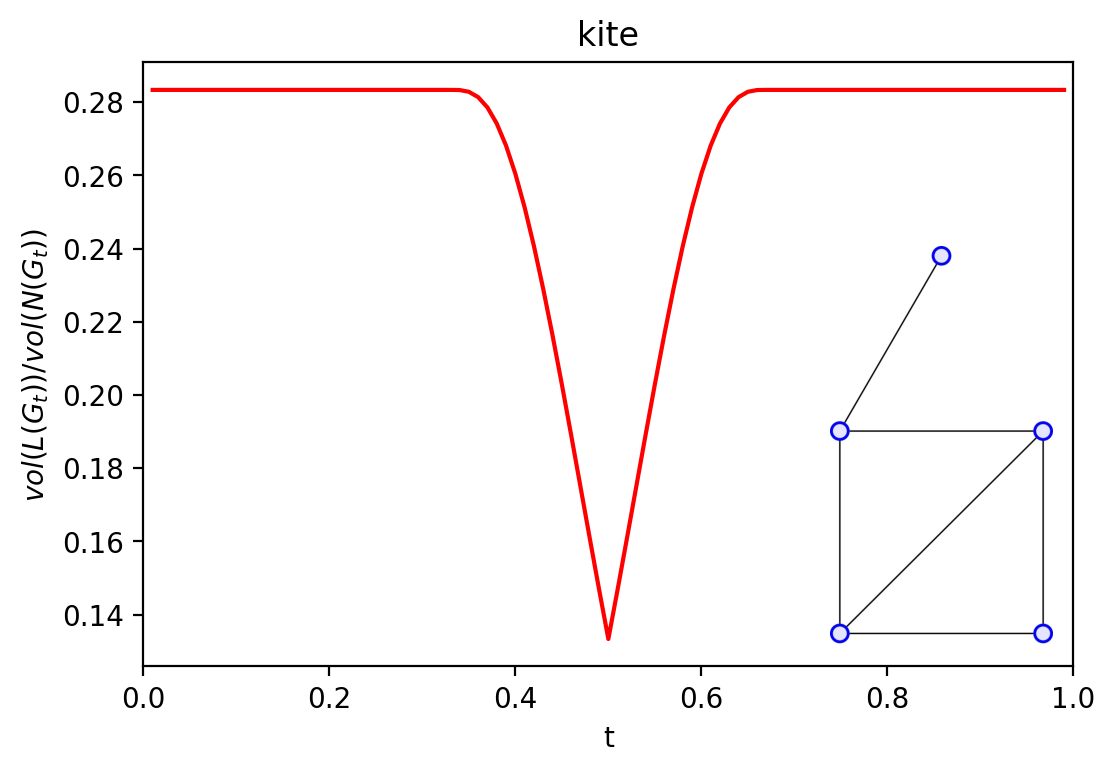}
\includegraphics[scale = 0.49]{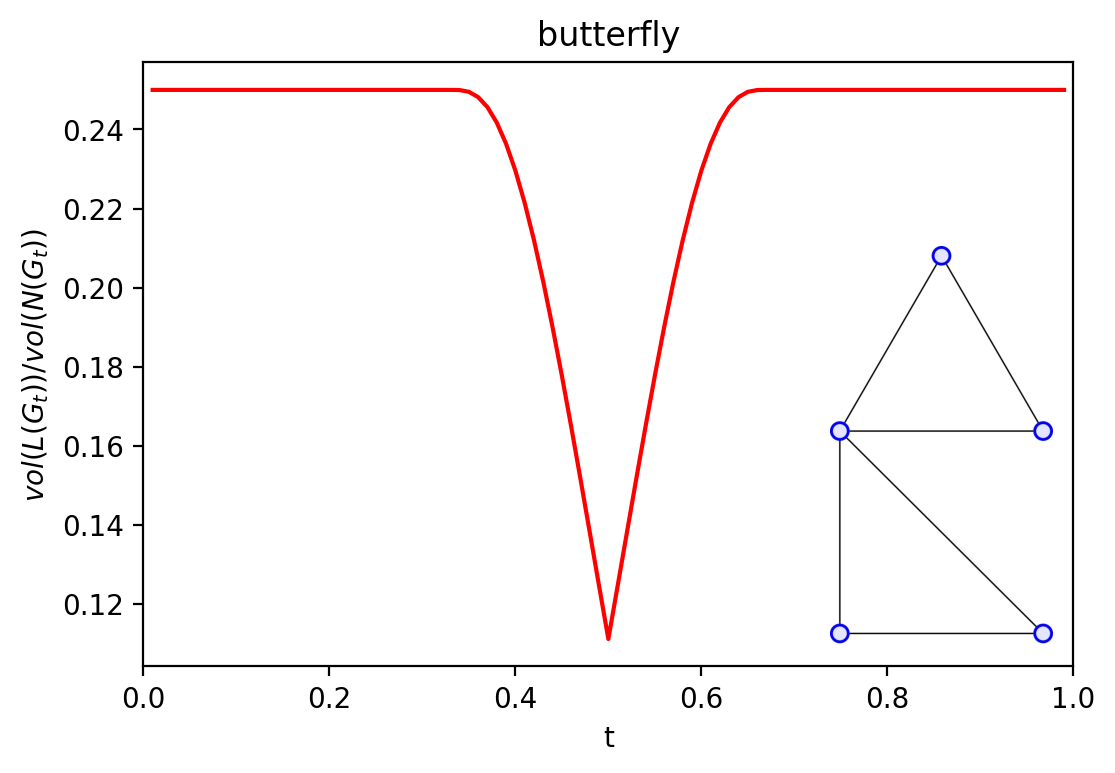}
\end{minipage}
\label{fig:5-vertices-1}
\end{figure}

\begin{figure}[htb]
\centering
\begin{minipage}[c]{.49\textwidth}
\includegraphics[scale = 0.49]{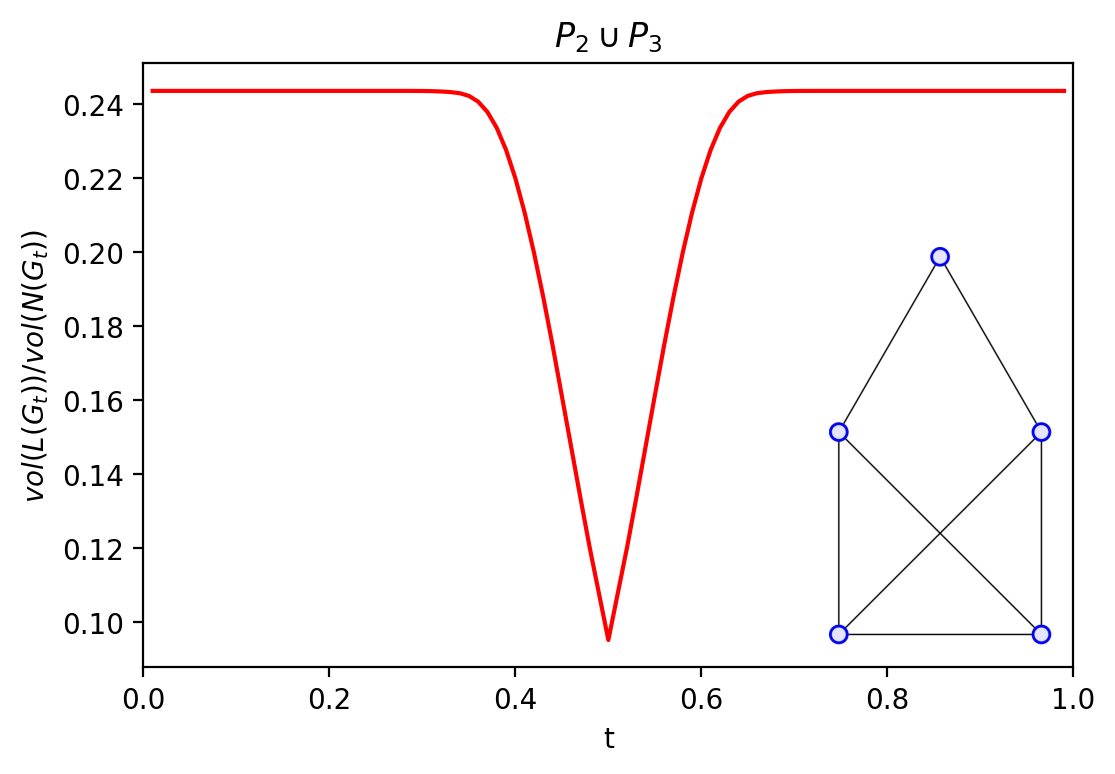}
\includegraphics[scale = 0.49]{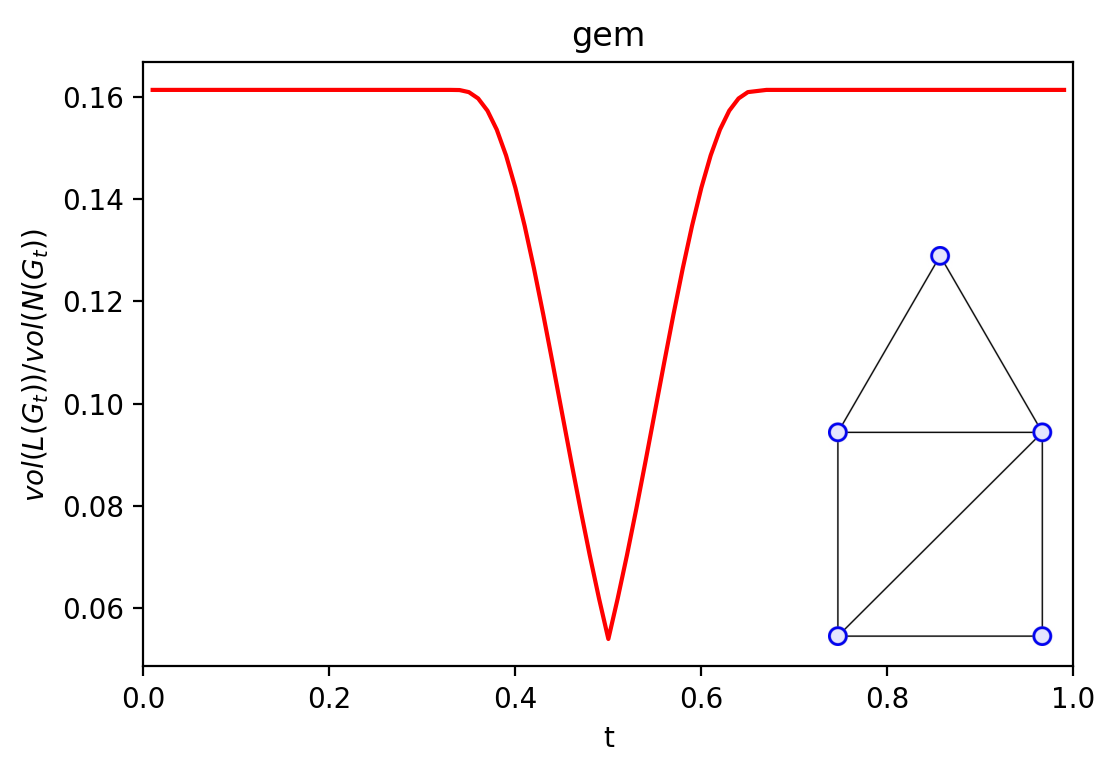}
\includegraphics[scale = 0.49]{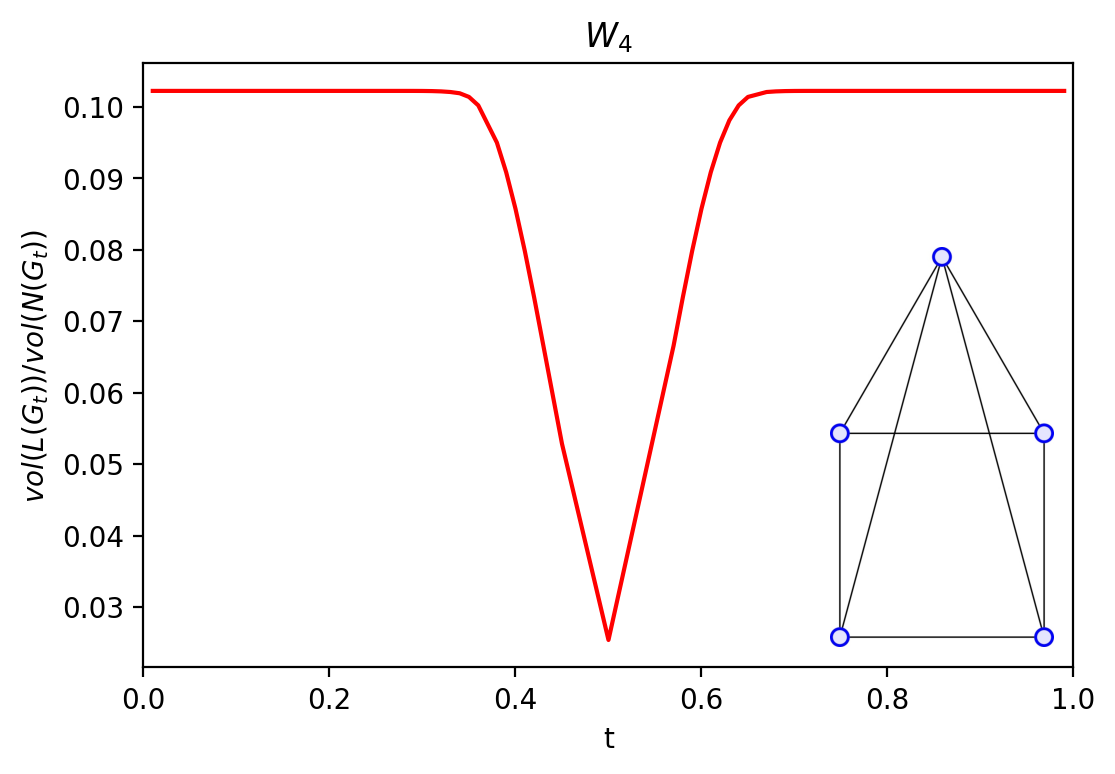}
\end{minipage}\hfill
\begin{minipage}[c]{.45\textwidth}
\includegraphics[scale = 0.49]{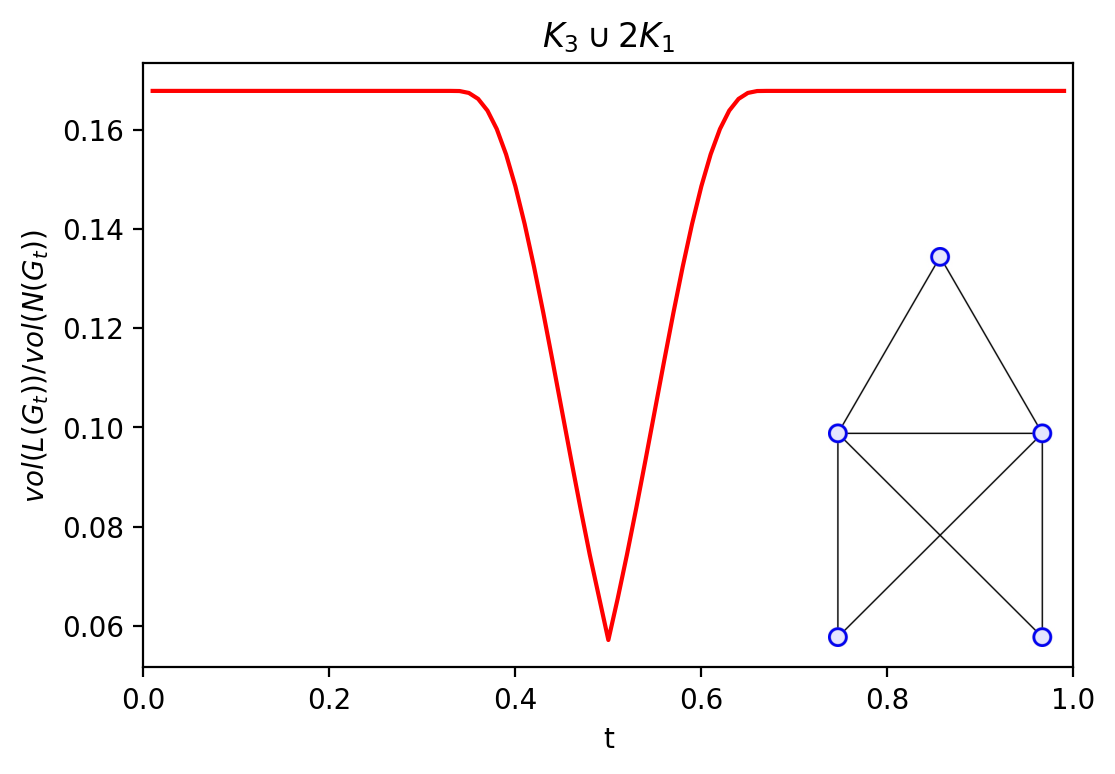}
\includegraphics[scale = 0.49]{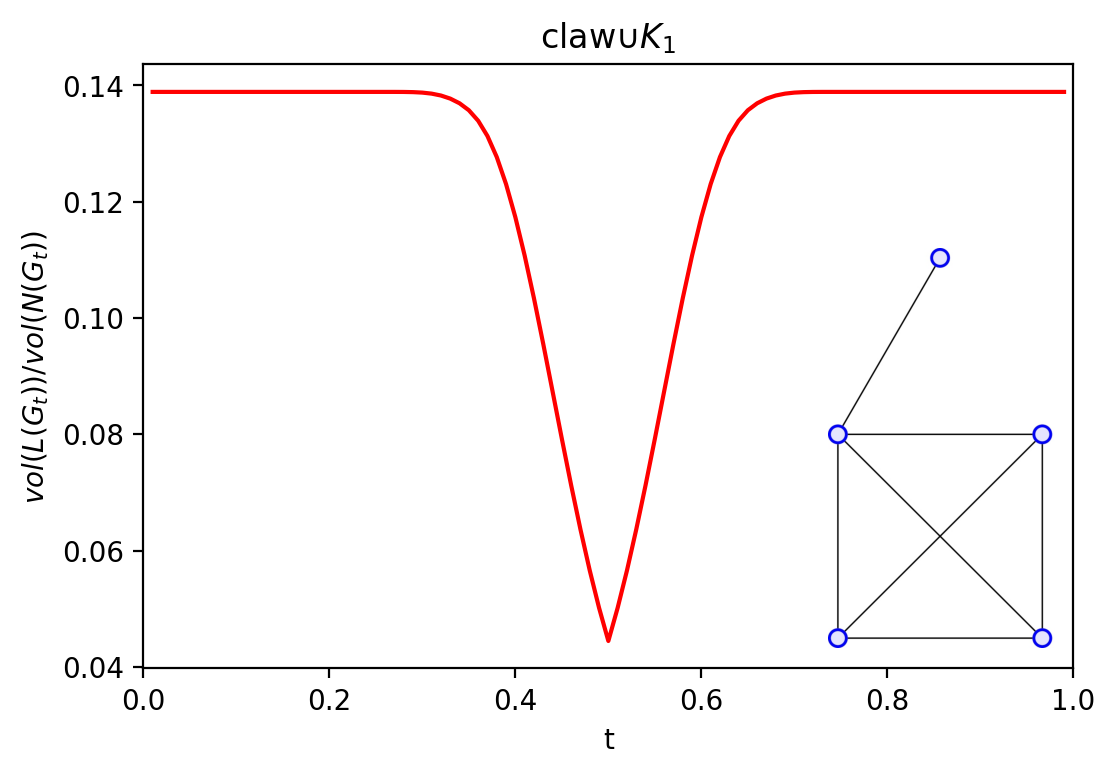}
\includegraphics[scale = 0.49]{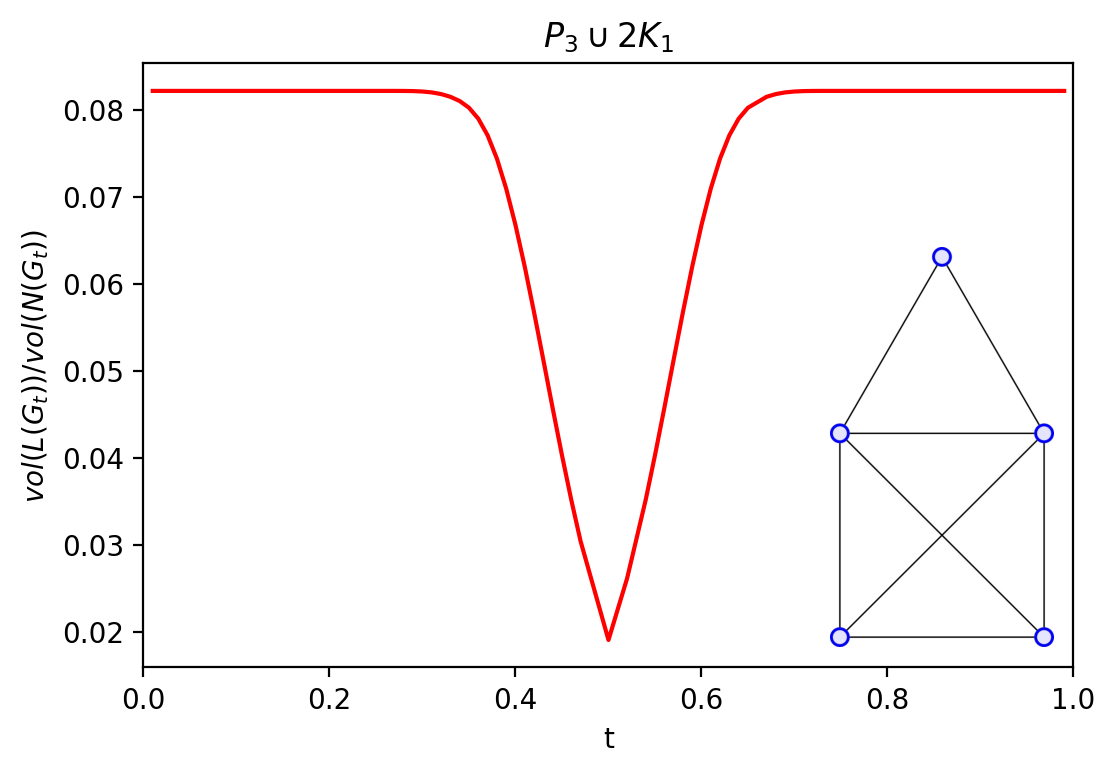}
\end{minipage}
\label{fig:5-vertices-2}
\caption{Plots of volume ratio of symmetric slices given by $\vol \mathsf L(p=(t,t,t,t,t), G)/\vol \mathsf N(p=(t,t,t,t,t), G)$ as a function of $t$ for different graphs with $5$ vertices. The graphs are displayed on the lower right corner of each plot. The plots for graphs $K_5-e$ and $K_5$ are excluded due to computational intractability.}
\end{figure}

\begin{table}[htb]
\begin{center}
\bgroup
\def\arraystretch{1.5}
\begin{tabular}{|c|c|c|c|c|}
\hline
\rowcolor[HTML]{C0C0C0} 
Graph $G$ & $\tw(G)$ & $\tau(G)$ & $\rho_{0+}(G)$ & $\rho_{1/2}(G)$\\ 
\hline
$C_5$ & 2 & $\frac{1}{3}$ & $\frac{23}{24}$ & $\frac{13}{15}$ \\
\hline
$4$-Pan & 2 & $\frac{1}{3}$ & $\frac{5}{6}$ & $\frac{2}{3}$\\
\hline
$K_{2, 3}$ & 2 & $\frac{1}{3}$ & $\frac{13}{20}$ & $\frac{2}{5}$\\
\hline
bull/cricket/$\bar P$ & 2 & $\frac{1}{3}$ & $\frac{1}{2}$ & $\frac{1}{3}$\\
\hline
house & 2 & $\frac{1}{3}$ & $\frac{7}{16}$ & $\frac{11}{45}$\\
\hline
kite/dart & 2 & $\frac{1}{3}$ & $\frac{17}{60}$ & $\frac{2}{15}$\\
\hline
butterfly & 2 & $\frac{1}{3}$ & $\frac{1}{4}$ & $\frac{1}{9}$\\
\hline
$P_2 \cup P_3$ & 3 & $\frac{1}{4}$ & $\frac{307}{1260}$ & $\frac{2}{21}$\\
\hline
$K_3 \cup 2K_1$ & 2 & $\frac{1}{3}$ & $\frac{47}{280}$ & $\frac{2}{35}$ \\
\hline
gem & 2 & $\frac{1}{3}$ & $\frac{271}{1680}$ & $\frac{17}{315}$\\
\hline
claw $\cup K_1$ & 3 & $\frac{1}{4}$ & $\frac{5}{36}$ & $\frac{2}{45}$\\
\hline
$W_4$ & 3 & $\frac{1}{4}$ & $\frac{229}{2240}$ & $\frac{8}{315}$\\
\hline
$P_3 \cup K_1$ & 3 & $\frac{1}{4}$ & $\frac{1657}{20160}$ & $\frac{2}{105}$\\
\hline
$K_5-e$ & 3 & $\frac{1}{4}$ & $\frac{1049}{24192}$ & $\frac{4}{567}$ \\
\hline
$K_5$ & 4 & $\frac{1}{5}$ & $\frac{14731}{725760}$ & $\frac{32}{14175}$\\
\hline
\end{tabular}
\egroup
\end{center}
\caption{List of parameters for graphs with 5 vertices: treewidth of the graph $\tw(G)$, the threshold $\tau(G)$ until which the volume ratio is constant, the initial (constant) volume ration $\rho_{0+}(G)$, the volume ration at $p=1/2$, $\rho_{1/2}(G)$. The volume computations are obtained using a Julia package \cite{kaluba2020polymake} for Polymake \cite{gawrilow2000polymake}.}
\label{tbl:5vertex}
\end{table}

Again we see from \cref{tbl:5vertex} that \cref{cnj:tw} holds for all graphs with 5 vertices.

\end{document}